\documentclass[12pt,oneside]{amsart}
\usepackage[latin9]{inputenc}
\usepackage{geometry}
\usepackage{mathrsfs}
\usepackage{mathtools}
\usepackage{amstext}
\usepackage{amsthm}
\usepackage{amssymb}
\usepackage{mathdots}
\usepackage{graphicx}
\usepackage{esint}
\usepackage[unicode=true,pdfusetitle,
 bookmarks=true,bookmarksnumbered=true,bookmarksopen=true,bookmarksopenlevel=1,
 breaklinks=false,pdfborder={0 0 1},backref=section,colorlinks=false]
 {hyperref}
\hypersetup{
 hypertex,bookmarksnumbered=true}

\makeatletter

\DeclareTextSymbolDefault{\textquotedbl}{T1}
\providecommand{\tabularnewline}{\\}

\numberwithin{equation}{section}
\numberwithin{figure}{section}
\theoremstyle{plain}
\newtheorem{thm}{\protect\theoremname}
\theoremstyle{plain}
\newtheorem{lem}[thm]{\protect\lemmaname}
\theoremstyle{remark}
\newtheorem{rem}[thm]{\protect\remarkname}
\theoremstyle{definition}
\newtheorem{defn}[thm]{\protect\definitionname}
\theoremstyle{plain}
\newtheorem{prop}[thm]{\protect\propositionname}
\theoremstyle{plain}
\newtheorem{cor}[thm]{\protect\corollaryname}

\makeatother

\providecommand{\corollaryname}{Corollary}
\providecommand{\definitionname}{Definition}
\providecommand{\lemmaname}{Lemma}
\providecommand{\propositionname}{Proposition}
\providecommand{\remarkname}{Remark}
\providecommand{\theoremname}{Theorem}

\begin{document}
\title{Synthesis of lossless electric circuits based on prescribed Jordan
forms}
\author{Alexander Figotin}
\address{University of California at Irvine, CA 92967}
\begin{abstract}
We advance here an algorithm of the synthesis of lossless electric
circuits such that their evolution matrices have the prescribed Jordan
canonical forms subject to natural constraints. Every synthesized
circuit consists of a chain-like sequence of $LC$-loops coupled by
gyrators. All involved capacitances, inductances and gyrator resistances
are either positive or negative with values determined by explicit
formulas. A circuit must have at least one negative capacitance or
inductance for having a nontrivial Jordan block for the relevant matrix.
\end{abstract}

\maketitle

\section{Introduction\label{sec:intro}}

This work is motivated by an interest to electromagnetic and optical
systems exhibiting Jordan eigenvector degeneracy, which is a degeneracy
of the system evolution matrix when not only some eigenvalues coincide
but the corresponding eigenvectors coincide also. Another way to describe
the eigenvector degeneracy of a matrix is by acknowledging that there
is no a basis in the relevant vector space made of eigenvectors of
the matrix. Such degenerate system states are quite often are referred
to as \emph{exceptional points of degeneracy} (EPDs), \cite[II.1]{Kato}.
A particularly important class of applications of EPDs is sensing,
\cite{CheN}. \cite{PeLiXu}, \cite{Wie}, \cite{Wie1}. Other potential
applications include (i) enhancement of the gain in active systems,
\cite{VOFC}, \cite{MLSPL}, \cite{OVFC}, \cite{OVFC1}, \cite{OTC},
and (ii) directivity of antennas, \cite{OthCap}. A variety of systems
have been suggested that exhibit EPDs in space for waveguide structures
and time for circuits. These systems are based on: (i) non-Hermitian
parity-time (PT) symmetric coupled systems, which are systems with
balanced loss and gain, \cite{BenBoe}, \cite{RKEC}, \cite{OGC};
(ii) coupled resonators \cite{SLZEK}, \cite{HMHCK}, \cite{HHWGECK},
(iii) electronic circuits involving dissipation, \cite{SteHeSch}. 

Systems with EPDs in cited above literature commonly involve loss
and gain elements suggesting that they might be essential to the existence
of EPDs, see for instance \cite{Berr}. It turns out though that the
presence of loss and gain elements in a system is not necessary for
having EPD regimes. An interesting system without loss and gain elements
has been proposed in \cite{KNAC} where the authors demonstrate that
EPDs can exist for a single $LC$ resonator with time-periodic modulation.
Our own studies in \cite{FigTWTbk} show that an analytical model
of traveling wave tube (TWT) has the Jordan eigenvector degeneracy
at some points of the system dispersion relation. This TWT system
is governed by a Lagrangian and consequently it is a perfectly conservative
system. Inspired by those studies we raised a question if simple lossless
(perfectly conservative) circuits exist such that their evolution
matrices exhibit the Jordan eigenvector degeneracy. We answered to
the question positively by constructing circuits with prescribed degeneracies.

Our primary goal here is to synthesize a lossless electric circuit
so that its evolution matrix $\mathscr{H}$ has a prescribed Jordan
canonical form $\mathscr{J}$ subject to natural constraints considered
later on. Hence by the definition of the Jordan canonical form, $\mathscr{H}=S\mathscr{J}S^{-1}$
where $S$ is an invertible matrix and $\mathscr{J}$ is a block diagonal
matrix of the form

\begin{gather}
\mathscr{J}=\left[\begin{array}{cccc}
J_{n_{1}}\left(\zeta_{1}\right) & \cdots & 0 & 0\\
0 & \ddots & \cdots & \vdots\\
\vdots & \ddots & J_{n_{q-1}}\left(\zeta_{q-1}\right) & 0\\
0 & \cdots & 0 & J_{n_{q}}\left(\zeta_{q}\right)
\end{array}\right],\quad J_{n}\left(\zeta\right)=\left[\begin{array}{ccccc}
\zeta & 1 & \cdots & 0 & 0\\
0 & \zeta & 1 & \cdots & 0\\
0 & 0 & \ddots & \cdots & \vdots\\
\vdots & \vdots & \ddots & \zeta & 1\\
0 & 0 & \cdots & 0 & \zeta
\end{array}\right].\label{eq:Joblo1b}
\end{gather}
$\zeta_{j}$ are real or complex numbers, and $J_{n}\left(\zeta\right)$
is the so-called Jordan block which is $n\times n$ matrix. For $n=1$
the matrix $J_{n}\left(\zeta\right)=\zeta$ turns just into number
$\zeta$.

As to the evolution matrix $\mathscr{H}$ we assume that the circuit
evolution is governed by the following linear equation
\begin{equation}
\partial_{t}X=\mathscr{H}X\label{eq:XHX1a}
\end{equation}
where $X$ is $2n$ dimensional vector-column describing the circuit
state and $\mathscr{H}$ is $2n\times2n$ matrix where $n>1$ is an
integer. The particular choice of the dimensions is explained by our
desire to have an underlying Lagrangian and Hamiltonian structure
so that equation (\ref{eq:XHX1a}) will be the Hamilton evolution
equation. Consequently, $2n\times2n$ matrix $\mathscr{H}$ is going
to be a Hamiltonian matrix and we will refer to it as the \emph{circuit
evolution matrix} or just \emph{circuit matrix}, see Section \ref{sec:LagHam}.
To meet the dimension requirements of the evolution equation (\ref{eq:XHX1a})
the circuit topological structure is expected to have $n$ \emph{fundamental
loops} or f-loops for short, see Section \ref{sec:e-net}. The circuit
state then is described by the corresponding $n$ time dependent charges
$q_{k}\left(t\right)$ which are the time integrals of the relevant
loop currents $\partial_{t}q_{k}\left(t\right)$. Hamiltonian formulations
of the dynamics of $LC$ circuits has been studied, see for instance
\cite{Masc} and references therein.

The eigenvalue problem associated with the evolution evolution (\ref{eq:XHX1a})
is
\begin{equation}
\mathscr{H}X=sX,\quad s=i\omega,\label{eq:XHX1aa}
\end{equation}
where $\omega$ is the frequency. Notice that the eigenvalue (spectral
parameter) $s$ is pure imaginary for real frequencies.

As to the prescribed Jordan canonical form $\mathscr{J}$ we are rather
interested in the simplest possible systems exhibiting nontrivial
Jordan blocks than systems that can have arbitrary Jordan canonical
form allowed for Hamiltonian matrices. It turns out that if the Jordan
canonical form $\mathscr{J}$ of the circuit matrix $\mathscr{H}$
has a nontrivial Jordan block then the circuit must have at least
one negative capacitance or inductance, see Section \ref{subsec:pos-Ham}.
The Jordan forms associated with Hamiltonian matrices must satisfy
certain constraints considered in Section \ref{sec:synth-elem}. The
origin of the constraints is the fundamental property of a Hamiltonian
matrix $\mathscr{H}$ to be similar to $-\mathscr{H}^{\mathrm{T}}$
which is the transposed to $\mathscr{H}$ matrix. This special property
of a Hamiltonian matrix combined with the general statement that every
square matrix $M$ is similar to the transposed to it matrix $M^{\mathrm{T}}$
impose the following constraints on the spectral structure of matrix
$\mathscr{H}$: (i) if $s$ in an eigenvalue of $\mathscr{H}$ then
$-s$ is its eigenvalue as well; (ii) the Jordan blocks corresponding
to the eigenvalues $s$ and $-s$ have the same structure. If in addition
to that the entries of the Hamiltonian matrix $\mathscr{H}$ are real-valued
then the following properties hold: (i) if $s$ in an eigenvalue of
$\mathscr{H}$ then$-s$, $\bar{s}$ and $-\bar{s}$, where $\bar{s}$
is complex-conjugate to $s$, are its eigenvalues as well; (ii) the
Jordan blocks corresponding to $s$,$-s$, $\bar{s}$ and $-\bar{s}$
have the same structure. We refer to the listed properties as \emph{Hamiltonian
spectral symmetry}, see Sections \ref{sec:synth-elem}, \ref{subsec:ham-mat}.
Apart from the Hamiltonian spectral symmetry the Jordan structure
of Hamiltonian matrices can be arbitrary, \cite[2.2]{ArnGiv}. Our
approach to the generation of Hamiltonian and the corresponding Hamiltonian
matrices is intimately related to the Hamiltonian canonical forms,
see Section \ref{sec:can-forms} and references therein.

Another significant mathematical input to the synthesis of the simplest
possible systems exhibiting nontrivial Jordan blocks comes from the
property of a square matrix $M$ to be \emph{cyclic} (also called
\emph{non-derogatory}), see Section \ref{sec:co-mat} and references
therein. We remind that a square matrix $M$ is called cyclic (or
non-derogatory) if the geometric multiplicity of each of its eigenvalues
is exactly $1$, or in other words, if every eigenvalue of $M$ has
exactly one eigenvector. Consequently, if a square matrix $M$ is
cyclic its Jordan form $J_{M}$ is completely determined by its \emph{characteristic
polynomia}l $\chi\left(s\right)=\det\left\{ s\mathbb{I}-M\right\} $
where $\mathbb{I}$ is the identity matrix of the relevant dimension.
Namely, every eigenvalue $s_{0}$ of $M$ of multiplicity $m$ is
associated with the single Jordan block $J_{m}\left(s_{0}\right)$
in the Jordan form $J_{M}$ of $M$. Consequently, \emph{for a cyclic
matrix $M$ its characteristic polynomial $\chi\left(s\right)=\det\left\{ s\mathbb{I}-M\right\} $
encodes all the information about its Jordan form }$J_{M}$. Another
property of any cyclic matrix $M$ associated with the a monic polynomial
$\chi$ is that it is similar to the so-called companion matrix $C_{\chi}$
defined by simple explicit expression involving the coefficients of
the polynomial $\chi$, see Section \ref{sec:co-mat} and references
therein. Companion matrix $C_{\chi}$ is naturally related to the
high-order differential equation $\chi\left(\partial_{t}\right)x\left(t\right)=0$
where $x\left(t\right)$ is a complex-valued function of $t$, see
Sections \ref{sec:co-mat} and \ref{sec:dif-jord}. This fact underlines
the relevance of the cyclicity property to the evolution of simpler
systems described by higher order differential equations for a scalar
function. In light of the above discussion, we focus on cyclic Hamiltonian
matrices $\mathscr{H}$ for they lead to the simplest circuits with
the evolution matrices $\mathscr{H}$ having nontrivial Jordan forms
$\mathscr{J}$.

Suppose the prescribed Jordan form $\mathscr{J}$ is an $2n$$\times2n$
matrix subject to the Hamiltonian spectral symmetry and the cyclicity
conditions. The synthesis of a circuit associated with $\mathscr{J}$
involves the following steps. We introduce first the characteristic
polynomial $\chi\left(s\right)=\det\left\{ s\mathbb{I}_{2n}-\mathscr{J}\right\} $
which is an even monic polynomial $\chi\left(s\right)$ of the degree
$2n$. We consider then the companion to $\chi\left(s\right)$ matrix
$\mathscr{C}$, see Section \ref{sec:co-mat}, which by the design
has $\mathscr{J}$ as its Jordan form, that is

\begin{equation}
\mathscr{C}=\mathscr{Y}\mathscr{J}\mathscr{Y}^{-1},\label{eq:XHX1b}
\end{equation}
where the columns of matrix $\mathscr{Y}$ form the so-called Jordan
basis of the companion matrix $\mathscr{C}$ associated with the characteristic
polynomial $\chi\left(s\right)=\det\left\{ s\mathbb{I}_{2n}-\mathscr{J}\right\} $,
see Section \ref{sec:co-mat}. We proceed with an introduction of
our principal Hamiltonian $\mathcal{H}$, see Section \ref{sec:synth-elem},
and recover from it $2n\times2n$ Hamiltonian matrix $\mathscr{H}$
that governs the system evolution according to equation (\ref{eq:XHX1a}).
As the result of our particular choice of the Hamiltonian $\mathcal{H}$
the corresponding to it Hamiltonian matrix $\mathscr{H}$ is similar
to the companion matrix $\mathscr{C}$ and consequently it has exactly
the same Jordan form $\mathscr{J}$ as $\mathscr{C}$. In particular,
we construct an $2n\times2n$ matrix $T$ such that
\begin{equation}
\mathscr{C}=T^{-1}\mathscr{H}T,\quad\mathscr{H}=\mathscr{Z}\mathscr{J}\mathscr{Z}^{-1},\quad\mathscr{Z}=T\mathscr{Y},\label{eq:XHX1c}
\end{equation}
where the columns of matrix $\mathscr{Z}$ form a Jordan basis of
the evolution matrix $\mathscr{H}$. The relations (\ref{eq:XHX1b})
and (\ref{eq:XHX1c}) between involved matrices are considered in
Section \ref{sec:priHam}. 

To relate the constructed Hamiltonian $\mathcal{H}$ to a circuit
we introduce the corresponding to it Lagrangian $\mathcal{L}$. Finally,
based on the Lagrangian $\mathcal{L}$ we design the relevant to it
circuit, see Section \ref{sec:pri-cir}. Consequently, this circuit
evolution is governed by equation (\ref{eq:XHX1a}) with cyclic Hamiltonian
matrix $\mathscr{H}$ that has the prescribed$\mathscr{J}$ as its
Jordan form. Each of the described steps of the circuit synthesis
and the quantities constructed in the process provide insights into
the circuit features.

In the light of our studies we can revisit now the question whether
the presence of the balanced loss and gain is essential for achieving
an electric circuit governed by the evolution matrix with nontrivial
Jordan form. We have succeeded in constructing lossless circuits associated
with nontrivial Jordan forms. Each of these circuits though must involve
at least one negative capacitance or inductance. If we take a look
at the physical implementations of negative capacitance and inductance
provided in Section \ref{subsec:neg-RCL} we find that they involve
matched positive and negative resistances. \emph{Based on this we
may conclude that (i) the presence of the balanced loss and gain is
essential for achieving negative values for the capacitance and the
inductance; (ii) the presence of at least one capacitor or inductor
of negative value of the capacitance or inductance respectively is
necessary for achieving a lossless electric circuit associated with
nontrivial Jordan form}.

The structure of the paper is as follows. In Section \ref{sec:pri-cir}
we show our principal circuit tailored to the desired Jordan form
$\mathscr{J}$ subject to natural constraints. In Section \ref{sec:spe-cir}
we introduce special circuits tailored to specially chosen characteristic
polynomials $\chi\left(s\right)$ and the corresponding Jordan forms
made of exactly two Jordan blocks of the size 2, 3 and 4. In Section
\ref{sec:synth-elem} we provide our strategy for the synthesis of
circuits associated wit the desired Jordan forms. Section \ref{sec:priHam}
is devoted to the analysis of our principal circuit Hamiltonian which
is the basis to the circuit synthesis. In Section \ref{sec:expriHam}
we consider examples of the principal circuit Hamiltonian and significant
matrices. Section \ref{sec:LagHam} provides aspects of the Lagrangian
and Hamiltonian formalisms as well important properties of Hamiltonian
matrices. In Section \ref{sec:e-net} we review basic elements of
the electric networks and their elements including gyrators and negative
capacitances and inductances. Sections \ref{sec:jord-form}-\ref{sec:can-forms}
are devoted to a number of mathematical subjects needed for our analysis.
In Section \ref{sec:notation} we provide the list of notations used
throughout the paper.

\section{Principal circuit\label{sec:pri-cir}}

Leaving the technical details of the circuit synthesis to the following
sections we present here our principal circuit design that implements
the desired Jordan form $\mathscr{J}$ of the circuit evolution matrix
$\mathscr{H}$. Quite remarkably the topology of circuits associated
with different Jordan forms is essentially the same. The difference
between the circuits is in: (i) the number of involved $LC$-loops;
(ii) particular values of the involved capacitances, inductances and
gyration resistances. Fig. \ref{fig:pri-cir} shows our \emph{principal
circuit} made of $n$ $LC$-loops coupled by gyrators. Quantities
$L_{j}$, $C_{j}$ and $G_{j}$ are respectively inductances, capacitances
and gyrator resistances.
\begin{figure}[h]
\centering{}\includegraphics[scale=0.6]{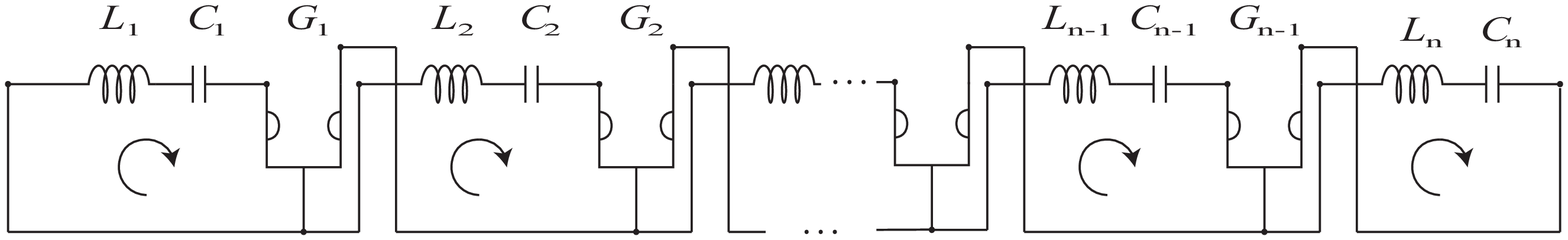}\caption{\label{fig:pri-cir} The principal circuit made of $n$ $LC$-loops.
Notice the difference between the left and the right connections for
the gyrators and $LC$-loops. It is explained by the non-reciprocity
of the gyrators and is designed to be consistent with (i) the standard
port assignment and selection of positive directions for the loop
currents and the gyrator; (ii) the sign of gyration resistance as
shown in Fig. \ref{fig:cir-gyr} and equations (\ref{eq:crvc1b}).
The values of quantities $L_{j}$, $C_{j}$ and $G_{j}$ are determined
by equations (\ref{eq:Lpricir1f})-(\ref{eq:Lpricir1g}) relating
them to the coefficients of the relevant polynomial.}
\end{figure}

To simplify equations throughout the paper we introduce the following
dimensionless version of some of the involved quantities
\begin{equation}
\tilde{t}=\omega_{0}t,\quad\breve{C}_{j}=\frac{C_{j}}{\left|C_{n}\right|},\quad\breve{L}_{j}=\omega_{0}^{2}\left|C_{n}\right|L_{j},\;\breve{G}_{j}=G_{j}\left|C_{n}\right|\omega_{0},\quad\mathcal{\breve{L}}=\left|C_{n}\right|\mathcal{L},\label{eq:pricirC1a}
\end{equation}
where $\omega_{0}>0$ is a unit of frequency and $1\leq j\leq n$
and $\mathcal{\breve{L}}$ is the scaled Lagrangian. To have less
cluttered formulas we actually omit ``hat'' from $\mathcal{\breve{L}}$,
$\tilde{t}$, $\breve{C}_{j}$, $\breve{L}_{j}$ and $\breve{G}_{j}$
and simply remember from now on that we use the relevant letters for
the dimensionless quantities and the scaled Lagrangian.

The \emph{principal circuit Lagrangian} associated with the principal
circuit depicted in Fig. \ref{fig:pri-cir} is

\begin{gather}
\mathcal{L}=\sum_{k=1}^{n}\frac{L_{k}\left(\partial_{t}q_{k}\right)^{2}}{2}-\sum_{k=1}^{n}\frac{\left(q_{k}\right)^{2}}{2C_{k}}+\sum_{k=1}^{n-1}\frac{G_{k}\left(q_{k}\partial_{t}q_{k+1}-q_{k+1}\partial_{t}q_{k}\right)}{2},\label{eq:Lpricir1a}\\
q_{k}=q_{k}\left(t\right)=\int i_{k}\,dt,\quad1\leq k\leq n,\nonumber 
\end{gather}
where $q_{k}$ and $i_{k}$ are respectively the charges and the currents
associated with $LC$-loops of the principal circuit depicted in Fig.
\ref{fig:pri-cir}. The corresponding \emph{Euler-Lagrange (EL) equations}
are

\begin{gather}
L_{1}\partial_{t}^{2}q_{1}-G_{1}\partial_{t}q_{2}+\frac{q_{1}}{C_{1}}=0,\quad L_{n}\partial_{t}^{2}q_{n}+G_{n-1}\partial_{t}q_{n-1}+\frac{q_{n}}{C_{n}}=0,\label{eq:Lpricir1b}
\end{gather}
\begin{gather}
L_{k}\partial_{t}^{2}q_{k}+G_{k-1}\partial_{t}q_{k-1}-G_{k}\partial_{t}q_{k+1}+\frac{q_{k}}{C_{k}}=0,\quad1<k<n.\label{eq:Lpricir1c}
\end{gather}

It is well known that the EL equations (\ref{eq:Lpricir1b}) and (\ref{eq:Lpricir1c})
represent the Kirchhoff voltage law for each of the $n$ f-loops,
see Section \ref{sec:e-net}. Indeed, each term in these equations
is associated with the voltage drop for the relevant electric element
as it can be verified by comparison with the voltage-current relations
reviewed in Section \ref{subsec:cir-elem}). As to the Kirchhoff current
law, one finds that it is already enforced by the selection of $n$
involved f-loops and currents $\partial_{t}q_{k}$ there. Indeed,
according to the gyrator settings the $k$-th gyrator has exactly
two incoming currents $\partial_{t}q_{k}$ and $\partial_{t}q_{k+1}$.
Indeed, the outgoing current, that passes through the gyrator branch
common to the $k$-th and the $\left(k+1\right)$-th f-loops, is equal
to the sum $\partial_{t}q_{k}+$$\partial_{t}q_{k+1}$. When exiting
this common branch the current $\partial_{t}q_{k}+$$\partial_{t}q_{k+1}$
splits into currents $\partial_{t}q_{k}$ and $\partial_{t}q_{k+1}$
in perfect compliance with the Kirchhoff current law. Notice that
each of the two equations in (\ref{eq:Lpricir1b}) corresponds to
the first and the last f-loops and involves only a single gyration
resistance. Each of the other f-loops has two adjacent f-loops, and,
consequently, the relevant to it equation in (\ref{eq:Lpricir1c})
involves two gyration resistances. Notice also the difference between
left and right connections between the gyrators and $LC$-loops. It
is explained by the non-reciprocity of the gyrators and is designed
to be consistent with (i) the standard port assignment and selection
of positive directions for the loop currents and the gyrator; (ii)
the sign of gyration resistance as shown in Fig. \ref{fig:cir-gyr},
see also equations (\ref{eq:crvc1b}) in Section \ref{sec:e-net}.

To make now a contact between the principal circuit, as depicted in
Fig. \ref{fig:pri-cir} and governed by the Lagrangian (\ref{eq:Lpricir1a}),
and the desired Jordan form $\mathscr{J}$ we introduce its characteristic
polynomial $\chi\left(s\right)$ that has to be of the form

\begin{gather}
\chi\left(s\right)=\det\left\{ s\mathbb{I}_{2n}-\mathscr{J}\right\} =s^{2n}+\left(-1\right)^{n}\sum_{k=1}^{n}a_{n-k}s^{2\left(n-k\right)}=\label{eq:Lpricir1d}\\
=s^{2n}+\left(-1\right)^{n}\left(a_{n-1}s^{2\left(n-1\right)}+a_{n-2}s^{2\left(n-2\right)}+\cdots a_{0}\right),\nonumber 
\end{gather}
where parameters $a_{k}$ are real-valued and satisfy
\begin{gather}
-\infty<a_{k}<\infty,\quad0\leq k\leq n-1;\;a_{0}\neq0.\label{eq:Lpricir1e}
\end{gather}

The Jordan form $\mathscr{J}$ has to satisfy some a priori symmetry
conditions to be associated with a Hamiltonian matrix $\mathscr{H}$.
In particular, its characteristic polynomial $\chi\left(s\right)$
has to be even polynomial as indicated by equations (\ref{eq:Lpricir1d})
and its parameters $a_{k}$ must be as described in relations (\ref{eq:Lpricir1e}).
The details on indicated properties of matrices $\mathscr{J}$ and
$\mathscr{H}$ are provided in Section \ref{sec:priHam}.

We relate then the principal circuit to the characteristic polynomial
$\chi\left(s\right)$ by the setting up the following expressions
for the circuit electric inductances, capacitances and gyration resistances
in terms of the coefficients $a_{k}$ of the polynomial $\chi\left(s\right)$:
\begin{equation}
L_{j}=\frac{1}{\left(-1\right)^{j-1}a_{j-1}},\quad1\leq j\leq n;\label{eq:Lpricir1f}
\end{equation}
\begin{gather}
C_{j}=\left(-1\right)^{j-1}a_{j},\quad G_{j}=\frac{1}{\left(-1\right)^{j-1}a_{j}}\quad1\leq j\leq n-1,\quad C_{n}=-1.\label{eq:Lpricir1g}
\end{gather}
Notice, that the equations (\ref{eq:Lpricir1f})-(\ref{eq:Lpricir1g})
imply the following identities
\begin{equation}
C_{j}L_{j+1}=-1,\quad C_{j}G_{j}=1,\quad1\leq j\leq n-1,\label{eq:Lpricir2a}
\end{equation}
as well as following expressions for coefficients $a_{j}$ in terms
of the circuit parameters
\begin{equation}
a_{j}=\frac{1}{\left(-1\right)^{j}L_{j+1}},\quad0\leq j\leq n-1,\quad a_{j}=\left(-1\right)^{j-1}\frac{1}{G_{j}}=\left(-1\right)^{j-1}C_{j},\quad1\leq j\leq n-1.\label{eq:Lpricir2b}
\end{equation}

Under assumptions that the circuit elements values satisfy equations
(\ref{eq:Lpricir1f})-(\ref{eq:Lpricir1g}) the Lagrangian $\mathcal{L}$
defined by equations (\ref{eq:Lpricir1a}) is related to our principal
Hamiltonian $\mathcal{H}$ defined by equations (\ref{eq:Hamapq1a}).
The relationship between $\mathcal{L}$ and $\mathcal{H}$ is as follows.
The Lagrangian $\mathcal{L}^{\prime}$ obtained from $\mathcal{H}$
by the Legendre transformation has exactly the same EL equations as
the Lagrangian $\mathcal{L}$, see Sections \ref{sec:priHam}, \ref{subsec:Lag}. 

We show in Section \ref{sec:dif-jord} that any solution to the EL
equations (\ref{eq:Lpricir1b}) and (\ref{eq:Lpricir1c}) satisfies
also the scalar differential equation
\begin{equation}
\chi\left(\partial_{t}\right)q_{k}\left(t\right)=0,\quad1\leq k\leq n,\label{eq:pricir2e}
\end{equation}
indicating that the circuit Hamiltonian matrix $\mathscr{H}$ is cyclic
and is determined by the characteristic polynomial $\chi\left(s\right)$
defined by equations (\ref{eq:Lpricir1d}).

\subsection{Principle circuit for two loops\label{subsec:pri-cir2}}

The principal circuit for two-loops is shown in Fig. \ref{fig:pri-cir2}.
It is the simplest case of our principal circuit that carries most
of the significant properties of the general case.
\begin{figure}[h]
\centering{}\includegraphics[scale=0.8]{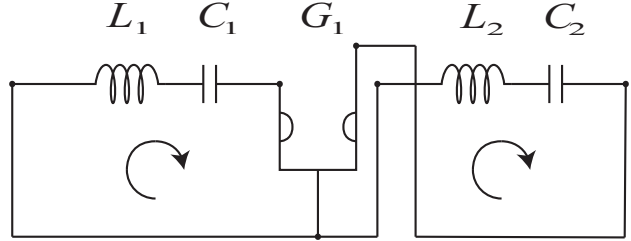}\caption{\label{fig:pri-cir2} Principle circuit for 2 $LC$-loops.}
\end{figure}
The the general form (\ref{eq:Lpricir1d}) of the characteristic polynomial
for $n=2$ turns into
\begin{equation}
\chi\left(s\right)=s^{4}+a_{1}s^{2}+a_{0}.\label{eq:pricir21a}
\end{equation}
The general form (\ref{eq:Lpricir1a}) of principal circuit Lagrangian
yields for $n=2$

\begin{equation}
\mathcal{L}=\frac{L_{1}\left(\partial_{t}q_{1}\right)^{2}}{2}+\frac{L_{2}\left(\partial_{t}q_{2}\right)^{2}}{2}-\frac{\left(q_{1}\right)^{2}}{2C_{1}}-\frac{\left(q_{2}\right)^{2}}{2C_{2}}+\frac{G_{1}\left(q_{1}\partial_{t}q_{2}-q_{2}\partial_{t}q_{1}\right)}{2},\label{eq:specir21b}
\end{equation}
and the corresponding EL equations are
\begin{gather}
L_{1}\partial_{t}^{2}q_{1}-G_{1}\partial_{t}q_{2}+\frac{q_{1}}{C_{1}}=0,\quad L_{2}\partial_{t}^{2}q_{2}+G_{1}\partial_{t}q_{1}+\frac{q_{2}}{C_{2}}=0.\label{eq:pricir21c}
\end{gather}
In particular, as the consequence of equations (\ref{eq:Lpricir2a})-(\ref{eq:Lpricir1c})
as well as the data in Table \ref{tab:pri-cir2} the following identities
hold
\begin{equation}
C_{1}L_{2}=-1,\quad C_{1}G_{1}=1;\quad a_{0}=\frac{1}{L_{1}},\quad a_{1}=-\frac{1}{L_{2}}=\frac{1}{G_{1}}=C_{1}.\label{eq:pricir21e}
\end{equation}

The set of values of the principal circuit elements described by equations
(\ref{eq:Lpricir1f})-(\ref{eq:Lpricir1g}) for $n=2$ are listed
in Table \ref{tab:pri-cir2}. 
\begin{table}[h]
\centering{}\caption{Circuit elements values, $n=2$}
\label{tab:pri-cir2}%
\begin{tabular}{|c||c||c|}
\hline 
$k$ & $1$ & $2$\tabularnewline[0.1cm]
\hline 
\hline 
$L_{k}$ & $\frac{1}{a_{0}}$ & $-\frac{1}{a_{1}}$\tabularnewline[0.1cm]
\hline 
\hline 
$C_{k}$ & $a_{1}$ & $-1$\tabularnewline[0.1cm]
\hline 
\hline 
$G_{k}$ & $\frac{1}{a_{1}}$ & \tabularnewline[0.1cm]
\hline 
\end{tabular}.
\end{table}
The significant circuit matrices in this case are as follows:
\begin{equation}
\mathscr{H}=\left[\begin{array}{rrrr}
0 & 0 & a_{0} & 0\\
1 & 0 & 0 & -a_{1}\\
0 & 0 & 0 & -1\\
0 & 1 & 0 & 0
\end{array}\right],\quad\mathscr{C}=\left[\begin{array}{rrrr}
0 & 1 & 0 & 0\\
0 & 0 & 1 & 0\\
0 & 0 & 0 & 1\\
-a_{0} & 0 & -a_{1} & 0
\end{array}\right],\quad T=\left[\begin{array}{rrrr}
0 & -a_{0} & 0 & 0\\
-a_{0} & 0 & -a_{1} & 0\\
0 & 0 & -1 & 0\\
0 & 0 & 0 & 1
\end{array}\right].\label{eq:pricir22a}
\end{equation}
The determinants of the above matrices are as follows
\begin{equation}
\det\left\{ \mathscr{H}\right\} =\det\left\{ \mathscr{C}\right\} =a_{0},\quad\det\left\{ T\right\} =a_{0}^{2}.\label{eq:pricir22d}
\end{equation}
The similarity between matrices $\mathscr{H}$ and $\mathscr{C}$
takes here the form
\begin{equation}
T^{-1}\mathscr{H}T=\mathscr{C},\label{eq:pricir22e}
\end{equation}
and can be verified by showing that $\mathscr{H}T=\mathscr{C}T^{-1}$
based on expressions (\ref{eq:pricir22a}) for the involved matrices.

Solutions to the EL equations (\ref{eq:pricir21c}) according to equations
(\ref{eq:Lpricir1e}) satisfy the following scalar differential equation
\begin{equation}
\left(\partial_{t}^{4}+a_{1}\partial_{t}^{2}+a_{0}\right)q_{k}\left(t\right)=0,\quad k=1,2.\label{eq:pricir22f}
\end{equation}

\section{Special circuits\label{sec:spe-cir}}

We define special circuits as implementations of our principal circuit
tailored to specially chosen characteristic polynomials $\chi\left(s\right)$
to achieve the desired Jordan forms. Namely, we are interested in
\begin{gather}
\chi\left(s\right)=\left(s^{2}-a^{2}\right)^{n},\quad\left(s^{2}+b^{2}\right)^{n},\quad n=2,3,4,\label{eq:specirk1a}
\end{gather}
for real $a$ and $b$ corresponding to the Jordan form $\mathscr{J}$
made of two Jordan blocks $J_{n}\left(\pm a\right)$ and $J_{n}\left(\pm b\mathrm{i}\right)$
respectively, where $J_{n}\left(s\right)$ is $n\times n$ matrix
defined by equation (\ref{eq:Joblo1b}). In Section \ref{subsec:spe-cirn}
we consider special circuits associated with the characteristic polynomials
$\chi\left(s\right)$ defined by equation (\ref{eq:specirk1a}) for
arbitrary $n\geq2$.

We are also interested in
\begin{gather}
\chi\left(s\right)=\left[\left(s-\zeta\right)\left(s-\bar{\zeta}\right)\left(s+\zeta\right)\left(s+\bar{\zeta}\right)\right]^{2},\quad\zeta=a+b\mathrm{i},\label{eq:specirk1c}
\end{gather}
where $a\neq0$ and $b\neq0$ corresponding to the Jordan form $\mathscr{J}$
made of 4 Jordan blocks $J_{2}\left(\pm a\pm b\mathrm{i}\right)$
and $J_{2}\left(\pm a\mp b\mathrm{i}\right)$.

When considering special circuits we evaluate the significant matrices
$\mathscr{H}$, $\mathscr{J}$, $\mathscr{C}$, $\mathscr{Y}$, $\mathscr{Z}$
and $T$ related by equations (\ref{eq:XHX1b}) and (\ref{eq:XHX1c}),
see Sections \ref{sec:intro}. \ref{sec:priHam}.

\subsection{Special circuits for two real or pure imaginary eigenvalues and $n=2$\label{subsec:spe-cir2}}

This is the simplest case demonstrating nontrivial Jordan forms and
for that reason we study it greater detail. The special circuit shown
in Fig. \ref{fig:pri-cir2} has $2$ f-loops. To get the desired Jordan
form we use one of the polynomials
\begin{equation}
\chi\left(s\right)=\left(s^{2}-a^{2}\right)^{2},\quad\left(s^{2}+b^{2}\right)^{2},\label{eq:specir2c}
\end{equation}
and assign to the circuit elements the values provided in Table \ref{tab:spe-cir2}
(parameter $r$ refers to the roots of the characteristic polynomial
$\chi\left(s\right)$, that is $\chi\left(r\right)=0$).
\begin{table}[h]
\centering{}\caption{Circuit elements values, $n=2$}
\label{tab:spe-cir2}%
\begin{tabular}{|c||c||c|}
\hline 
$k$ & $1$ & $2$\tabularnewline[0.1cm]
\hline 
\hline 
$L_{k}$ & $\frac{1}{a^{4}}$ & $\frac{1}{2\,a^{2}}$\tabularnewline[0.1cm]
\hline 
\hline 
$C_{k}$ & $-2\,a^{2}$ & $-1$\tabularnewline[0.1cm]
\hline 
\hline 
$G_{k}$ & $-\frac{1}{2\,a^{2}}$ & \tabularnewline[0.1cm]
\hline 
$r;\;n$ & $\pm a;$$\;2$ & \tabularnewline[0.1cm]
\hline 
\end{tabular},$\quad$%
\begin{tabular}{|c||c||c|}
\hline 
$k$ & $1$ & $2$\tabularnewline[0.1cm]
\hline 
\hline 
$L_{k}$ & $\frac{1}{b^{4}}$ & $-\frac{1}{2\,b^{2}}$\tabularnewline[0.1cm]
\hline 
\hline 
$C_{k}$ & $2\,b^{2}$ & $-1$\tabularnewline[0.1cm]
\hline 
\hline 
$G_{k}$ & $\frac{1}{2\,b^{2}}$ & \tabularnewline[0.1cm]
\hline 
$r;\;n$ & $\pm b\mathrm{i};$$\;2$ & \tabularnewline[0.1cm]
\hline 
\end{tabular}.
\end{table}

In the case of the first polynomial $\chi\left(s\right)$ that has
real roots $\pm a$ the circuit significant matrices are as follows:
\begin{equation}
\mathscr{H}=\left[\begin{array}{rrrr}
0 & 0 & a^{4} & 0\\
1 & 0 & 0 & 2\,a^{2}\\
0 & 0 & 0 & -1\\
0 & 1 & 0 & 0
\end{array}\right],\quad\mathscr{J}=\left[\begin{array}{rrrr}
a & 1 & 0 & 0\\
0 & a & 0 & 0\\
0 & 0 & -a & 1\\
0 & 0 & 0 & -a
\end{array}\right],\quad T=\left[\begin{array}{rrrr}
0 & -a^{4} & 0 & 0\\
-a^{4} & 0 & 2\,a^{2} & 0\\
0 & 0 & -1 & 0\\
0 & 0 & 0 & 1
\end{array}\right],\label{eq:specir2dH}
\end{equation}
\begin{equation}
\mathscr{C}=\left[\begin{array}{rrrr}
0 & 1 & 0 & 0\\
0 & 0 & 1 & 0\\
0 & 0 & 0 & 1\\
-a^{4} & 0 & 2\,a^{2} & 0
\end{array}\right],\:\mathscr{Y}=\left[\begin{array}{rrrr}
1 & 0 & 1 & 0\\
a & 1 & -a & 1\\
a^{2} & 2\,a & a^{2} & -2\,a\\
a^{3} & 3\,a^{2} & -a^{3} & 3\,a^{2}
\end{array}\right],\:\mathscr{Z}=\left[\begin{array}{rrrr}
-a^{5} & -a^{4} & a^{5} & -a^{4}\\
a^{4} & 4\,a^{3} & a^{4} & -4\,a^{3}\\
-a^{2} & -2\,a & -a^{2} & 2\,a\\
a^{3} & 3\,a^{2} & -a^{3} & 3\,a^{2}
\end{array}\right].\label{eq:specir2dC}
\end{equation}
In the case of the second polynomial $\chi\left(s\right)$ that has
pure imaginary roots $\pm b\mathrm{i}$ the significant circuit matrices
are

\begin{equation}
\mathscr{H}=\left[\begin{array}{rrrr}
0 & 0 & b^{4} & 0\\
1 & 0 & 0 & -2\,b^{2}\\
0 & 0 & 0 & -1\\
0 & 1 & 0 & 0
\end{array}\right],\quad\mathscr{J}=\left[\begin{array}{rrrr}
b\mathrm{i} & 1 & 0 & 0\\
0 & b\mathrm{i} & 0 & 0\\
0 & 0 & -b\mathrm{i} & 1\\
0 & 0 & 0 & -b\mathrm{i}
\end{array}\right],\quad T=\left[\begin{array}{rrrr}
0 & -b^{4} & 0 & 0\\
b^{4} & 0 & -2\,b^{2} & 0\\
0 & 0 & -1 & 0\\
0 & 0 & 0 & 1
\end{array}\right],\label{eq:specir2eH}
\end{equation}
\begin{equation}
\mathscr{C}=\left[\begin{array}{rrrr}
0 & 1 & 0 & 0\\
0 & 0 & 1 & 0\\
0 & 0 & 0 & 1\\
-b^{4} & 0 & -2\,b^{2} & 0
\end{array}\right],\quad\mathscr{Y}=\left[\begin{array}{rrrr}
1 & 0 & 1 & 0\\
b\mathrm{i} & 1 & -b\mathrm{i} & 1\\
-b^{2} & 2\,b\mathrm{i} & -b^{2} & -2\,b\mathrm{i}\\
-b^{3}\mathrm{i} & -3\,b^{2} & b^{3}\mathrm{i} & -3\,b^{2}
\end{array}\right],\label{eq:specir2eC}
\end{equation}
\begin{equation}
\mathscr{Z}=\left[\begin{array}{rrrr}
-b^{5}\mathrm{i} & -b^{4} & b^{5}\mathrm{i} & -b^{4}\\
b^{4} & -4\,b^{3}\mathrm{i} & b^{4} & 4\,b^{3}\mathrm{i}\\
b^{2} & -2\,b\mathrm{i} & b^{2} & 2\,b\mathrm{i}\\
-b^{3}\mathrm{i} & -3\,b^{2} & b^{3}\mathrm{i} & -3\,b^{2}
\end{array}\right].\label{eq:specir2eZ}
\end{equation}
Applying to the case $n=2$ the general formulas (\ref{eq:specirn2h})
for the eigenfrequencies $\omega_{j}=\pm\frac{1}{\sqrt{L_{j}C_{j}}}$
of $LC$-loops as they were decoupled we obtain
\begin{equation}
\omega_{1}=\pm\frac{b}{\sqrt{2}},\quad\omega_{2}=\pm\sqrt{2}b.\label{eq:specir2d}
\end{equation}

\subsection{Special circuit for two real or pure imaginary eigenvalues and $n=3$}

This special circuit shown in Fig. \ref{fig:pri-cir3} has $3$ f-loops.
It provides for two Jordan blocks of order 3 for circuit elements
values as in Table \ref{tab:spe-cir3a} (parameter $r$ refers to
the roots of the characteristic polynomial $\chi\left(s\right)$,
that is $\chi\left(r\right)=0$).
\begin{figure}[h]
\centering{}\includegraphics[scale=0.8]{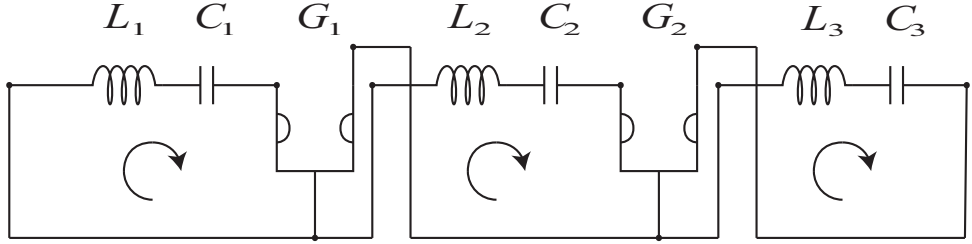}\caption{\label{fig:pri-cir3} Special circuit for 3 $LC$-loops.}
\end{figure}
\begin{table}[h]
\centering{}\caption{Circuit elements values, $n=3$}
\label{tab:spe-cir3a}%
\begin{tabular}{|c||c||c||c|}
\hline 
$k$ & $1$ & $2$ & $3$\tabularnewline[0.1cm]
\hline 
\hline 
$L_{k}$ & $\frac{1}{a^{6}}$ & $\frac{1}{3\,a^{4}}$ & $\frac{1}{3\,a^{2}}$\tabularnewline[0.1cm]
\hline 
\hline 
$C_{k}$ & $-3\,a^{4}$ & $-3\,a^{2}$ & $-1$\tabularnewline[0.1cm]
\hline 
\hline 
$G_{k}$ & $-\frac{1}{3\,a^{4}}$ & $-\frac{1}{3\,a^{2}}$ & \tabularnewline[0.1cm]
\hline 
\hline 
$r;\;n$ & $\pm a;$$\;3$ &  & \tabularnewline[0.1cm]
\hline 
\end{tabular},$\quad$%
\begin{tabular}{|c||c||c||c|}
\hline 
$k$ & $1$ & $2$ & $3$\tabularnewline[0.1cm]
\hline 
\hline 
$L_{k}$ & $-\frac{1}{b^{6}}$ & $\frac{1}{3\,b^{4}}$ & $-\frac{1}{3\,b^{2}}$\tabularnewline[0.1cm]
\hline 
\hline 
$C_{k}$ & $-3\,b^{4}$ & $3\,b^{2}$ & $-1$\tabularnewline[0.1cm]
\hline 
\hline 
$G_{k}$ & $-\frac{1}{3\,b^{4}}$ & $\frac{1}{3\,b^{2}}$ & \tabularnewline[0.1cm]
\hline 
\hline 
$r;\;n$ & $\pm b\mathrm{i};$$\;3$ &  & \tabularnewline[0.1cm]
\hline 
\end{tabular}.
\end{table}
The corresponding polynomial is
\begin{equation}
\chi\left(s\right)=\left(s^{2}-a^{2}\right)^{3},\quad\left(s^{2}+b^{2}\right)^{3},\label{eq:specir3a}
\end{equation}
In the case when the roots of $\chi\left(s\right)$ are real numbers
$\pm a$ the circuit matrices are
\begin{gather}
\mathscr{H}=\left[\begin{array}{rrrrrr}
0 & 0 & 0 & a^{6} & 0 & 0\\
1 & 0 & 0 & 0 & 3\,a^{4} & 0\\
0 & 1 & 0 & 0 & 0 & 3\,a^{2}\\
0 & 0 & 0 & 0 & -1 & 0\\
0 & 0 & 0 & 0 & 0 & -1\\
0 & 0 & 1 & 0 & 0 & 0
\end{array}\right],\quad\mathscr{J}=\left[\begin{array}{rrrrrr}
a & 1 & 0 & 0 & 0 & 0\\
0 & a & 1 & 0 & 0 & 0\\
0 & 0 & a & 0 & 0 & 0\\
0 & 0 & 0 & -a & 1 & 0\\
0 & 0 & 0 & 0 & -a & 1\\
0 & 0 & 0 & 0 & 0 & -a
\end{array}\right],\label{eq:specir3bH}
\end{gather}

\begin{gather}
\mathscr{Z}=\left[\begin{array}{rrrrrr}
a^{8} & 2\,a^{7} & a^{6} & a^{8} & -2\,a^{7} & a^{6}\\
-2\,a^{7} & -8\,a^{6} & -9\,a^{5} & 2\,a^{7} & -8\,a^{6} & 9\,a^{5}\\
a^{6} & 6\,a^{5} & 15\,a^{4} & a^{6} & -6\,a^{5} & 15\,a^{4}\\
a^{3} & 3\,a^{2} & 3\,a & -a^{3} & 3\,a^{2} & -3\,a\\
-a^{4} & -4\,a^{3} & -6\,a^{2} & -a^{4} & 4\,a^{3} & -6\,a^{2}\\
a^{5} & 5\,a^{4} & 10\,a^{3} & -a^{5} & 5\,a^{4} & -10\,a^{3}
\end{array}\right].\label{eq:specir3bZ}
\end{gather}

In the case when the roots of $\chi\left(s\right)$ are pure imaginary
numbers $\pm b\mathrm{i}$ the circuit matrices are

\begin{equation}
\mathscr{H}=\left[\begin{array}{rrrrrr}
0 & 0 & 0 & -b^{6} & 0 & 0\\
1 & 0 & 0 & 0 & 3\,b^{4} & 0\\
0 & 1 & 0 & 0 & 0 & -3\,b^{2}\\
0 & 0 & 0 & 0 & -1 & 0\\
0 & 0 & 0 & 0 & 0 & -1\\
0 & 0 & 1 & 0 & 0 & 0
\end{array}\right],\quad\mathscr{J}=\left[\begin{array}{rrrrrr}
b\mathrm{i} & 1 & 0 & 0 & 0 & 0\\
0 & b\mathrm{i} & 1 & 0 & 0 & 0\\
0 & 0 & b\mathrm{i} & 0 & 0 & 0\\
0 & 0 & 0 & -b\mathrm{i} & 1 & 0\\
0 & 0 & 0 & 0 & -b\mathrm{i} & 1\\
0 & 0 & 0 & 0 & 0 & -b\mathrm{i}
\end{array}\right],\label{eq:specir3cH}
\end{equation}

\begin{gather}
\mathscr{Z}=\left[\begin{array}{rrrrrr}
b^{8} & -2\,b^{7}\mathrm{i} & -b^{6} & b^{8} & 2\,b^{7}\mathrm{i} & -b^{6}\\
2\,b^{7}\mathrm{i} & 8\,b^{6} & -9\,b^{5}\mathrm{i} & -2\,b^{7}\mathrm{i} & 8\,b^{6} & 9\,b^{5}\mathrm{i}\\
-b^{6} & 6\,b^{5}\mathrm{i} & 15\,b^{4} & -b^{6} & -6\,b^{5}\mathrm{i} & 15\,b^{4}\\
-b^{3}\mathrm{i} & -3\,b^{2} & 3\,b\mathrm{i} & b^{3}\mathrm{i} & -3\,b^{2} & -3\,b\mathrm{i}\\
-b^{4} & 4\,b^{3}\mathrm{i} & 6\,b^{2} & -b^{4} & -4\,b^{3}\mathrm{i} & 6\,b^{2}\\
b^{5}\mathrm{i} & 5\,b^{4} & -10\,b^{3}\mathrm{i} & -b^{5}\mathrm{i} & 5\,b^{4} & 10\,b^{3}\mathrm{i}
\end{array}\right].\label{eq:specir3cJZ}
\end{gather}
Applying to the case $n=3$ the general formulas (\ref{eq:specirn2h})
for the eigenfrequencies $\omega_{j}=\pm\frac{1}{\sqrt{L_{j}C_{j}}}$
of $LC$-loops as they were decoupled we obtain
\begin{equation}
\omega_{1}=\pm\frac{b}{\sqrt{3}},\quad\omega_{2}=\pm b,\quad\omega_{3}=\pm\sqrt{3}b.\label{eq:specir3d}
\end{equation}

\subsection{Special circuit for two real or pure imaginary eigenvalues and $n=4$}

This special circuit shown in Fig. has $4$ f-loops. It provides for
two Jordan blocks of order 4 for circuit elements values as in Tables
\ref{tab:spe-cir4a} and \ref{tab:spe-cir4b}.
\begin{figure}[h]
\centering{}\includegraphics[scale=0.8]{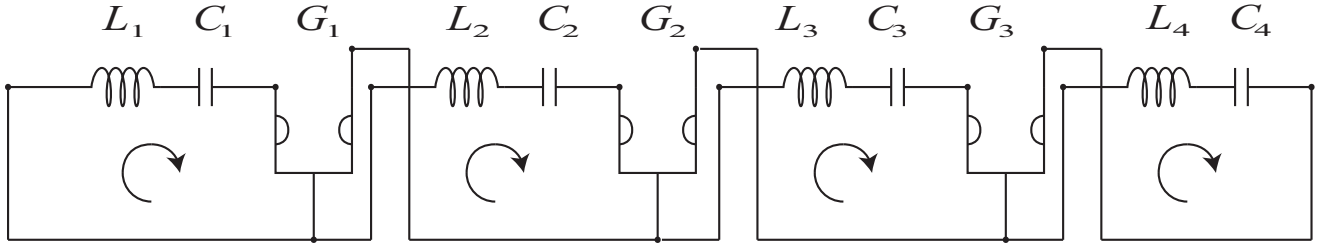}\caption{\label{fig:pri-cir4} Special circuit for 3 $LC$-loops.}
\end{figure}

The values of circuit elements are provided in Tables \ref{tab:spe-cir4a}
and \ref{tab:spe-cir4b} (parameter $r$ refers to the roots of the
characteristic polynomial $\chi\left(s\right)$, that is $\chi\left(r\right)=0$).
\begin{table}[h]
\centering{}\caption{Circuit elements values, $n=4$}
\label{tab:spe-cir4a}%
\begin{tabular}{|c||c||c||c||c|}
\hline 
$k$ & $1$ & $2$ & $3$ & $4$\tabularnewline[0.1cm]
\hline 
\hline 
$L_{k}$ & $\frac{1}{a^{8}}$ & $\frac{1}{4\,a^{6}}$ & $\frac{1}{6\,a^{4}}$ & $\frac{1}{4\,a^{2}}$\tabularnewline[0.1cm]
\hline 
\hline 
$C_{k}$ & $-4\,a^{6}$ & $-6\,a^{4}$ & $-4\,a^{2}$ & $-1$\tabularnewline[0.1cm]
\hline 
\hline 
$G_{k}$ & $-\frac{1}{4\,a^{6}}$ & $-\frac{1}{6\,a^{4}}$ & $-\frac{1}{4\,a^{2}}$ & \tabularnewline[0.1cm]
\hline 
\hline 
$r;\;n$ & $\pm a;$$\;4$ &  &  & \tabularnewline[0.1cm]
\hline 
\end{tabular}
\end{table}

The corresponding polynomial is
\begin{equation}
\chi\left(s\right)=\left(s^{2}-a^{2}\right)^{4},\quad\left(s^{2}+b^{2}\right)^{4},\label{eq:specir4a}
\end{equation}
In the case when the roots of $\chi\left(s\right)$ are real numbers
$\pm a$ the circuit matrices are
\begin{equation}
\mathscr{H}=\left[\begin{array}{rrrrrrrr}
0 & 0 & 0 & 0 & a^{8} & 0 & 0 & 0\\
1 & 0 & 0 & 0 & 0 & 4\,a^{6} & 0 & 0\\
0 & 1 & 0 & 0 & 0 & 0 & 6\,a^{4} & 0\\
0 & 0 & 1 & 0 & 0 & 0 & 0 & 4\,a^{2}\\
0 & 0 & 0 & 0 & 0 & -1 & 0 & 0\\
0 & 0 & 0 & 0 & 0 & 0 & -1 & 0\\
0 & 0 & 0 & 0 & 0 & 0 & 0 & -1\\
0 & 0 & 0 & 1 & 0 & 0 & 0 & 0
\end{array}\right],\quad\mathscr{J}=\left[\begin{array}{rrrrrrrr}
a & 1 & 0 & 0 & 0 & 0 & 0 & 0\\
0 & a & 1 & 0 & 0 & 0 & 0 & 0\\
0 & 0 & a & 1 & 0 & 0 & 0 & 0\\
0 & 0 & 0 & a & 0 & 0 & 0 & 0\\
0 & 0 & 0 & 0 & -a & 1 & 0 & 0\\
0 & 0 & 0 & 0 & 0 & -a & 1 & 0\\
0 & 0 & 0 & 0 & 0 & 0 & -a & 1\\
0 & 0 & 0 & 0 & 0 & 0 & 0 & -a
\end{array}\right],\label{eq:specir4bH}
\end{equation}

\begin{equation}
\mathscr{Z}=\left[\begin{array}{rrrrrrrr}
-a^{11} & -3\,a^{10} & -3\,a^{9} & -a^{8} & a^{11} & -3\,a^{10} & 3\,a^{9} & -a^{8}\\
3\,a^{10} & 14\,a^{9} & 23\,a^{8} & 16\,a^{7} & 3\,a^{10} & -14\,a^{9} & 23\,a^{8} & -16\,a^{7}\\
-3\,a^{9} & -19\,a^{8} & -48\,a^{7} & -56\,a^{6} & 3\,a^{9} & -19\,a^{8} & 48\,a^{7} & -56\,a^{6}\\
a^{8} & 8\,a^{7} & 28\,a^{6} & 56\,a^{5} & a^{8} & -8\,a^{7} & 28\,a^{6} & -56\,a^{5}\\
-a^{4} & -4\,a^{3} & -6\,a^{2} & -4\,a & -a^{4} & 4\,a^{3} & -6\,a^{2} & 4\,a\\
a^{5} & 5\,a^{4} & 10\,a^{3} & 10\,a^{2} & -a^{5} & 5\,a^{4} & -10\,a^{3} & 10\,a^{2}\\
-a^{6} & -6\,a^{5} & -15\,a^{4} & -20\,a^{3} & -a^{6} & 6\,a^{5} & -15\,a^{4} & 20\,a^{3}\\
a^{7} & 7\,a^{6} & 21\,a^{5} & 35\,a^{4} & -a^{7} & 7\,a^{6} & -21\,a^{5} & 35\,a^{4}
\end{array}\right]\label{eq:specir4bZ}
\end{equation}
\begin{table}[h]
\centering{}\caption{Circuit elements values, $n=4$}
\label{tab:spe-cir4b}%
\begin{tabular}{|c||c||c||c||c|}
\hline 
$k$ & $1$ & $2$ & $3$ & $4$\tabularnewline[0.1cm]
\hline 
\hline 
$L_{k}$ & $\frac{1}{b^{8}}$ & -$\frac{1}{4\,b^{6}}$ & $\frac{1}{6\,b^{4}}$ & -$\frac{1}{4\,b^{2}}$\tabularnewline[0.1cm]
\hline 
\hline 
$C_{k}$ & $4\,b^{6}$ & $-6\,b^{4}$ & $4\,b^{2}$ & $-1$\tabularnewline[0.1cm]
\hline 
\hline 
$G_{k}$ & $\frac{1}{4\,b^{6}}$ & $-\frac{1}{6\,b^{4}}$ & $\frac{1}{4\,b^{2}}$ & \tabularnewline[0.1cm]
\hline 
\hline 
$r;\;n$ & $\pm b\mathrm{i};$$\;4$ &  &  & \tabularnewline[0.1cm]
\hline 
\end{tabular}
\end{table}
In the case when the roots of $\chi\left(s\right)$ are pure imaginary
numbers $\pm b\mathrm{i}$ the circuit matrices are

\begin{equation}
\mathscr{H}=\left[\begin{array}{rrrrrrrr}
0 & 0 & 0 & 0 & b^{8} & 0 & 0 & 0\\
1 & 0 & 0 & 0 & 0 & -4\,b^{6} & 0 & 0\\
0 & 1 & 0 & 0 & 0 & 0 & 6\,b^{4} & 0\\
0 & 0 & 1 & 0 & 0 & 0 & 0 & -4\,b^{2}\\
0 & 0 & 0 & 0 & 0 & -1 & 0 & 0\\
0 & 0 & 0 & 0 & 0 & 0 & -1 & 0\\
0 & 0 & 0 & 0 & 0 & 0 & 0 & -1\\
0 & 0 & 0 & 1 & 0 & 0 & 0 & 0
\end{array}\right],\quad\mathscr{J}=\left[\begin{array}{rrrrrrrr}
b\mathrm{i} & 1 & 0 & 0 & 0 & 0 & 0 & 0\\
0 & b\mathrm{i} & 1 & 0 & 0 & 0 & 0 & 0\\
0 & 0 & b\mathrm{i} & 1 & 0 & 0 & 0 & 0\\
0 & 0 & 0 & b\mathrm{i} & 0 & 0 & 0 & 0\\
0 & 0 & 0 & 0 & -b\mathrm{i} & 1 & 0 & 0\\
0 & 0 & 0 & 0 & 0 & -b\mathrm{i} & 1 & 0\\
0 & 0 & 0 & 0 & 0 & 0 & -b\mathrm{i} & 1\\
0 & 0 & 0 & 0 & 0 & 0 & 0 & -b\mathrm{i}
\end{array}\right],\label{eq:specir4cH}
\end{equation}

\begin{equation}
\mathscr{Z}=\left[\begin{array}{rrrrrrrr}
b^{11}\mathrm{i} & 3\,b^{10} & -3\,b^{9}\mathrm{i} & -b^{8} & -b^{11}\mathrm{i} & 3\,b^{10} & 3\,b^{9}\mathrm{i} & -b^{8}\\
-3\,b^{10} & 14\,b^{9}\mathrm{i} & 23\,b^{8} & -16\,b^{7}\mathrm{i} & -3\,b^{10} & -14\,b^{9}\mathrm{i} & 23\,b^{8} & 16\,b^{7}\mathrm{i}\\
-3\,b^{9}\mathrm{i} & -19\,b^{8} & 48\,b^{7}\mathrm{i} & 56\,b^{6} & 3\,b^{9}\mathrm{i} & -19\,b^{8} & -48\,b^{7}\mathrm{i} & 56\,b^{6}\\
b^{8} & -8\,b^{7}\mathrm{i} & -28\,b^{6} & 56\,b^{5}\mathrm{i} & b^{8} & 8\,b^{7}\mathrm{i} & -28\,b^{6} & -56\,b^{5}\mathrm{i}\\
-b^{4} & 4\,b^{3}\mathrm{i} & 6\,b^{2} & -4\,b\mathrm{i} & -b^{4} & -4\,b^{3}\mathrm{i} & 6\,b^{2} & 4\,b\mathrm{i}\\
b^{5}\mathrm{i} & 5\,b^{4} & -10\,b^{3}\mathrm{i} & -10\,b^{2} & -b^{5}\mathrm{i} & 5\,b^{4} & 10\,b^{3}\mathrm{i} & -10\,b^{2}\\
b^{6} & -6\,b^{5}\mathrm{i} & -15\,b^{4} & 20\,b^{3}\mathrm{i} & b^{6} & 6\,b^{5}\mathrm{i} & -15\,b^{4} & -20\,b^{3}\mathrm{i}\\
-b^{7}\mathrm{i} & -7\,b^{6} & 21\,b^{5}\mathrm{i} & 35\,b^{4} & b^{7}\mathrm{i} & -7\,b^{6} & -21\,b^{5}\mathrm{i} & 35\,b^{4}
\end{array}\right].\label{eq:specir4cZ}
\end{equation}
Applying to the case $n=4$ the general formulas (\ref{eq:specirn2h})
for the eigenfrequencies $\omega_{j}=\pm\frac{1}{\sqrt{L_{j}C_{j}}}$
of $LC$-loops as they were decoupled we obtain
\begin{equation}
\omega_{1}=\pm\frac{b}{2},\quad\omega_{2}=\pm\frac{\sqrt{6}}{3}b,\quad\omega_{3}=\pm\frac{\sqrt{6}}{2}b,\quad\omega_{4}=\pm2b.\label{eq:specir4d}
\end{equation}

\subsection{Special circuit for Hamiltonian quadruple of complex eigenvalues
and $n=2$}

The special circuit for $4$ f-loops is shown in Fig. \ref{fig:pri-cir4}.
The circuit polynomial in this special case is
\begin{gather}
\chi\left(s\right)=\left[\left(s-\zeta\right)\left(s-\bar{\zeta}\right)\left(s+\zeta\right)\left(s+\bar{\zeta}\right)\right]^{2}=\label{eq:chiquad1}\\
=\left(a^{2}+b^{2}+2bs+s^{2}\right)^{2}\left(a^{2}+b^{2}-2bs+s^{2}\right)^{2},\quad\zeta=a+b\mathrm{i},\nonumber 
\end{gather}
where $a\neq0$ and $b\neq0$. The values of the circuit elements
are provided in Table \ref{tab:spe-cir4com} (parameter $r$ refers
to the roots of the characteristic polynomial $\chi\left(s\right)$,
that is $\chi\left(r\right)=0$).. We have computed all significant
matrices but the number of their entries combined with the length
of their expressions are too large to be displayed here.
\begin{table}[h]
\centering{}\caption{Circuit elements values, $n=2$}
\label{tab:spe-cir4com}%
\begin{tabular}{|c||c||c||c||c|}
\hline 
$k$ & $1$ & $2$ & $3$ & $4$\tabularnewline[0.1cm]
\hline 
\hline 
$L_{k}$ & $\frac{1}{\left(a^{2}+b^{2}\right)^{4}}$ & $-\frac{1}{4\,\left(a^{2}-b^{2}\right)\left(a^{2}+b^{2}\right)^{2}}$ & $\frac{1}{2\left(3\,a^{4}+3\,b^{4}-a^{2}b^{2}\right)}$ & $-\frac{1}{4\,\left(a^{2}-b^{2}\right)}$\tabularnewline[0.1cm]
\hline 
\hline 
$C_{k}$ & $4\,\left(a^{2}-b^{2}\right)\left(a^{2}+b^{2}\right)^{2}$ & $-2\left(3\,a^{4}+3\,b^{4}-a^{2}b^{2}\right)$ & $4\,\left(a^{2}-b^{2}\right)$ & $-1$\tabularnewline[0.1cm]
\hline 
\hline 
$G_{k}$ & $\frac{1}{4\,\left(a^{2}-b^{2}\right)\left(a^{2}+b^{2}\right)^{2}}$ & $-\frac{1}{2\left(2a^{2}b^{2}-3\,a^{4}-3\,b^{4}\right)}$ & $\frac{1}{4\,\left(a^{2}-b^{2}\right)}$ & \tabularnewline[0.1cm]
\hline 
\hline 
$r;\;n$ & $a\pm b\mathrm{i},\;-a\pm b\mathrm{i}$;$\;2$ &  &  & \tabularnewline[0.1cm]
\hline 
\end{tabular}
\end{table}

\subsection{Special circuits for two real or pure imaginary eigenvalues for $n\protect\geq2$\label{subsec:spe-cirn}}

We consider here the special circuits associated with the characteristic
polynomials $\chi\left(s\right)$ defined by equation (\ref{eq:specirk1a})
for arbitrary $n\geq2$. Namely, in case of $\chi\left(s\right)=\left(s^{2}-a^{2}\right)^{n}$
using the binomial formula and expression (\ref{eq:Lpricir1d}) for
general characteristic polynomial we obtain the following expressions
for coefficients
\begin{equation}
a_{j}=\left(-1\right)^{j}\binom{n}{j}a^{2\left(n-j\right)},\quad1\leq j\leq n,\label{eq:specirn1a}
\end{equation}
where $\binom{n}{j}$ is the binomial coefficient defined in equation
(\ref{eq:cafoH1g}). Then equations (\ref{eq:Lpricir1f})-(\ref{eq:Lpricir1g})
we obtain the following formulas for the circuit elements
\begin{equation}
L_{j}=\frac{1}{\left(-1\right)^{j-1}a_{j-1}}=\frac{1}{\binom{n}{j-1}a^{2\left(n-j+1\right)}},\quad1\leq j\leq n;\label{eq:specirn1b}
\end{equation}
\begin{equation}
C_{j}=\left(-1\right)^{j-1}a_{j}=-\binom{n}{j}a^{2\left(n-j\right)},\quad1\leq j\leq n-1,\quad C_{n}=-1\label{eq:specirn1c}
\end{equation}
\begin{equation}
G_{j}=\frac{1}{\left(-1\right)^{j-1}a_{j}}=-\frac{1}{\binom{n}{j}a^{2\left(n-j\right)}},\quad1\leq j\leq n-1.\label{eq:specirn1d}
\end{equation}
The equations (\ref{eq:specirn1b}) and (\ref{eq:specirn1b}) imply
\begin{equation}
L_{j}C_{j}=-\frac{\binom{n}{j}}{\binom{n}{j-1}a^{2}}=-\frac{n-j+1}{ja^{2}},\quad1\leq j\leq n.\label{eq:specirn1e}
\end{equation}
The formulas (\ref{eq:specirn1e}) imply in turn the following expressions
for the eigenfrequencies of involved $LC$-loops when decoupled
\begin{equation}
\omega_{j}=\pm\frac{1}{\sqrt{L_{j}C_{j}}}=\pm\sqrt{\frac{j}{n-j+1}}a\mathrm{i},\quad1\leq j\leq n.\label{eq:specirn1h}
\end{equation}
The eigenfrequencies $\omega_{j}$ in equations (\ref{eq:specirn1h})
are evidently pure imaginary.

The values of circuit elements and other quantities associated with
the characteristic polynomial $\chi\left(s\right)=\left(s^{2}+b^{2}\right)^{n}$
can be readily obtained from equations (\ref{eq:specirn1a})-(\ref{eq:specirn1h})
by plugging into them $a=b\mathrm{i}$. That yields
\begin{equation}
a_{j}=\left(-1\right)^{n}\binom{n}{j}b^{2\left(n-j\right)},\quad1\leq j\leq n.\label{eq:specirn2a}
\end{equation}
\begin{equation}
L_{j}=\frac{1}{\left(-1\right)^{j-1}a_{j-1}}=\frac{\left(-1\right)^{n-j+1}}{\binom{n}{j-1}b^{2\left(n-j+1\right)}},\quad1\leq j\leq n;\label{eq:specirn2b}
\end{equation}
\begin{equation}
C_{j}=\left(-1\right)^{j-1}a_{j}=\left(-1\right)^{n-j+1}\binom{n}{j}b^{2\left(n-j\right)},\quad1\leq j\leq n-1,\quad C_{n}=-1;\label{eq:specirn2c}
\end{equation}
\begin{equation}
G_{j}=\frac{1}{\left(-1\right)^{j-1}a_{j}}=\frac{\left(-1\right)^{n-j+1}}{\binom{n}{j}b^{2\left(n-j\right)}},\quad1\leq j\leq n-1.\label{eq:specirn2d}
\end{equation}
\begin{equation}
L_{j}C_{j}=\frac{\binom{n}{j}}{\binom{n}{j-1}b^{2}}=\frac{n-j+1}{jb^{2}},,\quad1\leq j\leq n.\label{eq:specirn2e}
\end{equation}
Then the eigenfrequencies of $LC$-loops as they were decoupled are
\begin{equation}
\omega_{j}=\pm\frac{1}{\sqrt{L_{j}C_{j}}}=\pm\sqrt{\frac{j}{n-j+1}}b,\quad1\leq j\leq n.\label{eq:specirn2h}
\end{equation}

\section{Jordan form of a circuit composed of two $LC$-loops and a gyrator}

In previous sections our primary goal was to introduce and study some
circuits with evolution matrices exhibiting nontrivial Jordan block.
The goal of this section is somewhat different. We want to study here
our simplest circuit composed of two $LC$-loops coupled by a gyrator
as shown in Fig. \ref{fig:pri-cir2} without imposing initially any
a priori assumptions on the circuit parameters $L_{1}$, $C_{1}$,
$L_{2}$, $C_{2}$ and $G_{1}$ except for that they are all real
and non-zero. The Lagrangian $\mathcal{L}$ for such a circuit is
described by equation (\ref{eq:specir21b}) and its evolution equations
are the corresponding EL equations (\ref{eq:pricir21c}). These equations
can be readily recast into the following matrix form
\begin{equation}
A\left(\partial_{t}\right)q=0,\quad q=\left[\begin{array}{c}
q_{1}\\
q_{2}
\end{array}\right],\quad A\left(s\right)=\left[\begin{array}{cc}
s^{2}+\frac{1}{L_{1}C_{1}} & -\frac{sG_{1}}{L_{1}}\\
\frac{sG_{1}}{L_{2}} & s^{2}+\frac{1}{L_{2}C_{2}}
\end{array}\right],\label{eq:Adqs1a}
\end{equation}
where $A\left(s\right)$ is evidently a $2\times2$ monic matrix polynomial
of $s$ of the degree 2, namely
\begin{equation}
A\left(s\right)=s^{2}\left[\begin{array}{cc}
1 & 0\\
0 & 1
\end{array}\right]+s\left[\begin{array}{cc}
0 & -\frac{G_{1}}{L_{1}}\\
\frac{G_{1}}{L_{2}} & 0
\end{array}\right]+\left[\begin{array}{cc}
\frac{1}{L_{1}C_{1}} & 0\\
0 & \frac{1}{L_{2}C_{2}}
\end{array}\right].\label{eq:Adqa1b}
\end{equation}
Then matrix polynomial eigenvalue problem associated with the matrix
differential equation (\ref{eq:Adqs1a}) is
\begin{equation}
A\left(s\right)q=\left[\begin{array}{cc}
s^{2}+\frac{1}{L_{1}C_{1}} & -\frac{sG_{1}}{L_{1}}\\
\frac{sG_{1}}{L_{2}} & s^{2}+\frac{1}{L_{2}C_{2}}
\end{array}\right]\left[\begin{array}{c}
q_{1}\\
q_{2}
\end{array}\right]=0.\label{eq:Adqs1c}
\end{equation}
The matrix polynomial eigenvalue problem (\ref{eq:Adqs1c}) is evidently
nonlinear. According to matrix polynimals theory reviewed in Section
\ref{sec:mat-poly} the second-order vector differential equation
(\ref{eq:Adqs1a}) can be reduced to the standard first-order vector
differential equation
\begin{equation}
\partial_{t}\mathsf{q}=\mathscr{C}\mathsf{q},\quad\mathscr{C}=\left[\begin{array}{rrrr}
0 & 0 & 1 & 0\\
0 & 0 & 0 & 1\\
-\frac{1}{L_{1}C_{1}} & 0 & 0 & \frac{G_{1}}{L_{1}}\\
0 & -\frac{1}{L_{2}C_{2}} & -\frac{G_{1}}{L_{2}} & 0
\end{array}\right],\quad\mathsf{q}=\left[\begin{array}{c}
q\\
\partial_{t}q
\end{array}\right],\label{eq:Adqs1d}
\end{equation}
where $\mathscr{C}$ is the $4\times4$ companion matrix for the matrix
polynomial $A\left(s\right)$. Then the standard eigenvalue problem
corresponding to the matrix polynomial eigenvalue problem (\ref{eq:Adqa1b})
is
\begin{equation}
\left(s\mathbb{I}-\mathscr{C}\right)\mathsf{q}=0,\quad\mathsf{q}=\left[\begin{array}{c}
q\\
sq
\end{array}\right].\label{eq:Adqs1e}
\end{equation}
The characteristic polynomial associated with the matrix polynomial
$A\left(s\right)$ and its linearized version $s\mathbb{I}-\mathscr{C}$
is
\begin{equation}
\chi\left(s\right)=\det\left\{ A\left(s\right)\right\} =\det\left\{ s\mathbb{I}-\mathscr{C}\right\} =s^{4}+\left(\frac{G_{1}^{2}}{L_{1}L_{2}}+\frac{1}{L_{2}C_{2}}+\frac{1}{L_{1}C_{1}}\right)s^{2}+\frac{1}{L_{1}C_{1}L_{2}C_{2}}.\label{eq:Adqs2a}
\end{equation}
Consequently, the eigenvalues associated with the eigenvalue problems
(\ref{eq:Adqs1c}) and (\ref{eq:Adqs1e}) can be found from the characteristic
equation
\begin{equation}
\chi\left(s\right)=s^{4}+\left(\frac{G_{1}^{2}}{L_{1}L_{2}}+\frac{1}{L_{2}C_{2}}+\frac{1}{L_{1}C_{1}}\right)s^{2}+\frac{1}{L_{1}C_{1}L_{2}C_{2}}=0.\label{eq:Adqs2as}
\end{equation}

Our primary goal here is to identify all real non-zero values of the
circuit parameters $L_{1}$, $C_{1}$, $L_{2}$, $C_{2}$ and $G_{1}$
for which the matrix $\mathscr{C}$ defined by equation (\ref{eq:Adqs1d})
has nontrivial Jordan form. Our general studies of matrix polynomials
in Section \ref{sec:mat-poly}, particularly Theorem \ref{thm:Jord-block}
imply the following statement.
\begin{thm}[Jordan form of the companion matrix]
\label{thm:Jord2LC} Let $s_{0}$ be an eigenvalue of the companion
matrix $\mathscr{C}$ defined by equation (\ref{eq:Adqs1d}) such
that its algebraic multiplicity $m\left(s_{0}\right)$$\geq2$. Then
(i) $s_{0}\neq0$; (ii) $-s_{0}$ is also an eigenvalue of $\mathscr{C}$;
(iii) $s_{0}$ is either real or pure imaginary; (iv) $m\left(s_{0}\right)$$=m\left(-s_{0}\right)=2$;
(v) the Jordan form $\mathscr{J}$ of the matrix $\mathscr{C}$ is
\begin{equation}
\mathscr{J}=\left[\begin{array}{rrrr}
s_{0} & 1 & 0 & 0\\
0 & s_{0} & 0 & 0\\
0 & 0 & -s_{0} & 1\\
0 & 0 & 0 & -s_{0}
\end{array}\right].\label{eq:msJ1a}
\end{equation}
That is because of the special form the companion matrix $\mathscr{C}$
the eigenvalue degeneracy for $\mathscr{C}$ implies that its Jordan
form $\mathscr{J}$ consists of two Jordan blocks of the size 2.
\end{thm}

\begin{proof}
The eigenvalue $s_{0}$ satisfies $\chi\left(s_{0}\right)=0$. Since
in view of equation (\ref{eq:Adqs2a}) $\chi\left(0\right)=\frac{1}{L_{1}C_{1}L_{2}C_{2}}\neq0$
we infer that $s_{0}\neq0$. The characteristic equation (\ref{eq:Adqs2as})
implies that $\chi\left(-s_{0}\right)=\chi\left(s_{0}\right)=0$ and
hence $-s_{0}$ is an eigenvalue. Notice that since all coefficients
of the characteristic equation (\ref{eq:Adqs2as}) are real then number
$\bar{s}_{0}$ which complex-conjugate to $s_{0}$ is also an eigenvalue
since $\chi\left(\bar{s}_{0}\right)=\overline{\chi\left(s_{0}\right)}=0$.
If $s_{0}$ would be a complex number with non-zero real and imaginary
parts then there would be four distinct eigenvalues $s_{0}$, $-s_{0}$,
$\bar{s}_{0}$ and $-\bar{s}_{0}$ for the the forth-degree characteristic
equation. That would make it impossible for the algebraic multiplicity
of $s_{0}$ to satisfy $m\left(s_{0}\right)$$\geq2$ which is a condition
of the theorem. Hence we have to infer that $s_{0}$ is either real
or pure imaginary. Consequently $s_{0}$ and $-s_{0}$ are the only
eigenvalues and are the roots of the characteristic polynomial $\chi\left(s\right)$.
Since $\chi\left(s\right)$ involves only even degrees of $s$ we
also have $m\left(s_{0}\right)=m\left(-s_{0}\right)$ implying $m\left(s_{0}\right)=m\left(-s_{0}\right)=2$.

Notice now that 
\begin{equation}
A\left(s\right)\neq0\text{ for any complex }s\label{eq:msJ1b}
\end{equation}
for otherwise based on entries of matrix $A\left(s\right)$ defined
by equation (\ref{eq:Adqs1a}) we will have to infer consequently
that $s=0$ and then $\frac{1}{L_{1}C_{1}}=0$, but the later is impossible.
Notice also
\begin{equation}
\dim\left\{ \ker\left\{ A\left(s_{0}\right)\right\} \right\} =1\label{eq:msJ1c}
\end{equation}
Indeed, since $\det\left\{ A\left(s_{0}\right)\right\} =0$ we have
$\dim\left\{ \ker\left\{ A\left(s_{0}\right)\right\} \right\} \geq1$.
On the other hand in view of relation (\ref{eq:msJ1b}) $A\left(s_{0}\right)\neq0$
implying that $\dim\left\{ \ker\left\{ A\left(s_{0}\right)\right\} \right\} <2$.
Hence we may conclude that equation (\ref{eq:msJ1c}) holds. Using
equation (\ref{eq:msJ1c}) and the statement of Theorem \ref{thm:Jord-block}
we can infer that matrix $\mathscr{C}$ has exactly one Jordan block
associated with $s_{0}$ of the size $m\left(s_{0}\right)=2$ and
the same statement holds for $-s_{0}$. Consequently, the Jordan form
of matrix $\mathscr{C}$ satisfies equation (\ref{eq:msJ1a}) and
this completes the proof of the theorem.
\end{proof}

\subsection{Characteristic equation and eigenvalue degeneracy}

The further analytical developments suggest to introduce the following
variables
\begin{equation}
\xi_{1}=\frac{1}{L_{1}C_{1}},\quad\xi_{2}=\frac{1}{L_{2}C_{2}},\quad g=G_{1}^{2},\quad h=s^{2},\label{eq:Adqs2b}
\end{equation}
and refer to positive $g$ as the \emph{gyration parameter}. Then
the companion matrix $\mathscr{C}$ defined by equations (\ref{eq:Adqs1d})
and its characteristic function $\chi\left(s\right)$ as in equation
(\ref{eq:Adqs2as}) take respectively the following forms
\begin{equation}
\mathscr{C}=\left[\begin{array}{rrrr}
0 & 0 & 1 & 0\\
0 & 0 & 0 & 1\\
-\xi_{1} & 0 & 0 & \frac{G_{1}}{L_{1}}\\
0 & -\xi_{2} & -\frac{G_{1}}{L_{2}} & 0
\end{array}\right],\label{eq:Adqs2c}
\end{equation}
\begin{equation}
\chi\left(s\right)=\chi_{h}=h^{2}+\left(\xi_{1}+\xi_{2}+\frac{g}{L_{1}L_{2}}\right)h+\xi_{1}\xi_{2},\quad h=s^{2}.\label{eq:Adqs2d}
\end{equation}
Being interested in degenerate eigenvalues satisfying equation $\chi\left(s\right)=\chi_{h}=0$
we turn to the discriminant $\Delta_{h}$ of the quadratic polynomial
$\chi_{h}$ defined by equation (\ref{eq:Adqs2d}), namely
\begin{equation}
\Delta_{h}=\frac{g^{2}}{L_{1}^{2}L_{2}^{2}}+\frac{2\left(\xi_{1}+\xi_{2}\right)g}{L_{1}L_{2}}+\left(\xi_{1}-\xi_{2}\right)^{2}.\label{eq:Adqs2e}
\end{equation}
Recall that the solutions to the quadratic equation $\chi_{h}=0$
are
\begin{equation}
h_{\pm}=\frac{-\left(\xi_{1}+\xi_{2}+\frac{g}{L_{1}L_{2}}\right)\pm\sqrt{\Delta_{h}}}{2}.\label{eq:Adqs2eh}
\end{equation}
Then the corresponding four solutions $s$ to the characteristic equation
(\ref{eq:Adqs2as}), that is the eigenvalues, are
\begin{equation}
s=\pm\sqrt{h_{+}},\pm\sqrt{h_{-}},\label{eq:Adqs2ehs}
\end{equation}
where $h_{\pm}$ satisfy equations (\ref{eq:Adqs2eh}).

Turning back to $h_{\pm}$ in (\ref{eq:Adqs2eh}) we notice that the
eigenvalue degeneracy condition turns into equation $\Delta_{h}=0$.
This equation can be viewed a constraint on the coefficients of the
quadratic in $h$ polynomial $\chi_{h}$ and ultimately on the circuit
parameters, namely
\begin{equation}
L_{1}^{2}L_{2}^{2}\Delta_{h}=g^{2}+2\left(\xi_{1}+\xi_{2}\right)gL_{1}L_{2}+\left(\xi_{1}-\xi_{2}\right)^{2}L_{1}^{2}L_{2}^{2}=0.\label{eq:Adqs3a}
\end{equation}
Equation (\ref{eq:Adqs3a}) is evidently a quadratic equation for
$g$. Being given remaining circuit coefficients $\xi_{1}$, $\xi_{2}$,
$L_{1}$ and $L_{2}$ this quadratic in $g$ equation has exactly
two solutions
\begin{equation}
\dot{g}_{\delta}=\left(-\xi_{1}-\xi_{2}+2\delta\sqrt{\xi_{1}\xi_{2}}\right)L_{1}L_{2},\quad\delta=\pm1.\label{eq:Adqs3b}
\end{equation}
We refer to $\dot{g}_{\delta}$ in equations (\ref{eq:Adqs3b}) as\emph{
special values of the gyration parameter $g$}. For the two special
values $\dot{g}$ we get from equations (\ref{eq:Adqs2eh}) the corresponding
two degenerate roots
\begin{equation}
\dot{h}=-\frac{\xi_{1}+\xi_{2}+\frac{\dot{g}}{L_{1}L_{2}}}{2}=\pm\sqrt{\xi_{1}\xi_{2}}.\label{eq:Adqs3bh}
\end{equation}
Since $G_{1}$ is real then $g=G_{1}^{2}$ is real as well. The expression
(\ref{eq:Adqs3b}) for $g$ is real-valued if and only if
\begin{equation}
\xi_{1}\xi_{2}>0,\text{ or equivalently }\mathrm{\frac{\xi_{1}}{\left|\xi_{1}\right|}=\frac{\xi_{2}}{\left|\xi_{2}\right|}}=\sigma,\label{eq:Adqs3c}
\end{equation}
where we introduced a binary variable $\sigma$ taking values $\pm1$.
We refer to $\sigma$ as the \emph{circuit sign index}. Relations
(\ref{eq:Adqs3c}) imply in particular that the equality of signs
$\mathrm{sign}\,\left\{ \xi_{1}\right\} =\mathrm{sign}\,\left\{ \xi_{2}\right\} $
is a necessary condition for the eigenvalue degeneracy condition $\Delta_{h}=0$
provided that $g$ has to be real-valued.

It follows then from relations (\ref{eq:Adqs3b}) and (\ref{eq:Adqs3c})
that the special values of the gyration parameter $g_{\delta}$ can
be recast as
\begin{gather}
\dot{g}_{\delta}=-\sigma\left(\sqrt{\left|\xi_{1}\right|}+\delta\sqrt{\left|\xi_{2}\right|}\right)^{2}L_{1}L_{2},\quad\delta=\pm1,\label{eq:Adqs3d}
\end{gather}
where $\sqrt{\xi}>0$ for $\xi>0$. Recall that $g=G_{1}^{2}>0$ and
to provide for that we must have in right-hand side of equations (\ref{eq:Adqs3d})
\begin{equation}
-\sigma L_{1}L_{2}>0\text{, or equivalently }\frac{L_{1}L_{2}}{\left|L_{1}L_{2}\right|}=-\sigma.\label{eq:Adqs3e}
\end{equation}
Relations (\ref{eq:Adqs3c}) and (\ref{eq:Adqs3e}) on the signs of
the involved parameters can be combined into the \emph{circuit sign
constraints}
\begin{equation}
\mathrm{sign}\,\left\{ \xi_{1}\right\} =\mathrm{sign}\,\left\{ \xi_{2}\right\} =-\mathrm{sign}\,\left\{ L_{1}L_{2}\right\} =\mathrm{sign}\,\left\{ \sigma\right\} .\label{eq:Adqs3es}
\end{equation}
\emph{Notice that the sign constraints (\ref{eq:Adqs3es}) involving
the circuit index $\sigma$ defined by (\ref{eq:Adqs3c}) are necessary
for the eigenvalue degeneracy condition $\Delta_{h}=0$. }Combing
equations (\ref{eq:Adqs3d}) and (\ref{eq:Adqs3e}) we obtain
\begin{gather}
\dot{g}_{\delta}=\left(\sqrt{\left|\xi_{1}\right|}+\delta\sqrt{\left|\xi_{2}\right|}\right)^{2}\left|L_{1}L_{2}\right|,\quad\delta=\pm1,\text{ assuming the circuit sign constraints}.\label{eq:Adqs3g}
\end{gather}
Since $g=G_{1}^{2}>0$ the special values of the gyrator resistance
$\dot{G}_{1}$ corresponding to the special values $\dot{g}_{\delta}$
as in equation (\ref{eq:Adqs3g}) are
\begin{gather}
\dot{G}_{1}=\sigma_{1}\sqrt{\dot{g}_{\delta}}=\sigma_{1}\left(\sqrt{\left|\xi_{1}\right|}+\delta\sqrt{\left|\xi_{2}\right|}\right)\sqrt{\left|L_{1}L_{2}\right|},\text{ assuming the circuit sign constraints, }\label{eq:Adqs3h}
\end{gather}
where the binary variable $\sigma_{1}$ takes values $\pm1$.

Using representation (\ref{eq:Adqs3g}) for $\dot{g}_{\delta}$ under
the circuit sign constraints (\ref{eq:Adqs3es}) we can recast the
expression for the degenerate root $\dot{h}$ in equations (\ref{eq:Adqs3bh})
as follows
\begin{equation}
\dot{h}_{\delta}=\sigma\delta\sqrt{\left|\xi_{1}\right|}\sqrt{\left|\xi_{2}\right|},\text{for }g=g_{\delta}=\left(\sqrt{\left|\xi_{1}\right|}+\delta\sqrt{\left|\xi_{2}\right|}\right)^{2}\left|L_{1}L_{2}\right|,\quad\delta=\pm1,\label{eq:Adqs3f}
\end{equation}
where $\sigma$ is the circuit sign index defined by equations (\ref{eq:Adqs3c})
and $\sqrt{\xi}>0$ for $\xi>0$.

The important elements of the above analysis are summarized in the
following statement.
\begin{lem}[the signs of the circuit parameters]
\label{lem:cirsign} Let the circuit be as depicted in Fig. \ref{fig:pri-cir2}
and all its parameters $L_{1}$, $C_{1}$, $L_{2}$, $C_{2}$ and
$G_{1}$ be real and non-zero. Then the companion matrix $\mathscr{C}$
satisfying equations (\ref{eq:Adqs1d}) and (\ref{eq:Adqs2c}) has
a degenerate eigenvalue if and only if (i) the sign constraints (\ref{eq:Adqs3es})
hold and (ii) $g=G_{1}^{2}$ satisfies equations (\ref{eq:Adqs3g}).
Under the circuit sign constraints (\ref{eq:Adqs3es}) the degenerate
solution $\dot{h}_{\delta}$ to the quadratic equation $\chi_{h}=0$
(see equation (\ref{eq:Adqs2d})) is determined by equation (\ref{eq:Adqs3f}).

The sign constraints (\ref{eq:Adqs3es}) are equivalent to
\begin{equation}
\mathrm{sign\,}\left\{ C_{1}\right\} =-\mathrm{sign\,}\left\{ L_{2}\right\} ,\quad\mathrm{sign\,}\left\{ L_{1}\right\} =-\mathrm{sign\,}\left\{ C_{2}\right\} .\label{eq:Adqs4b}
\end{equation}
\end{lem}

\begin{proof}
The statements of lemma before equations (\ref{eq:Adqs4b}) have been
already argued. The verification of the equivalency of equations (\ref{eq:Adqs3es})
and (\ref{eq:Adqs4b}) is straightforward.
\end{proof}
Notice that when Lemma \ref{lem:cirsign} provides sharp criteria
for the companion companion matrix $\mathscr{C}$ to have a degenerate
eigenvalue Theorem \ref{thm:Jord2LC} states that if the companion
companion matrix $\mathscr{C}$ has such an eigenvalue then its Jordan
form $\mathscr{J}$ is formed by two Jordan blocks as in equation
(\ref{eq:msJ1a}). Consequently the following statement hold that
combines statements of Theorem \ref{thm:Jord2LC} and Lemma \ref{lem:cirsign}.
\begin{thm}[Jordan form under degeneracy]
\label{thm:jordeg} Let the circuit be as depicted in Fig. \ref{fig:pri-cir2}
and let all its parameters $L_{1}$, $C_{1}$, $L_{2}$, $C_{2}$
and $G_{1}$ be real and non-zero. Then the companion matrix $\mathscr{C}$
satisfying equations (\ref{eq:Adqs1d}) and (\ref{eq:Adqs2c}) has
the Jordan form
\begin{equation}
\mathscr{J}=\left[\begin{array}{rrrr}
s_{0} & 1 & 0 & 0\\
0 & s_{0} & 0 & 0\\
0 & 0 & -s_{0} & 1\\
0 & 0 & 0 & -s_{0}
\end{array}\right],\label{eq:Adqs4c}
\end{equation}
if and only if the circuit parameters satisfy the degeneracy conditions
described in Lemma \ref{lem:cirsign}. Then for $g=g_{\delta}$ we
have
\begin{equation}
\dot{h}_{\delta}=\sigma\delta\sqrt{\left|\xi_{1}\right|}\sqrt{\left|\xi_{2}\right|},\text{ for  }g=g_{\delta}=\left(\sqrt{\left|\xi_{1}\right|}+\delta\sqrt{\left|\xi_{2}\right|}\right)^{2}\left|L_{1}L_{2}\right|,\label{eq:Adqs4ca}
\end{equation}
\begin{equation}
\pm s_{0}=\pm\sqrt{\sigma\delta\sqrt{\left|\xi_{1}\right|}\sqrt{\left|\xi_{2}\right|}},\text{ for  }g=g_{\delta}=\left(\sqrt{\left|\xi_{1}\right|}+\delta\sqrt{\left|\xi_{2}\right|}\right)^{2}\left|L_{1}L_{2}\right|,\label{eq:Adqs4cb}
\end{equation}
where $\sqrt{\xi}>0$ for $\xi>0$, $\delta=\pm1$ and $\sigma$ is
the circuit sign index defined by equations (\ref{eq:Adqs3c}). According
to formula (\ref{eq:Adqs4cb}) degenerate eigenvalues $\pm s_{0}$
depend on the product $\sigma\delta$ and $\left|\xi_{1}\right|,\left|\xi_{2}\right|$
and consequently they are either real or pure imaginary depending
on whether $\delta=\sigma$ or $\delta=-\sigma$.

Notice that in the special case when $\left|\xi_{1}\right|=\left|\xi_{2}\right|$
the parameter $g_{\delta}$ takes only one non-zero value, namely
\begin{equation}
g=g_{1}=4\left|\xi_{1}\right|\left|L_{1}L_{2}\right|,\quad L_{1}L_{2}=-\sigma\left|L_{1}L_{2}\right|,\label{eq:Adqs4d}
\end{equation}
whereas $g_{-1}=0$ which is inconsistent with our assumption $G_{1}\neq0$.
Evidently for $g=0$ the circuit breaks into two independent $LC$-circuits
and in this case the relevant Jordan form is a diagonal $4\times4$
matrix with eigenvalues $\pm\sqrt{\xi_{1}}$ and $\pm\sqrt{\xi_{2}}$
.
\end{thm}

\begin{rem}[instability and marginal stability]
\label{rem:stability} Notice according to formula (\ref{eq:Adqs4cb})
degenerate eigenvalues $\pm s_{0}$ are real for $\delta=\sigma$
and hence they correspond to exponentially growing and decaying in
time solutions indicating instability. For $\delta=-\sigma$ the degenerate
eigenvalues $\pm s_{0}$ are pure imaginary corresponding to oscillatory
solutions indicating that there is at least marginal stability.
\end{rem}

To get a graphical illustration for complex-valued circuit eigenvalues
as a function of the gyration parameter $g$ use the following data
\begin{equation}
\left|\xi_{1}\right|=1,\quad\left|\xi_{2}\right|=2,\quad\left|L_{1}\right|=1,\quad\left|L_{2}\right|=2,\quad\sigma=1,\label{eq:xiLdat1a}
\end{equation}
 and that corresponds to
\begin{equation}
\xi_{1}=1,\quad\xi_{2}=2,\quad L_{1}=\pm1,\quad L_{2}=\mp2.\label{eq:xiLdat1b}
\end{equation}
It follows then from representation (\ref{eq:Adqs3g}) for $\dot{g}_{\delta}$
that the corresponding special values $\dot{g}_{\delta}$ are
\begin{equation}
\dot{g}_{-1}=2\left(1-\sqrt{2}\right)^{2}\cong0.3431457498,\quad\dot{g}_{1}=2\left(1+\sqrt{2}\right)^{2}\cong11.65685425.\label{eq:xiLdat1c}
\end{equation}
\begin{figure}[h]
\centering{}\includegraphics[scale=0.5]{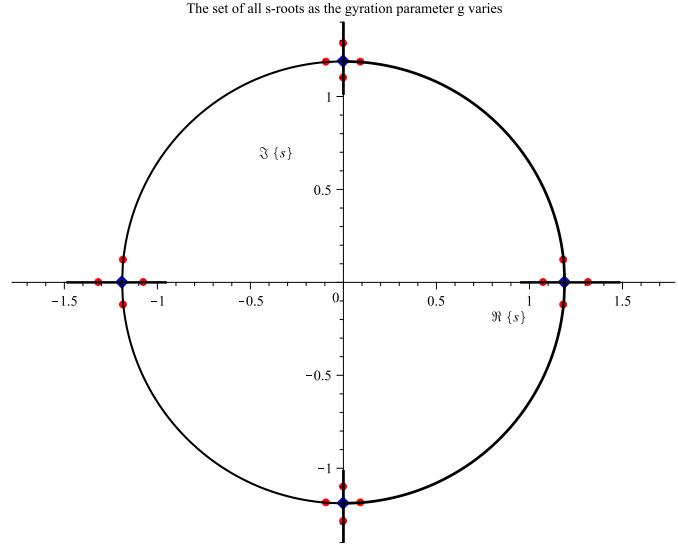}\caption{\label{fig:eigenvals} The plot shows the set $S_{\mathrm{eig}}$
of all complex valued eigenvalues $s$ defined by equations (\ref{eq:Adqs2ehs})
for the data in equations (\ref{eq:xiLdat1a}), (\ref{eq:xiLdat1b})
when the gyration parameter $g$ varies in interval containing special
values $\dot{g}_{-1}$ and $\dot{g}_{1}$ defined in relations (\ref{eq:xiLdat1c}).
The horizontal and vertical axes represent the real and the imaginary
parts $\Re\left\{ s\right\} $ and $\Im\left\{ s\right\} $ of eigenvalues
$s$. $S_{\mathrm{eig}}$ consists of the circle centered in the origin
of radius $\sqrt[4]{\left|\xi_{1}\right|\left|\xi_{2}\right|}$ and
four intersecting its intervals lying on real and imaginary axes.
The circular part of the set $S_{\mathrm{eig}}$ corresponds to all
eigenvalues for $\dot{g}_{-1}\protect\leq g\protect\leq$$\dot{g}_{1}$.
The degenerate eigenvalues $\pm s_{0}$ corresponding to $\dot{g}_{-1}$
and $\dot{g}_{1}$ and defined by equations (\ref{eq:Adqs4cb}) are
shown as solid diamond (blue) dots. Two of them are real, positive
and negative, numbers and another two are pure imaginary, with positive
and negative imaginary parts. The 16 solid circular (red) dots are
associated with 4 quadruples of eigenvalues corresponding to 4 different
values of the gyration parameter $g$ chosen to be slightly larger
or smaller than the special values $\dot{g}_{-1}$ and $\dot{g}_{1}$.
Let us take a look at any of the degenerate eigenvalues identified
by solid diamond (blue) dots. If $g$ is slightly different from its
special values $\dot{g}_{-1}$ and $\dot{g}_{1}$ then each degenerate
eigenvalue point splits into a pair of points identified by solid
circular (red) dots. They are either two real or two pure imaginary
points if $g$ is outside the interval $\left[\dot{g}_{-1},\dot{g}_{1}\right]$,
or alternatively they are two points lying on the circle, if $g$
is inside the interval $\left[\dot{g}_{-1},\dot{g}_{1}\right]$.}
\end{figure}

To explain the rise of the circular part of the set $S_{\mathrm{eig}}$in
Fig. \ref{fig:eigenvals} we recast the characteristic equation (\ref{eq:Adqs2d})
as follows
\begin{equation}
H+\frac{1}{H}=R,\quad R=\frac{\sigma}{\sqrt{\left|\xi_{1}\right|\left|\xi_{2}\right|}}\left[\frac{g}{\left|L_{1}\right|\left|L_{2}\right|}-\left(\left|\xi_{1}\right|+\left|\xi_{2}\right|\right)\right],\quad H=\frac{h}{\sqrt{\left|\xi_{1}\right|\left|\xi_{2}\right|}}.\label{eq:charHR1a}
\end{equation}
Notice that
\begin{equation}
R=2\delta\sigma,\text{for }g=\dot{g}_{\delta}=\left(\sqrt{\left|\xi_{1}\right|}+\delta\sqrt{\left|\xi_{2}\right|}\right)^{2}\left|L_{1}L_{2}\right|,\quad\delta=\pm1,\label{eq:charHR1b}
\end{equation}
where $\sigma=\pm1$ is the circuit sign index. Since $R$ depends
linearly on $g$ relations (\ref{eq:charHR1a}) and (\ref{eq:charHR1b})
imply
\begin{equation}
\left|R\right|\leq2,\text{ for  }\dot{g}_{-1}\leq g\leq\dot{g}_{1};\quad\left|R\right|>2,\text{ for  }g<\dot{g}_{-1}\text{ and }g>\dot{g}_{1}.\label{eq:charHR1c}
\end{equation}
It is an elementary fact that solutions $H$ to equation (\ref{eq:charHR1a})
satisfy the following relations:
\begin{equation}
H=\exp\left\{ \mathrm{i}\theta\right\} ,\text{ for }\left|R\right|\leq2,\text{ and }\cos\left(\theta\right)=\frac{R}{2},\quad0\leq\theta\leq\pi,\label{eq:charHR1d}
\end{equation}
\begin{equation}
H>0,\text{ for }R>2,\text{ and }H<0,\text{ for }R<-2.\label{eq:charHR1e}
\end{equation}
It is also evident from the form of equation (\ref{eq:charHR1a})
that if $H$ is its solution then $H^{-1}$ is a solution as well,
that is the two solutions to equation (\ref{eq:charHR1a}) always
come in pairs of the form $\left\{ H,H^{-1}\right\} $.

Since the eigenvalues $s$ satisfy $s=\pm\sqrt{h}$ the established
above properties of $h=\sqrt{\left|\xi_{1}\right|\left|\xi_{2}\right|}H$
can recast for $s$ as follows.
\begin{thm}[quadruples of eigenvalues]
\label{thm:quad} For every $g>0$ every solution $s$ to the characteristic
equation (\ref{eq:Adqs2as}) is of the form (\ref{eq:Adqs2ehs}) and
the number of solutions is exactly four counting their multiplicity.
Every such a quadruple of solutions is of the following form
\begin{equation}
\left\{ s,\frac{\sqrt{\left|\xi_{1}\right|\left|\xi_{2}\right|}}{s},-s,-\frac{\sqrt{\left|\xi_{1}\right|\left|\xi_{2}\right|}}{s}\right\} ,\label{eq:charHR2a}
\end{equation}
where $s$ is a solution to the characteristic equation (\ref{eq:Adqs2as}.
Then for $\dot{g}_{-1}\leq g\leq\dot{g}_{1}$ the quadruple of solutions
belongs to the circle $\left|s\right|=\sqrt[4]{\left|\xi_{1}\right|\left|\xi_{2}\right|}$
such that
\begin{equation}
s=\delta_{1}\sqrt[4]{\left|\xi_{1}\right|\left|\xi_{2}\right|}\exp\left\{ \mathrm{i}\delta_{2}\theta\right\} ,\cos\left(\theta\right)=\frac{R}{2},\quad0\leq\theta\leq\pi,\delta_{1},\delta_{2}=\pm1,\label{eq:charHR2b}
\end{equation}
where $R$ is defined in relations (\ref{eq:charHR1a}). If $g<\dot{g}_{-1}$
or $g>\dot{g}_{1}$ the quadruple of solutions consists of either
real numbers and pure imaginary numbers depending if $R>2$ or $R<-2$
respectively. In view of relations (\ref{eq:charHR1a}) and $s=\pm\sqrt{h}$
where $h=\sqrt{\left|\xi_{1}\right|\left|\xi_{2}\right|}H$ every
quadruple of solutions as in expression (\ref{eq:charHR2a}) is invariant
with respect to the complex conjugation transformation.
\end{thm}

The following remark discusses in some detail the transition of eigenvalues
lying on the circle $\left|s\right|=\sqrt[4]{\left|\xi_{1}\right|\left|\xi_{2}\right|}$
having non-zero real and imaginary parts into either real or pure
imaginary numbers as the value of the gyration parameter $g$ passes
through its special values $\dot{g}_{-1}$ or $\dot{g}_{1}$ at which
the eigenvalues degenerate.
\begin{rem}[transition at degeneracy points]
 According to formula (\ref{eq:Adqs4cb}) there is total of four
degenerate eigenvalues $\pm s_{0}$, namely $\pm$$\sqrt[4]{\left|\xi_{1}\right|\left|\xi_{2}\right|}$
and $\pm\mathrm{i}$$\sqrt[4]{\left|\xi_{1}\right|\left|\xi_{2}\right|}$
(depicted as solid diamond (blue) dots in Fig. \ref{fig:eigenvals})
that are associated with the two special values of the gyration parameter
$g_{\pm1}=\left(\sqrt{\left|\xi_{1}\right|}\pm\sqrt{\left|\xi_{2}\right|}\right)^{2}\left|L_{1}L_{2}\right|$.
For any value of the gyration parameter $g$ different than its two
special values there are exactly four distinct eigenvalues $s$ forming
a quadruple as in expression (\ref{eq:charHR2a}). If $\dot{g}_{-1}<g<\dot{g}_{1}$
and $g$ gets close to either $\dot{g}_{-1}$ or $\dot{g}_{1}$ the
corresponding four distinct eigenvalues on the circle $\left|s\right|=\sqrt[4]{\left|\xi_{1}\right|\left|\xi_{2}\right|}$
get close to either $\pm$$\sqrt[4]{\left|\xi_{1}\right|\left|\xi_{2}\right|}$
or $\pm\mathrm{i}$$\sqrt[4]{\left|\xi_{1}\right|\left|\xi_{2}\right|}$
as depicted in Fig. \ref{fig:eigenvals}) by solid circle (red) dots.
As $g$ approaches the special values $\dot{g}_{-1}$ or $\dot{g}_{1}$,
reaches them and gets out of the interval $\left[\dot{g}_{-1},\dot{g}_{1}\right]$
the corresponding solid circle (red) dots approach the relevant points
$\pm$$\sqrt[4]{\left|\xi_{1}\right|\left|\xi_{2}\right|}$ or $\pm\mathrm{i}$$\sqrt[4]{\left|\xi_{1}\right|\left|\xi_{2}\right|}$
, merge at them and then split again passing to respectively real
and imaginary axes as illustrated by Fig. \ref{fig:eigenvals}.
\end{rem}

\subsection{Eigenvectors and the Jordan basis}

Theorem \ref{thm:Jord2LC} provides a general statement that the degeneracy
the companion matrix $\mathscr{C}$ defined by equation (\ref{eq:Adqs1d})
implies that its Jordan form consists of 2 Jordan blocks as in equation
(\ref{eq:msJ1a}). We would to extend that statement with a construction
of the corresponding Jordan basis. With that in mind we introduce
the following matrix evidently related to the companion matrix $\mathscr{C}$
\begin{equation}
\mathsf{C}=\mathsf{C}\left(\xi_{1},\xi_{2},b_{1},b_{2}\right)=\left[\begin{array}{rrrr}
0 & 0 & 1 & 0\\
0 & 0 & 0 & 1\\
-\xi_{1} & 0 & 0 & b_{1}\\
0 & -\xi_{2} & b_{2} & 0
\end{array}\right],\label{eq:Adqs4da}
\end{equation}
namely
\begin{equation}
\mathscr{C}=\mathsf{C}\left(\xi_{1},\xi_{2},b_{1},b_{2}\right),\quad b_{1}=\frac{G_{1}}{L_{1}},\quad b_{2}=-\frac{G_{1}}{L_{2}}.\label{eq:Adqs4db}
\end{equation}
Notice that the change of the sign of $G_{1}$ of the gyration capacitance
or change of the sign the parameters $b_{1},b_{2}$ in matrix $C$
yield a matrix that is similar to the original matrix, that is if
$C$ is a matrix defined by equation (\ref{eq:Adqs4da}) we have
\begin{equation}
\mathsf{C}\left(\xi_{1},\xi_{2},-b_{1},-b_{2}\right)=\Lambda^{-1}\mathsf{C}\left(\xi_{1},\xi_{2},b_{1},b_{2}\right)\Lambda,\quad\Lambda=\left[\begin{array}{rrrr}
1 & 0 & 0 & 0\\
0 & -1 & 0 & 0\\
0 & 0 & 1 & 0\\
0 & 0 & 0 & -1
\end{array}\right].\label{eq:Adqs4dc}
\end{equation}
Equations (\ref{eq:Adqs4db}) and (\ref{eq:Adqs4dc}) readily imply
the following statement.
\begin{lem}[alteration of the gyration resistance]
\label{lem:altgyr} Let the circuit be as described in Theorem \ref{thm:jordeg}.
Then the sign alteration of the gyration resistance $G_{1}$ yields
a circuit with the evolution matrix $\mathscr{C}\left(-G_{1}\right)$
that is similar to the evolution matrix $\mathscr{C}\left(G_{1}\right)$
of the original circuit, that is
\begin{equation}
\mathscr{C}\left(-G_{1}\right)=\Lambda^{-1}\mathscr{C}\left(G_{1}\right)\Lambda,\quad\Lambda=\left[\begin{array}{rrrr}
1 & 0 & 0 & 0\\
0 & -1 & 0 & 0\\
0 & 0 & 1 & 0\\
0 & 0 & 0 & -1
\end{array}\right].\label{eq:Adqs4dd}
\end{equation}
\end{lem}

The above considerations suggest to introduce some two special form
matrices intimately related to the companion matrix $\mathscr{C}$
defined by expression (\ref{eq:Adqs2c}). It is a tedious but straightforward
exercise to verify that the following statements hold for these matrices.
\begin{lem}[Jordan form of a special matrix]
\label{lem:Cjord} Let $\mathsf{C}_{\pm}$ be a $4\times4$ matrices
of the form
\begin{equation}
\mathsf{C}_{\pm}=\mathsf{C}_{\pm}\left(\zeta_{1},\zeta_{2},b_{1},b_{2}\right)=\left[\begin{array}{rrrr}
0 & 0 & 1 & 0\\
0 & 0 & 0 & 1\\
-\zeta_{1}^{2} & 0 & 0 & b_{1}\\
0 & -\zeta_{2}^{2} & b_{2} & 0
\end{array}\right],\text{where }b_{1}b_{2}=\left(\zeta_{1}\pm\zeta_{2}\right)^{2}\neq0,\label{eq:matCb1a}
\end{equation}
where $\zeta_{1},\zeta_{2}$ and $b_{1},b_{2}$ are complex numbers.
Then matrices $\mathsf{C}_{\pm}$ can be recast as
\begin{equation}
\mathsf{C}_{\pm}=\mathsf{C}_{\pm}\left(\zeta_{1},\zeta_{2},\alpha\right)=\left[\begin{array}{rrrr}
0 & 0 & 1 & 0\\
0 & 0 & 0 & 1\\
-\zeta_{1}^{2} & 0 & 0 & \alpha\left(\zeta_{1}\pm\zeta_{2}\right)\\
0 & -\zeta_{2}^{2} & \frac{\zeta_{1}\pm\zeta_{2}}{\alpha} & 0
\end{array}\right],\quad\alpha=\frac{b_{1}}{\zeta_{1}\pm\zeta_{2}}.\label{eq:matCb1b}
\end{equation}
The Jordan forms $J_{\pm}$ of the corresponding matrices $\mathsf{C}_{\pm}$
are
\begin{equation}
J_{\pm}=J_{\pm}\left(\zeta_{1},\zeta_{2}\right)=Z_{\pm}^{-1}\mathsf{C}_{\pm}Z_{\pm}=\left[\begin{array}{rrrr}
\sqrt{\pm\zeta_{1}\zeta_{2}} & 1 & 0 & 0\\
0 & \sqrt{\pm\zeta_{1}\zeta_{2}} & 0 & 0\\
0 & 0 & -\sqrt{\pm\zeta_{1}\zeta_{2}} & 1\\
0 & 0 & 0 & -\sqrt{\pm\zeta_{1}\zeta_{2}}
\end{array}\right],\label{eq:matCb1c}
\end{equation}
where $\sqrt{\pm\zeta_{1}\zeta_{2}}$ is one the values of the square
root of $\pm\zeta_{1}\zeta_{2}$, and matrices $Z_{\pm}$ are
\begin{equation}
Z_{\pm}=Z_{\pm}\left(\zeta_{1},\zeta_{2},\alpha\right)=\left[\begin{array}{rrrr}
-\frac{\zeta_{1}}{4\sqrt{\pm\zeta_{1}\zeta_{2}}} & \frac{1}{2\left(\zeta_{1}\pm\zeta_{2}\right)} & \frac{\zeta_{1}}{4\sqrt{\pm\zeta_{1}\zeta_{2}}} & \frac{1}{2\left(\zeta_{1}\pm\zeta_{2}\right)}\\
\mp\frac{\zeta_{1}}{4\zeta_{2}\alpha} & \pm\frac{\zeta_{1}}{4\zeta_{2}\alpha\sqrt{\pm\zeta_{1}\zeta_{2}}} & \mp\frac{\zeta_{1}}{4\zeta_{2}\alpha} & -\frac{\zeta_{1}}{4\zeta_{2}\alpha\sqrt{\pm\zeta_{1}\zeta_{2}}}\\
-\frac{\zeta_{1}}{4} & -\frac{\zeta_{1}\left(\zeta_{1}\mp\zeta_{2}\right)}{4\left(\zeta_{1}\pm\zeta_{2}\right)\sqrt{\pm\zeta_{1}\zeta_{2}}} & -\frac{\zeta_{1}}{4} & \frac{\zeta_{1}\left(\zeta_{1}\mp\zeta_{2}\right)}{4\left(\zeta_{1}\pm\zeta_{2}\right)\sqrt{\pm\zeta_{1}\zeta_{2}}}\\
-\frac{\zeta_{1}^{2}}{4\alpha\sqrt{\pm\zeta_{1}\zeta_{2}}} & 0 & \frac{\zeta_{1}^{2}}{4\alpha\sqrt{\pm\zeta_{1}\zeta_{2}}} & 0
\end{array}\right].\label{eq:matCb1d}
\end{equation}
Notice that the columns of matrices $Z_{\pm}$ form the Jordan bases
of the corresponding matrices $\mathsf{C}_{\pm}$ , and the first
and the third columns $Z_{\pm}$ are the eigenvectors of the corresponding
matrices $\mathsf{C}_{\pm}$ with respective eigenvalues $\sqrt{\pm\zeta_{1}\zeta_{2}}$
and $-\sqrt{\pm\zeta_{1}\zeta_{2}}$.

Notice also that the matrices $\mathsf{C}_{\pm}$ and $Z_{\pm}$ can
be factorized as follows
\begin{equation}
\mathsf{C}_{\pm}\left(\zeta_{1},\zeta_{2},\alpha\right)=\Lambda_{\alpha}^{-1}\mathsf{C}_{\pm}\left(\zeta_{1},\zeta_{2},1\right)\Lambda_{\alpha},\quad Z_{\pm}\left(\zeta_{1},\zeta_{2},\alpha\right)=\Lambda_{\alpha}^{-1}Z_{\pm}\left(\zeta_{1},\zeta_{2},1\right)\label{eq:matCb1e}
\end{equation}
\begin{equation}
\mathsf{C}_{\pm}\left(\zeta_{1},\zeta_{2},-b_{1},-b_{2}\right)=\Lambda_{-1}^{-1}\mathsf{C}_{\pm}\left(\zeta_{1},\zeta_{2},b_{1},b_{2}\right)\Lambda_{-1},\label{eq:matCb1f}
\end{equation}
where matrix $\Lambda_{\alpha}$ is the following diagonal matrix
\begin{equation}
\Lambda_{\alpha}=\left[\begin{array}{rrrr}
1 & 0 & 0 & 0\\
0 & \alpha & 0 & 0\\
0 & 0 & 1 & 0\\
0 & 0 & 0 & \alpha
\end{array}\right].\label{eq:matCb1g}
\end{equation}
\end{lem}

In the case when the gyration resistance takes its special values
$\dot{G}_{1}$ as in equation (\ref{eq:Adqs3h}) the companion matrix
$\mathscr{C}$ defined by expression (\ref{eq:Adqs2c}) can be cast
as matrix $\mathsf{C}_{\pm}$ in equation (\ref{eq:matCb1a}). Indeed
using equation (\ref{eq:Adqs3h}) we readily obtain
\begin{gather}
-\frac{\dot{G}_{1}^{2}}{L_{1}L_{2}}=\sigma\left(\sqrt{\left|\xi_{1}\right|}+\delta\sqrt{\left|\xi_{2}\right|}\right)^{2}=\left(\zeta_{1}+\delta\zeta_{2}\right)^{2},\text{ where }\sqrt{\left|\xi\right|}>0,\label{eq:matCb1h}\\
\zeta_{j}=\left\{ \begin{array}{ccc}
\sqrt{\left|\xi_{j}\right|} & \text{for} & \sigma=1\\
\mathrm{i}\sqrt{\left|\xi_{j}\right|} & \text{for} & \sigma=-1
\end{array}\right.,\quad j=1,2.\nonumber 
\end{gather}
Combing equation (\ref{eq:matCb1h}) with Lemma \ref{lem:Cjord} we
arrive at the following statement.
\begin{thm}[degeneracy and the Jordan form]
\label{thm:degJor} Let the circuit be as depicted in Fig. \ref{fig:pri-cir2}
and let all its parameters $L_{1}$, $C_{1}$, $L_{2}$, $C_{2}$
and $G_{1}$ be real and non-zero. Then the companion matrix $\mathscr{C}$
defined by expression (\ref{eq:Adqs2c}) has a degenerate eigenvalue
if and only if (i) the sign constraints (\ref{eq:Adqs3es}) are satisfied
and (ii) its gyration parameter $g$ takes its two special values
\begin{equation}
\dot{g}_{\delta}=\left(\sqrt{\left|\xi_{1}\right|}+\delta\sqrt{\left|\xi_{2}\right|}\right)^{2}\left|L_{1}L_{2}\right|,\quad\delta=\pm1,\label{eq:gxiLL1a}
\end{equation}
and the corresponding special values of the gyration resistance
\begin{gather}
\dot{G}_{1}=\sigma_{1}\sqrt{\dot{g}_{\delta}}=\sigma_{1}\left(\sqrt{\left|\xi_{1}\right|}+\delta\sqrt{\left|\xi_{2}\right|}\right)\sqrt{\left|L_{1}L_{2}\right|},\label{eq:Adqs3h-1}
\end{gather}
where binary variable $\sigma_{1}$ takes values $\pm1$. For these
special values of the gyration resistance matrix $\mathscr{C}$ has
exactly two degenerate eigenvalues $\pm s_{0}$ of the multiplicity
two satisfying the following equations
\begin{equation}
\pm s_{0}=\pm\sqrt{\sigma\delta\sqrt{\left|\xi_{1}\right|}\sqrt{\left|\xi_{2}\right|}}\label{eq:gxiLL1c}
\end{equation}
where $\sqrt{\xi}>0$ for $\xi>0$, $\delta=\pm1$ and $\sigma=\pm1$
is the circuit sign index defined by equations (\ref{eq:Adqs3c}).
In addition to that matrix $\mathscr{C}$ can represented as matrix
$\mathsf{C}_{\pm}$ in equations (\ref{eq:matCb1b}),with $\zeta_{1},\zeta_{2}$
described by equations (\ref{eq:matCb1h}), and
\[
\alpha=\left\{ \begin{array}{ccc}
\frac{\left|L_{1}\right|}{L_{1}}\sigma_{1}\sqrt{\frac{\left|L_{2}\right|}{\left|L_{1}\right|}} & \text{for} & \sigma=1,\\
\frac{\left|L_{1}\right|}{L_{1}}\sigma_{1}\mathrm{i}\sqrt{\frac{\left|L_{2}\right|}{\left|L_{1}\right|}} & \text{for} & \sigma=-1.
\end{array}\right.
\]
Consequently, all statements of Lemma \ref{lem:Cjord} for matrix
$\mathsf{C}_{\pm}$ hold including its Jordan form $J_{\pm}$ as in
equations (\ref{eq:matCb1c}) and expressions (\ref{eq:matCb1d})
for the Jordan basis as columns of matrix $Z_{\pm}$.
\end{thm}

\section{Circuit synthesis strategy and elements\label{sec:synth-elem}}

The first goal of our synthesis process is to construct a Hamiltonian
system governed by the evolution equation (\ref{eq:XHX1a}) with the
circuit matrix $\mathscr{H}$ having the prescribed Jordan canonical
form subject to natural constraints. Consequently, the circuit matrix
$\mathscr{H}$ has to be a \emph{Hamiltonian matrix}, that is a matrix
obtained from a quadratic Hamiltonian $\mathcal{H}$ with real coefficients.
The spectrum $\mathrm{spec}\,\left(\mathscr{H}\right)$, that is the
set of all distinct eigenvalues, of a Hamiltonian matrix must have
the following property
\begin{gather}
\text{if }\zeta\in\mathrm{spec}\,\left(\mathscr{H}\right),\text{ then }-\zeta,\:\bar{\zeta},\:-\bar{\zeta}\in\mathrm{spec}\,\left(\mathscr{H}\right),\label{eq:XMX1s}\\
\text{the multiplicity of all four numbers is the same,}\nonumber 
\end{gather}
where $\bar{\zeta}$ stands for complex-conjugate to complex number
$\zeta$, see Section \ref{subsec:ham-mat} for details. We refer
to the property (\ref{eq:XMX1s}) as the \emph{Hamiltonian spectral
symmetry}. Suppose that $a\neq0$ and $b\neq0$ are real numbers.
Notice then that the set $\left\{ \zeta,-\zeta,\:\bar{\zeta},\:-\bar{\zeta}\right\} $
consists of (i) two numbers $\left\{ a,-a\right\} $ if $\zeta=a$;
(ii) two numbers $\left\{ b\mathrm{i},-b\mathrm{i}\right\} $ if $\zeta=b\mathrm{i}$;
(ii) four numbers
\begin{equation}
\left\{ a+b\mathrm{i},-a-b\mathrm{i},a-b\mathrm{i},-a+b\mathrm{i}\right\} ,\label{eq:XMX1ab}
\end{equation}
if $\zeta=a+b\mathrm{i}$.

To achieve the desired Jordan form for the circuit matrix $\mathscr{H}$
we introduce the characteristic polynomial $\chi_{\mathscr{H}}\left(s\right)=\chi_{\mathscr{J}}\left(s\right)$
and find its coefficients. Having coefficients $a_{k}$ of the polynomial
$\chi_{\mathscr{H}}\left(s\right)$ as in equations (\ref{eq:Lpricir1d})
and (\ref{eq:Lpricir1e}) we define the Hamiltonian $\mathcal{H}_{a}$
by the following explicit expression

\begin{gather}
\mathcal{H}_{a}=\sum_{k=1}^{n}p_{k+1}q_{k}+\frac{1}{2}\sum_{k=1}^{n}\left(-1\right)^{k-1}a_{k-1}p_{k}^{2}+\frac{1}{2}q_{n}^{2},\label{eq:Hamapq1a}\\
-\infty<a_{k}<\infty,\quad0\leq k\leq n-1;\quad a_{0}\neq0,\label{eq:Hamapq1aa}
\end{gather}
Notice that the system parameters $a_{k}$ can be negative and positive.
The particular choice of signs in expression (\ref{eq:Hamapq1a})
is a matter of convenience. \emph{Hamiltonian $\mathcal{H}_{a}$ defined
by equations (\ref{eq:Hamapq1a}) is fundamental to the synthesis
of all special circuits we construct and we refer to it as principal
Hamiltonian.}

The principal Hamiltonian matrix $\mathscr{H}_{a}$ that corresponds
to the principal Hamiltonian $\mathcal{H}_{a}$ has the following
properties (see Section \ref{sec:priHam} for details):
\begin{itemize}
\item the corresponding to $\mathcal{H}_{a}$ Hamiltonian matrix $\mathscr{H}_{a}$
has the polynomial $\chi\left(s\right)$ defined by equations (\ref{eq:Lpricir1d})
as its characteristic polynomial $\chi_{\mathscr{H}_{a}}\left(s\right)$,
and, consequently, the set of the distinct roots $s_{j}$ of the polynomial
is exactly the set of all distinct eigenvalues of the circuit matrix
$\mathscr{H}_{a}$, that is $\mathrm{spec}\,\left(\mathscr{H}_{a}\right)=\left\{ s_{j}\right\} $; 
\item since $a_{0}\neq0$ we have $s_{j}\neq0$ for every $j$;
\item the spectrum $\mathrm{spec}\,\left(\mathscr{H}_{a}\right)$ satisfies
Hamiltonian spectral symmetry condition (\ref{eq:XMX1s}); 
\item the circuit matrix $\mathscr{H}_{a}$ is\emph{ cyclic (nonderogatory)},
that is it the geometric multiplicity of every eigenvalue $s_{j}$
is exactly one, and every $s_{j}$ is associated with the single Jordan
block $J_{n_{j}}\left(s_{j}\right)$ of the size $n_{j}$ which is
the algebraic multiplicity of eigenvalue $s_{j}$ ; in other words
there is always a single Jordan block for each distinct eigenvalue;
the cyclicity property is an integral part of the construction yielding
simpler Jordan forms;
\item if a non-zero $\zeta\in\mathrm{spec}\,\left(\mathscr{H}_{a}\right)$
is real or pure imaginary then the Jordan form $\mathscr{J}_{a}$
of the Hamiltonian matrix $\mathscr{H}_{a}$ has two Jordan blocks
$J_{n}\left(\zeta\right)$ and $J_{n}\left(-\zeta\right)$ of the
matching size $n$ where $n$ is the multiplicity of $\zeta$ as the
a root of the polynomial $\chi\left(s\right)$.
\item if $\zeta\in\mathrm{spec}\,\left(\mathscr{H}_{a}\right)$ and $\zeta=a+b\mathrm{i}$
with $a\neq0$ and $b\neq0$ then the Jordan form $\mathscr{J}_{a}$
of the Hamiltonian matrix $\mathscr{H}_{a}$ has four Jordan blocks
$J_{n}\left(\pm a\pm b\mathrm{i}\right)$ and $J_{n}\left(\pm a\mp b\mathrm{i}\right)$
of the matching size $n$ where $n$ is the multiplicity of $\zeta$
as the a root of the polynomial $\chi\left(s\right)$.
\end{itemize}
Making particular choices of $a_{k}$ for the principal Hamiltonian
$\mathcal{H}_{a}$ allows to achieve the desired Jordan forms. With
that in mind we introduce the following specific polynomials for real
numbers non-zero numbers $a$ and $b$:
\begin{equation}
\chi\left(s\right)=\left(s^{2}-a^{2}\right)^{n},\quad\left(s^{2}+b^{2}\right)^{n},\label{eq:XMX2a}
\end{equation}
\begin{equation}
\chi\left(s\right)=\left[s^{4}+2\left(b^{2}-a^{2}\right)s^{2}+\left(a^{2}+b^{2}\right)^{2}\right]^{n}.\label{eq:XMH2b}
\end{equation}
Notice that polynomials in equations (\ref{eq:XMX2a}) have respectively
two real roots $\pm a$ and two pure imaginary roots $\pm b\mathrm{i}$
of multiplicities $n$, whereas the polynomial in equation (\ref{eq:XMH2b})
has four roots $\pm a\pm b\mathrm{i}$ and $\pm a\mp b\mathrm{i}$
of multiplicities $n$. The Jordan forms $\mathscr{J}_{a}$ of system
matrices $\mathscr{H}_{a}$ associated with the polynomials in equations
(\ref{eq:XMX2a}) and (\ref{eq:XMH2b}) are respectively

\begin{gather}
\mathscr{J}_{a}=\left[\begin{array}{rr}
J_{n}\left(\zeta\right) & 0\\
0 & J_{n}\left(-\zeta\right)
\end{array}\right],\quad\mathscr{J}_{a}=\left[\begin{array}{rrrr}
J_{n}\left(\zeta\right) & 0 & 0 & 0\\
0 & J_{n}\left(-\zeta\right) & 0 & 0\\
0 & 0 & J_{n}\left(\bar{\zeta}\right) & 0\\
0 & 0 & 0 & J_{n}\left(-\bar{\zeta}\right)
\end{array}\right],\quad\zeta=a,b\mathrm{i},\label{eq:XMX2c}
\end{gather}
where Jordan block $J_{n}\left(\zeta\right)$ is defined by equation
(\ref{eq:Joblo1b}).

We summarize now the important points of the analysis in Sections
\ref{sec:pri-cir} and \ref{sec:spe-cir} in the following statement.
\begin{thm}[principal circuit]
 Suppose that the principal circuit depicted in Fig. \ref{fig:pri-cir}
has its element values defined by equations (\ref{eq:Lpricir1f})-(\ref{eq:Lpricir1g}).
Then the dynamics of the principal circuit is governed by the principal
Lagrangian $\mathcal{L}$ defined by equation (\ref{eq:Lpricir1a})
and the principal Hamiltonian $\mathcal{H}$ defined by\emph{ defined
by equations (\ref{eq:Hamapq1a}). The corresponding EL} equations
(\ref{eq:Lpricir1b}) and (\ref{eq:Lpricir1c}) represent the Kirchhoff
voltage law, whereas the Kirchhoff current is enforced by the selection
of $n$ involved f-loops and currents $\partial_{t}q_{k}$.

The relevant Hamiltonian matrix $\mathscr{H}$ is cyclic (non-derogatory),
and its characteristic polynomial $\chi\left(s\right)$ is defined
by equations (\ref{eq:Lpricir1d}). The Jordan form $\mathscr{J}$
of matrix $\mathscr{H}$ is completely determined by $\chi\left(s\right)$.
In particular, each distinct root $s_{j}$ of $\chi\left(s\right)$
of the multiplicity $n_{j}$ is represented in $\mathscr{J}$ by the
single Jordan block $J_{n_{j}}\left(s_{j}\right)$ of the matching
size $n_{j}$.

For particular choices of the monic polynomial $\chi\left(s\right)$
as described in equations (\ref{eq:XMX2a}) and (\ref{eq:XMH2b})
one obtains circuits associated with the Jordan forms represented
respectively in equations (\ref{eq:XMX2c}).
\end{thm}

\section{The principal Hamiltonian and Lagrangian\label{sec:priHam}}

Suppose that the system configuration is described by time-dependent
$n$-dimensional vector-column $q$ and its dynamics is governed by
a Hamiltonian $\mathcal{H}=\mathcal{H}\left(p,q\right)$ where $p$
is the system momentum which is an $n$-dimensional vector-column
just as the configuration vector $q$. Suppose now that the Hamiltonian
$\mathcal{H}$ is defined by equations (\ref{eq:Hamapq1a}). To present
the system information in a compact matrix form we recast the representation
of the Hamiltonian (\ref{eq:Hamapq1a}) as
\begin{gather}
\mathcal{H}_{a}=\frac{1}{2}X^{\mathrm{T}}M_{\mathrm{H}}X,\quad X=\left[\begin{array}{l}
q\\
p
\end{array}\right],\quad M_{\mathrm{H}}=\left[\begin{array}{rr}
-\pi_{n} & K_{n}\\
K_{n}^{\mathrm{T}} & D_{a}
\end{array}\right],\label{eq:Hamapq1b}
\end{gather}
where $D_{a}$ and $\pi_{n}$ are diagonal $n\times n$ matrices defined
by

\begin{equation}
D_{a}=\left[\begin{array}{ccccc}
a_{0} & 0 & \cdots & 0 & 0\\
0 & -a_{1} & 0 & \cdots & 0\\
0 & 0 & \ddots & \cdots & \vdots\\
\vdots & \vdots & \ddots & \left(-1\right)^{n-3}a_{n-2} & 0\\
0 & 0 & \cdots & 0 & \left(-1\right)^{n-2}a_{n-1}
\end{array}\right],\quad\pi_{n}=\left[\begin{array}{ccccc}
0 & 0 & \cdots & 0 & 0\\
0 & 0 & 0 & \cdots & 0\\
0 & 0 & \ddots & \cdots & \vdots\\
\vdots & \vdots & \ddots & 0 & 0\\
0 & 0 & \cdots & 0 & 1
\end{array}\right],\label{eq:Hamapq1c}
\end{equation}
and $K_{n}$ is $n\times n$ nilpotent matrix defined by
\begin{equation}
K_{n}=\left[\begin{array}{ccccc}
0 & 1 & \cdots & 0 & 0\\
0 & 0 & 1 & \cdots & 0\\
0 & 0 & \ddots & \cdots & \vdots\\
\vdots & \vdots & \ddots & 0 & 1\\
0 & 0 & \cdots & 0 & 0
\end{array}\right].\label{eq:Hamapq1dn}
\end{equation}
We also make use of the Jordan block $J_{n}\left(\zeta\right)$ of
the size $n$ defined by
\begin{gather}
J_{n}\left(\zeta\right)=\zeta\mathbb{I}_{n}+K_{n}=\left[\begin{array}{ccccc}
\zeta & 1 & \cdots & 0 & 0\\
0 & \zeta & 1 & \cdots & 0\\
0 & 0 & \ddots & \cdots & \vdots\\
\vdots & \vdots & \ddots & \zeta & 1\\
0 & 0 & \cdots & 0 & \zeta
\end{array}\right],\label{eq:Hamapq1e}
\end{gather}
where $\mathbb{I}_{n}$ is $n\times n$ identity matrix.

The evolution equations for the principal Hamiltonian $\mathcal{H}_{a}$,
defined be equations (\ref{eq:Hamapq1b}), are
\begin{gather}
\partial_{t}X=\mathscr{H}_{a}X,\quad\mathscr{H}_{a}=\mathbb{J}M_{\mathrm{H}},\quad\mathbb{J}=\left[\begin{array}{rr}
0 & \mathbf{\mathbb{I}}\\
-\mathbf{\mathbb{I}} & 0
\end{array}\right],\label{eq:Hamapq1f}
\end{gather}
where the system state vector $X$ and matrix $M_{\mathrm{H}}$ are
defined by (\ref{eq:Hamapq1a}), and consequently

\begin{equation}
\mathscr{H}_{a}=\mathbb{J}M_{\mathrm{H}}=\left[\begin{array}{rr}
K_{n}^{\mathrm{T}} & -D_{a}\\
\pi_{n} & -K_{n}
\end{array}\right].\label{eq:Hamapq1g}
\end{equation}
With an eigenvalue problem in mind we introduce matrix
\begin{equation}
s\mathbb{I}_{2n}-\mathscr{H}_{a}=\left[\begin{array}{rr}
-J_{n}^{\mathrm{T}}\left(-s\right) & -D_{a}\\
-\pi_{n} & J_{n}\left(s\right)
\end{array}\right],\label{eq:Hamapq1z}
\end{equation}
and find then the corresponding characteristic function is equal to
\begin{gather}
\chi_{a}\left(s\right)=\det\left\{ s\mathbb{I}_{2n}-\mathscr{H}_{a}\right\} =s^{2n}+\left(-1\right)^{n}\sum_{k=1}^{n}a_{n-k}s^{2\left(n-k\right)}=\label{eq:Hamapq2a}\\
=s^{2n}+\left(-1\right)^{n}\left(a_{n-1}s^{2\left(n-1\right)}+a_{n-2}s^{2\left(n-2\right)}+\cdots a_{0}\right).\nonumber 
\end{gather}
To see that representation (\ref{eq:Hamapq2a}) for $\chi_{a}\left(s\right)$
holds we apply formula (\ref{eq:Block1f}) to the right-hand side
of equation (\ref{eq:Hamapq1z}) and obtain
\begin{gather}
\chi_{a}\left(s\right)=\det\left\{ s\mathbb{I}_{2n}-\mathscr{H}_{a}\right\} =\det\left\{ J_{n}^{\mathrm{T}}\left(-s\right)\right\} \det\left\{ -J_{n}\left(s\right)+\pi_{n}\left[J_{n}^{-1}\left(-s\right)\right]^{\mathrm{T}}D_{a}\right\} .\label{eq:Hamapq2aa}
\end{gather}
We use then equations (\ref{eq:Hamapq1c}) and (\ref{eqJnqp1e}) to
evaluate the right-hand side of equation (\ref{eq:Hamapq2aa}) and
arrive at the formula (\ref{eq:Hamapq2a}).

We introduce now the so-called \emph{companion} to the polynomial
$\chi_{a}\left(s\right)$, see Section \ref{sec:co-mat}, which is
$2n\times2n$ matrix defined by
\begin{gather}
\mathscr{C}_{a}=\left[\begin{array}{rrrrrrr}
0 & 1 & 0 & \cdots & \cdots & 0 & 0\\
0 & 0 & 1 & 0 & 0 & 0 & 0\\
0 & 0 & 0 & \ddots & \ddots & 0 & 0\\
\vdots & \vdots & \vdots & \ddots & \ddots & \vdots & \vdots\\
0 & 0 & 0 & \ddots & 0 & 1 & 0\\
0 & 0 & 0 & 0 & 0 & 0 & 1\\
c_{0} & 0 & c_{1} & 0 & \cdots & c_{n-1} & 0
\end{array}\right],\quad c_{k}=\left(-1\right)^{n-1}a_{k},\quad0\leq k\leq n.\label{eq:Hamapq2b}
\end{gather}
Notice that the eigenvalue problem for the companion matrix $\mathscr{C}_{a}$
has the following explicit form solution, see Section \ref{sec:co-mat}, 

\begin{gather}
\mathscr{C}_{a}Y\left(s\right)=sY\left(s\right),\quad Y\left(s\right)=\left[\begin{array}{c}
1\\
s\\
s^{2}\\
\vdots\\
s^{2n-2}\\
s^{2n-1}
\end{array}\right],\quad Y\left[k\right]=s^{k-1},\quad1\leq k\leq2n,\quad\chi_{a}\left(s\right)=0,\label{eq:Hamapq2c}
\end{gather}
where evidently vector polynomial $Y\left(s\right)$ is uniquely determined
by the corresponding eigenvalue $s$. If all eigenvalues $s_{j}$,
$1\leq j\leq2n$ of the companion matrix $\mathscr{C}_{a}$ are different
the set of the corresponding eigenvectors $Y\left(s_{j}\right)$ form
a basis that diagonalize matrix $\mathscr{C}_{a}$. In the general
case we introduce an $2n\times2n$ matrix $\mathscr{Y}_{a}$ as the
generalized Vandermonde matrix defined by equations (\ref{eq:compas2d}),
(\ref{eq:compas2f}). Then according to Proposition \ref{prop:cycJ}
we have
\begin{equation}
\mathscr{C}_{a}=\mathscr{Y}_{a}\mathscr{J}_{a}\mathscr{Y}_{a}^{-1},\label{eq:Hamapq2ca}
\end{equation}
where $\mathscr{J}_{a}$ is the Jordan form of the companion matrix
$\mathscr{C}_{a}$. We refer to $\mathscr{Y}_{a}$ as \emph{Jordan
basis matrix} for matrix $\mathscr{C}_{a}$. In the special case of
distinct eigenvalues matrix $\mathscr{Y}_{a}$ turns into the standard
Vandermonde matrix defined by equation (\ref{eq:compas3c}), that
is a matrix formed by column-vectors $Y\left(s_{j}\right)$ as in
equation (\ref{eq:Hamapq2c}).

Notice also that it follows from equations (\ref{eq:Hamapq2a}) and
(\ref{eq:Hamapq2b}) that
\begin{gather}
\det\left\{ \mathscr{H}_{a}\right\} =\det\left\{ \mathscr{C}_{a}\right\} =\left(-1\right)^{n}a_{0}.\label{eq:Hamapq2d}
\end{gather}

Let us turn now to the eigenvalue problem for the system matrix $\mathscr{H}_{a}$.
In view of equation (\ref{eq:Hamapq1z}) an eigenvector $Z$ of $\mathscr{H}_{a}$
satisfies
\begin{gather}
\left[\begin{array}{rr}
-J_{n}^{\mathrm{T}}\left(-s\right) & -D_{a}\\
-\pi_{n} & J_{n}\left(s\right)
\end{array}\right]Z\left(s\right)=0,\quad Z\left(s\right)=\left[\begin{array}{r}
q\left(s\right)\\
p\left(s\right)
\end{array}\right],\label{eq:Jnqp1a}
\end{gather}
or equivalently
\begin{gather}
J_{n}^{\mathrm{T}}\left(-s\right)q\left(s\right)+D_{a}p\left(s\right)=0,\label{eq:Jnqp1b}\\
-\pi_{n}q\left(s\right)+J_{n}\left(s\right)p\left(s\right)=0.\label{eq:Jnqp1c}
\end{gather}
Notice first that $\pi_{n}q\left(s\right)\neq0$ otherwise we consequently
obtain $p\left(s\right)=0$ from equation (\ref{eq:Jnqp1c}) and then
$q\left(s\right)=0$ from equation (\ref{eq:Jnqp1b}) implying $Z\left(s\right)=0$
contradicting that $Z\left(s\right)$ is an eigenvector. Using that
we normalize $q\left(s\right)$ by the following assumption
\begin{gather}
\pi_{n}q\left(s\right)=e_{n}\left[e_{n}^{T}q\left(s\right)\right]=s^{2n}e_{n},\text{ or equivalently }q_{n}\left(s\right)=e_{n}^{T}q\left(s\right)=s^{2n}.\label{eq:Jnqp1d}
\end{gather}
This particular choice of normalization makes the components of eigenvectors
to be polynomials of $s$ rather then rational functions. Combing
the explicit formula
\begin{equation}
\left[J_{n}\left(s\right)\right]^{-1}=\left[\begin{array}{ccccc}
\frac{1}{s} & -\frac{1}{s^{2}} & \frac{1}{s^{3}} & \cdots & \frac{\left(-1\right)^{n-1}}{s^{n}}\\
0 & \frac{1}{s} & -\frac{1}{s^{2}} & \cdots & \frac{\left(-1\right)^{n-2}}{s^{n-1}}\\
0 & 0 & \ddots & \cdots & \vdots\\
\vdots & \vdots & \ddots & \frac{1}{s} & -\frac{1}{s^{2}}\\
0 & 0 & \cdots & 0 & \frac{1}{s}
\end{array}\right],\label{eqJnqp1e}
\end{equation}
with equations (\ref{eq:Jnqp1c}) and (\ref{eqJnqp1e}) we readily
obtain
\begin{equation}
p\left(s\right)=s^{2n}\left(J_{n}\left(s\right)\right)^{-1}e_{n}=\left[\begin{array}{c}
\left(-1\right)^{n-1}s^{n}\\
\left(-1\right)^{n}s^{n+1}\\
\vdots\\
-s^{2n-2}\\
s^{2n-1}
\end{array}\right].\label{eq:Jnqp1f}
\end{equation}
Then plugging expression (\ref{eq:Jnqp1f}) into equation (\ref{eq:Jnqp1b})
yields
\begin{gather}
q\left(s\right)=\left[-J_{n}^{\mathrm{T}}\left(-s\right)\right]^{-1}D_{a}p==s^{2n}\left[-J_{n}^{\mathrm{T}}\left(-s\right)\right]^{-1}D_{a}\left(J_{n}\left(s\right)\right)^{-1}e_{n}.\label{eq:Jnqp1g}
\end{gather}
Using equations (\ref{eqJnqp1e}), (\ref{eq:Jnqp1f}) and (\ref{eq:Jnqp1g})
we obtain the following expressions for the components of $q\left(s\right)$
and $p\left(s\right)$
\begin{gather}
q_{j}\left(s\right)=\left(-1\right)^{n-1}\sum_{k=1}^{j}a_{k-1}s^{n-j+2\left(k-1\right)},\quad p_{j}\left(s\right)=\left(-1\right)^{n+j}s^{n+j-1},\quad1\leq j\leq n.\label{eq:Jnqp2a}
\end{gather}
Consequently we get the following representation for eigenvector $Z\left(s\right)$
\begin{gather}
\mathscr{H}_{a}Z\left(s\right)=sZ\left(s\right),\quad Z\left(s\right)=\left[\begin{array}{r}
q\left(s\right)\\
p\left(s\right)
\end{array}\right]=\left[\begin{array}{c}
q_{1}\left(s\right)\\
\vdots\\
q_{n}\left(s\right)\\
p_{1}\left(s\right)\\
\vdots\\
p_{n}\left(s\right)
\end{array}\right]=\sum_{k=0}^{2n-1}Z_{k}s^{k},\label{eq:Jnqp2c}
\end{gather}
where $q\left(s\right)$ and $p\left(s\right)$ are defined by equations
(\ref{eq:Jnqp2a}).

Notice that according to equations (\ref{eq:Jnqp2a}) and (\ref{eq:Jnqp2c})
the eigenvector $Z\left(s\right)$ of the system matrix $\mathscr{H}_{a}$
is uniquely determined by the corresponding eigenvalue $s$. Evidently,
$Z\left(s\right)$ is a vector polynomial of $s$ with vector coefficients
$Z_{k}$ which are determined by expressions (\ref{eq:Jnqp2a}) for
vectors $q\left(s\right)$ and $p\left(s\right)$. 

Comparing equations (\ref{eq:Jnqp2c}) and (\ref{eq:Hamapq2c}) we
arrive with the following relationship between eigenvectors $Z\left(s\right)$
and $Y\left(s\right)$
\begin{gather}
Z\left(s\right)=T_{a}Y\left(s\right),\quad T_{a}=\left[Z_{0}|Z_{1}|\ldots|Z_{2n-1}\right],\quad\mathrm{col}\,\left(T_{a},k\right)=Z_{k-1},\quad1\leq k\leq2n-1.\label{eq:Jnqp2d}
\end{gather}
Notice that $2n\times2n$ matrix $T_{a}$ in equations (\ref{eq:Hamapq2c})
is defined by its columns which are the vector coefficients $Z_{k}$
of the vector polynomial $Z\left(s\right)$. Just as the system matrix
$\mathscr{H}_{a}$ and the companion matrix $\mathscr{C}_{a}$ matrix
$T_{a}$ is completely defined by the system parameters $a_{k}$ and
hence by the polynomial $\chi_{a}\left(s\right)$. An analysis show
that $T_{a}$ is $2\times2$ upper triangular block matrix, with blocks
of the dimension $n\times n$, and based on that one can establish
that
\begin{equation}
\det\left\{ T_{a}\right\} =a_{0}^{n}.\label{eq:Jnqp2e}
\end{equation}
The significance of matrix $T_{a}$ is that it provides for the similarity
relation between between the system matrix $\mathscr{H}_{a}$ and
its companion matrix $\mathscr{C}_{a}$, that is
\begin{equation}
\mathscr{C}_{a}=T_{a}^{-1}\mathscr{H}_{a}T_{a}.\label{eq:Jnqp2f}
\end{equation}
Equations (\ref{eq:Hpqchi3c})-(\ref{eq:Hpqchi3e}) and (\ref{eq:Hpqchi4c})-(\ref{eq:Hpqchi4e})
show examples of matrices $\mathscr{H}_{a}$, $\mathscr{C}_{a}$ and
$T_{a}$ for the cases $n=3,4$. 

Notice then if we introduce $2n\times2n$ matrix
\begin{equation}
\mathscr{Z}_{a}=T_{a}\mathscr{Y}_{a}\label{eq:Jnqp2g}
\end{equation}
use it in combination with equations (\ref{eq:Hamapq2ca}) we obtain
\begin{equation}
\mathscr{H}_{a}=\mathscr{Z}_{a}\mathscr{J}_{a}\mathscr{Z}_{a}^{-1},\label{eq:Jnqp2h}
\end{equation}
where $\mathscr{J}_{a}$ is the Jordan form of the companion matrix
$\mathscr{C}_{a}$ and hence of the system matrix $\mathscr{H}_{a}$
as well. We refer to $\mathscr{Z}_{a}$ as \emph{Jordan basis matrix}
for matrix $\mathscr{H}_{a}$.

The \emph{principal Lagrangian} $\mathcal{L}_{a}$ obtained from the
principal Hamiltonian $\mathcal{H}_{a}$ by the Legendre transformation
is

\begin{gather}
\mathcal{L}_{a}=\sum_{k=1}^{n}\frac{\left(-1\right)^{k-1}}{a_{k-1}}v_{k+1}q_{k}+\frac{1}{2}\sum_{k=1}^{n}\frac{\left(-1\right)^{k-1}}{a_{k-1}}v_{k}^{2}+\frac{1}{2}\sum_{k=1}^{n}\frac{\left(-1\right)^{k}}{a_{k}}q_{k}^{2}+\frac{1}{2}q_{n}^{2},\label{eq:Lagham1a}\\
v_{k}=\partial_{t}q_{k},\quad1\leq k\leq n.\nonumber 
\end{gather}
An equivalent to $\mathcal{L}_{a}$ version of it with the skew-symmetric
gyroscopic part is the following Lagrangian

\begin{gather}
\mathcal{L}_{a}^{\prime}=\frac{1}{2}\sum_{k=1}^{n}\frac{\left(-1\right)^{k-1}}{a_{k-1}}\left(v_{k+1}q_{k}-v_{k}q_{k+1}\right)+\frac{1}{2}\sum_{k=1}^{n}\frac{\left(-1\right)^{k-1}}{a_{k-1}}v_{k}^{2}+\frac{1}{2}\sum_{k=1}^{n}\frac{\left(-1\right)^{k}}{a_{k}}q_{k}^{2}+\frac{1}{2}q_{n}^{2},\label{eq:Lagham1b}\\
v_{k}=\partial_{t}q_{k},\quad1\leq k\leq n.\nonumber 
\end{gather}
The equivalency between two Lagrangians defined by equations (\ref{eq:Lagham1a})
and (\ref{eq:Lagham1b}) is understood as that the corresponding EL
equations are same, see Section \ref{subsec:Lag}.

\section{Examples of the significant matrices for the principal Hamiltonian\label{sec:expriHam}}

We show in this section explicit form of matrices $\mathscr{H}_{a}$,
$\mathscr{C}_{a}$ and $T_{a}$ related to the principal Hamiltonian
defined by equations (\ref{eq:Hamapq1a}), (\ref{eq:Hamapq1aa}).
The expressions of these matrices are somewhat different for even
and odd $n$, and with that in mind we consider two case of $n=3$
and $n=4$.

\subsection{The principal Hamiltonian and significant matrices for $n=3$\label{subsec:priHam3}}

The principal Hamiltonian and the corresponding characteristic polynomials
for $n=4$ are respectively
\begin{equation}
\mathcal{H}=\sum_{k=1}^{3}p_{k+1}q_{k}+\frac{1}{2}\sum_{k=1}^{3}\left(-1\right)^{k-1}a_{k-1}p_{k}^{2}+\frac{1}{2}q_{3}^{2},\quad\chi\left(s\right)=s^{6}-a_{2}s^{4}-a_{1}s^{2}-a_{0}.\label{eq:Hpqchi3a}
\end{equation}
The significant matrices in this case are as follows:
\begin{equation}
\mathscr{H}=\left[\begin{array}{rrrrrr}
0 & 0 & 0 & a_{0} & 0 & 0\\
1 & 0 & 0 & 0 & -a_{1} & 0\\
0 & 1 & 0 & 0 & 0 & a_{2}\\
0 & 0 & 0 & 0 & -1 & 0\\
0 & 0 & 0 & 0 & 0 & -1\\
0 & 0 & 1 & 0 & 0 & 0
\end{array}\right],\quad\mathscr{C}=\left[\begin{array}{rrrrrr}
0 & 1 & 0 & 0 & 0 & 0\\
0 & 0 & 1 & 0 & 0 & 0\\
0 & 0 & 0 & 1 & 0 & 0\\
0 & 0 & 0 & 0 & 1 & 0\\
0 & 0 & 0 & 0 & 0 & 1\\
a_{0} & 0 & a_{1} & 0 & a_{2} & 0
\end{array}\right],\label{eq:Hpqchi3c}
\end{equation}
\begin{equation}
T=\left[\begin{array}{rrrrrr}
0 & 0 & a_{0} & 0 & 0 & 0\\
0 & a_{0} & 0 & a_{1} & 0 & 0\\
a_{0} & 0 & a_{1} & 0 & a_{2} & 0\\
0 & 0 & 0 & 1 & 0 & 0\\
0 & 0 & 0 & 0 & -1 & 0\\
0 & 0 & 0 & 0 & 0 & 1
\end{array}\right].\label{eq:Hpqchi3e}
\end{equation}

\subsection{The principal Hamiltonian and significant matrices for $n=4$\label{subsec:priHam4}}

The principal Hamiltonian and the corresponding characteristic polynomials
for $n=4$ are respectively
\begin{equation}
\mathcal{H}=\sum_{k=1}^{4}p_{k+1}q_{k}+\frac{1}{2}\sum_{k=1}^{4}\left(-1\right)^{k-1}a_{k-1}p_{k}^{2}+\frac{1}{2}q_{4}^{2},\label{eq:Hpqchi4a}
\end{equation}
\begin{equation}
\chi\left(s\right)=s^{8}+a_{3}s^{6}+a_{2}s^{4}+a_{1}s^{2}+a_{0}.\label{eq:Hpqchi4b}
\end{equation}
The significant matrices in this case are as follows:

\begin{equation}
\mathscr{H}=\left[\begin{array}{rrrrrrrr}
0 & 0 & 0 & 0 & a_{0} & 0 & 0 & 0\\
1 & 0 & 0 & 0 & 0 & -a_{1} & 0 & 0\\
0 & 1 & 0 & 0 & 0 & 0 & a_{2} & 0\\
0 & 0 & 1 & 0 & 0 & 0 & 0 & -a_{3}\\
0 & 0 & 0 & 0 & 0 & -1 & 0 & 0\\
0 & 0 & 0 & 0 & 0 & 0 & -1 & 0\\
0 & 0 & 0 & 0 & 0 & 0 & 0 & -1\\
0 & 0 & 0 & 1 & 0 & 0 & 0 & 0
\end{array}\right],\quad\mathscr{C}=\left[\begin{array}{rrrrrrrr}
0 & 1 & 0 & 0 & 0 & 0 & 0 & 0\\
0 & 0 & 1 & 0 & 0 & 0 & 0 & 0\\
0 & 0 & 0 & 1 & 0 & 0 & 0 & 0\\
0 & 0 & 0 & 0 & 1 & 0 & 0 & 0\\
0 & 0 & 0 & 0 & 0 & 1 & 0 & 0\\
0 & 0 & 0 & 0 & 0 & 0 & 1 & 0\\
0 & 0 & 0 & 0 & 0 & 0 & 0 & 1\\
-a_{0} & 0 & -a_{1} & 0 & -a_{2} & 0 & -a_{3} & 0
\end{array}\right],\label{eq:Hpqchi4c}
\end{equation}

\begin{equation}
T=\left[\begin{array}{rrrrrrrr}
0 & 0 & 0 & -a_{0} & 0 & 0 & 0 & 0\\
0 & 0 & -a_{0} & 0 & -a_{1} & 0 & 0 & 0\\
0 & -a_{0} & 0 & -a_{1} & 0 & -a_{2} & 0 & 0\\
-a_{0} & 0 & -a_{1} & 0 & -a_{2} & 0 & -a_{3} & 0\\
0 & 0 & 0 & 0 & -1 & 0 & 0 & 0\\
0 & 0 & 0 & 0 & 0 & 1 & 0 & 0\\
0 & 0 & 0 & 0 & 0 & 0 & -1 & 0\\
0 & 0 & 0 & 0 & 0 & 0 & 0 & 1
\end{array}\right].\label{eq:Hpqchi4e}
\end{equation}

\section{Lagrangian and Hamiltonian structures for linear systems\label{sec:LagHam}}

We provide here basic facts on the Lagrangian and Hamiltonian structures
for linear systems.

\subsection{Lagrangian\label{subsec:Lag}}

Lagrangian $\mathcal{L}$ for a linear system is a quadratic function
(bilinear form) of the system state $Q=\left[q_{r}\right]_{r=1}^{n}$
(column vector) and its time derivatives $\partial_{t}Q$, that is
\begin{gather}
\mathcal{L}=\mathcal{L}\left(Q,\partial_{t}Q\right)=\frac{1}{2}\left[\begin{array}{l}
Q\\
\partial_{t}Q
\end{array}\right]^{\mathrm{T}}M_{\mathrm{L}}\left[\begin{array}{l}
Q\\
\partial_{t}Q
\end{array}\right],\quad M_{\mathrm{L}}=\left[\begin{array}{rr}
-\eta & \theta^{\mathrm{T}}\\
\theta & \alpha
\end{array}\right],\label{eq:dlag1}
\end{gather}
where $\mathrm{T}$ denotes the matrix transposition operation, and
$\alpha,\eta$ and $\theta$ are $n\times n$-matrices with real-valued
entries. In addition to that, we assume matrices $\alpha,\eta$ to
be symmetric, that is 
\begin{equation}
\alpha=\alpha^{\mathrm{T}},\qquad\eta=\eta^{\mathrm{T}}.\label{eq:dlag2}
\end{equation}
Consequently,
\begin{equation}
\mathcal{L}=\frac{1}{2}\partial_{t}Q^{\mathrm{T}}\alpha\partial_{t}Q+\partial_{t}Q^{\mathrm{T}}\theta Q-\frac{1}{2}Q^{\mathrm{T}}\eta Q.\label{eq:dlag3a}
\end{equation}
Then by Hamilton's principle, the system evolution is governed by
the EL equations 
\begin{equation}
\frac{d}{dt}\left(\frac{\partial\mathcal{L}}{\partial\partial_{t}Q}\right)-\frac{\partial\mathcal{L}}{\partial Q}=0,\label{eq:dlag4}
\end{equation}
which, in view of equation (\ref{eq:dlag3a}) for the Lagrangian $\mathcal{L}$,
turns into the following second-order vector ordinary differential
equation (ODE):
\begin{equation}
\alpha\partial_{t}^{2}Q+\left(\theta-\theta^{\mathrm{T}}\right)\partial_{t}Q+\eta Q=0.\label{eq:dlag5}
\end{equation}
Notice that matrix $\theta$ enters equation (\ref{eq:dlag5}) through
its skew-symmetric component $\frac{1}{2}\left(\theta-\theta^{\mathrm{T}}\right)$
justifying as a possibility to impose the skew-symmetry assumption
on $\theta$, that is
\begin{equation}
\theta^{\mathrm{T}}=-\theta.\label{eq:dlag3aa}
\end{equation}
Indeed, the symmetric part $\theta_{s}=\frac{1}{2}\left(\theta+\theta^{\mathrm{T}}\right)$
of the matrix $\theta$ is associated with a term to the Lagrangian
which can be recast as is the complete (total) derivative, namely
$\frac{1}{2}\partial_{t}\left(Q^{\mathrm{T}}\theta_{s}Q\right)$.
It is a well known fact that adding to a Lagrangian the complete (total)
derivative of a function of $Q$ does not alter the the EL equations.
Namely, the EL equations are invariant under the Lagrangian gauge
transform $\mathcal{L}\rightarrow\mathcal{L}+\partial_{t}F\left(q,t\right)$,
\cite[2.9, 2.10]{Scheck}, \cite[I.2]{LanLifM}.

Under the assumption (\ref{eq:dlag3aa}) equation (\ref{eq:dlag5})
turns into its version with the skew-symmetric $\theta$
\begin{equation}
\alpha\partial_{t}^{2}Q+2\theta\partial_{t}Q+\eta Q=0,\text{ if }\theta^{\mathrm{T}}=-\theta.\label{eq:dlag5aa}
\end{equation}
It turns out though that our our principal Lagrangian that corresponds
to the principal Hamiltonian by the Legendre transformation does not
have skew-symmetric $\theta$ satisfying (\ref{eq:dlag3aa}). For
this reason we don't impose the condition of skew-symmetry on $\theta$. 

The EL equations are the second order ODE. The standard way to reduce
them to the equivalent first order ODE yields

\begin{gather}
\partial_{t}Y=\mathscr{L}Y,\quad Y=\left[\begin{array}{r}
Q\\
\partial_{t}Q
\end{array}\right],\label{eq:YLY1a}
\end{gather}
where
\begin{gather}
\mathscr{L}=\left[\begin{array}{rr}
0 & \mathbb{I}\\
-\alpha^{-1}\eta & -\alpha^{-1}\left(\theta-\theta^{\mathrm{T}}\right)
\end{array}\right]=\left[\begin{array}{rr}
\mathbb{I} & 0\\
0 & -\alpha^{-1}
\end{array}\right]\left[\begin{array}{rr}
0 & \mathbb{I}\\
\eta & \left(\theta-\theta^{\mathrm{T}}\right)
\end{array}\right].\label{eq:YLY1b}
\end{gather}

With the spectral analysis of equation (\ref{eq:dlag5}) in mind we
can recast it as
\begin{equation}
A\left(\partial_{t}\right)Q=0,\quad A\left(s\right)=\alpha s^{2}+2\theta s+\eta.\label{eq:dlag5a}
\end{equation}
where evidently $A\left(s\right)$ is the $n\times n$-matrix polynomial.

\subsection{Hamiltonian\label{subsec:Ham}}

An alternative to equations (\ref{eq:YLY1a}) and (\ref{eq:YLY1b})
way to replace the second-order vector ODE (\ref{eq:dlag5}) with
the first-order one with the Hamilton equations associated with the
Hamiltonian $\mathcal{H}$ defined by
\begin{gather}
\mathcal{H}=\mathcal{H}\left(P,Q\right)=P^{\mathrm{T}}\partial_{t}Q-\mathcal{L}\left(Q,\partial_{t}Q\right),\quad P=\frac{\partial\mathcal{L}}{\partial\partial_{t}Q}=\alpha\partial_{t}Q+\theta Q.\label{eq:dlag6}
\end{gather}
Notice that the second equation in (\ref{eq:dlag6}) implies the following
relations between the velocity and momentum vectors:
\begin{equation}
\partial_{t}Q=\alpha^{-1}\left(P-\theta Q\right),\quad P=\alpha\partial_{t}Q+\theta Q.\label{eq:dlag6a}
\end{equation}
Consequently
\begin{gather}
\mathcal{H}\left(P,Q\right)=\frac{1}{2}\left[\left(P-\theta Q\right)^{T}\alpha^{-1}\left(P-\theta Q\right)+Q^{T}\eta Q\right]=\frac{1}{2}\partial_{t}Q^{T}\alpha\partial_{t}Q+\frac{1}{2}Q^{T}\eta Q.\label{eq:dlag7}
\end{gather}
Notice also that equations (\ref{eq:dlag6a}) imply
\begin{gather}
\left[\begin{array}{l}
Q\\
P
\end{array}\right]=\left[\begin{array}{lr}
\mathbb{I} & 0\\
\theta & \alpha
\end{array}\right]\left[\begin{array}{r}
Q\\
\partial_{t}Q
\end{array}\right],\quad\left[\begin{array}{r}
Q\\
\partial_{t}Q
\end{array}\right]=\left[\begin{array}{rr}
\mathbb{I} & 0\\
-\alpha^{-1} & \theta\alpha^{-1}
\end{array}\right]\left[\begin{array}{l}
Q\\
P
\end{array}\right].\label{eq:dlag7p}
\end{gather}
$\mathcal{H}$ can be interpreted as the system energy which is a
conserved quantity, that is
\begin{equation}
\partial_{t}\mathcal{H}\left(P,Q\right)=0.\label{eq:dlag7b}
\end{equation}
The function $\mathcal{H}\left(P,Q\right)$ defined by (\ref{eq:dlag7})
can be recast into the following form
\begin{equation}
\mathcal{H}\left(P,Q\right)=\frac{1}{2}\left[\begin{array}{r}
Q\\
P
\end{array}\right]^{\mathrm{T}}M_{\mathrm{H}}\left[\begin{array}{r}
Q\\
P
\end{array}\right],\label{eq:dlag8}
\end{equation}
where $M_{\mathrm{H}}$ is the $2n\times2n$ matrix having the block
form
\begin{gather}
M_{\mathrm{H}}=\left[\begin{array}{rr}
\theta^{\mathrm{T}}\alpha^{-1}\theta+\eta & -\theta^{\mathrm{T}}\alpha^{-1}\\
-\alpha^{-1}\theta & \alpha^{-1}
\end{array}\right]=\left[\begin{array}{rr}
\mathbb{I} & -\theta^{\mathrm{T}}\\
0 & \mathbb{I}
\end{array}\right]\left[\begin{array}{ll}
\eta & 0\\
0 & \alpha^{-1}
\end{array}\right]\left[\begin{array}{ll}
\mathbf{\mathbb{I}} & 0\\
-\theta & \mathbf{\mathbb{I}}
\end{array}\right],\label{eq:dlag7a}
\end{gather}
where $\mathbb{I}$ is the identity $n\times n$-matrix. The Hamiltonian
form of the Euler-Lagrange equation (\ref{eq:dlag4}) reads
\begin{gather}
\partial_{t}u=\mathbb{J}M_{\mathrm{H}}u,\quad u=\left[\begin{array}{r}
Q\\
P
\end{array}\right],\quad\mathbb{J}=\left[\begin{array}{rr}
0 & \mathbf{\mathbb{I}}\\
-\mathbf{\mathbb{I}} & 0
\end{array}\right].\label{eq:dlag9}
\end{gather}
Matrix $\mathbb{J}$ defined in equations (\ref{eq:dlag9}) is called
\emph{unit imaginary matrix} and it satisfies \cite[3.1]{BernM}
\begin{equation}
\mathbb{J}=-\mathbb{J}^{\mathrm{T}}=-\mathbb{J}^{-1}.\label{eq:dlag9J}
\end{equation}
Notice that in view of equations (\ref{eq:dlag7a}), (\ref{eq:dlag9})
we have
\begin{gather}
\mathbb{J}M_{\mathrm{H}}==\left[\begin{array}{rr}
-\alpha^{-1}\theta & \alpha^{-1}\\
-\theta^{\mathrm{T}}\alpha^{-1}\theta-\eta & \theta^{\mathrm{T}}\alpha^{-1}
\end{array}\right].\label{eq:dlag9M}
\end{gather}
Then the corresponding to Hamilton vector equation (\ref{eq:dlag9})
matrix similar to the companion polynomial matrix $\mathsf{C}_{A}\left(s\right)=s\mathsf{B}-\mathsf{A}$
in (\ref{eq:CBA1a}) is
\begin{gather}
\mathsf{C}\left(\mathbb{J}M_{\mathrm{H}};s\right)=s\left[\begin{array}{lr}
\mathbf{\mathbb{I}} & 0\\
0 & \mathbf{\mathbb{I}}
\end{array}\right]-\mathbb{J}M_{\mathrm{H}}==\left[\begin{array}{rr}
s+\alpha^{-1}\theta & -\alpha^{-1}\\
\theta^{\mathrm{T}}\alpha^{-1}\theta+\eta & s-\theta^{\mathrm{T}}\alpha^{-1}
\end{array}\right].\label{eq:dlag9C}
\end{gather}
Let us introduce matrix
\begin{equation}
\mathscr{H}=\mathbb{J}M_{\mathrm{H}}.\label{eq:MJHam1a}
\end{equation}
Notice that in view of equations (\ref{eq:dlag7a}), (\ref{eq:dlag9J})
we have $M_{\mathrm{H}}^{\mathrm{T}}=M_{\mathrm{H}}$ and
\begin{gather}
\left[\mathscr{H}\right]^{\mathrm{T}}=-M_{\mathrm{H}}\mathbb{J}=-\mathbb{J}\left[-\mathbb{J}M_{\mathrm{H}}\right]\mathbb{J}=-\mathbb{J}^{-1}\left[\mathscr{H}\right]\mathbb{J},\label{eq:MJHam1b}
\end{gather}
implying that the transposed to $\mathscr{H}$ matrix $\left[\mathscr{H}\right]^{\mathrm{T}}$
is similar to $-\mathscr{H}$.

\subsection{Relationship between the Lagrangian and Hamiltonian\label{subsec:Lag-to-Ham}}

Notice that under assumption that $\alpha^{-1}$ exists according
to equations (\ref{eq:dlag1}) and (\ref{eq:dlag7a}) we have
\begin{gather}
M_{\mathrm{L}}=\left[\begin{array}{rr}
-\eta & \theta^{\mathrm{T}}\\
\theta & \alpha
\end{array}\right],\quad M_{\mathrm{H}}=\left[\begin{array}{rr}
\eta_{\mathrm{H}} & \theta_{\mathrm{H}}^{\mathrm{T}}\\
\theta_{\mathrm{H}} & \alpha_{\mathrm{H}}
\end{array}\right]=\left[\begin{array}{rr}
\theta^{\mathrm{T}}\alpha^{-1}\theta+\eta & -\theta^{\mathrm{T}}\alpha^{-1}\\
-\alpha^{-1}\theta & \alpha^{-1}
\end{array}\right],\label{eq:LaH1a}
\end{gather}
implying
\begin{gather}
\alpha_{\mathrm{H}}=\alpha^{-1},\quad\theta_{\mathrm{H}}=-\alpha^{-1}\theta,\quad\eta_{\mathrm{H}}=\theta^{\mathrm{T}}\alpha^{-1}\theta+\eta,\label{eq:LaH1b}
\end{gather}
 and
\begin{gather}
\alpha=\alpha_{\mathrm{H}}^{-1},\quad\theta=-\alpha_{\mathrm{H}}^{-1}\theta_{\mathrm{H}},\quad\eta=\eta_{\mathrm{H}}-\theta_{\mathrm{H}}^{\mathrm{T}}\alpha_{\mathrm{H}}^{-1}\theta_{\mathrm{H}}.\label{eq:LaH1c}
\end{gather}

\subsection{Lagrangian and Hamiltonian for higher order ODEs\label{subsec:Lag-Ham}}

If the Lagrangian $\mathcal{L}$ depends on higher order derivatives
as in
\[
\mathcal{L}=\frac{1}{2}\left[x_{n}^{2}+\sum_{m=0}^{n-1}a_{m}x_{m}^{2}\right],\quad x_{m}=\partial_{t}^{m}x,
\]
then the corresponding equations for its extremals are \cite[1.2.3, 3.1.4]{ArnGiv}
\[
\partial_{t}^{2n}x+\sum_{m=0}^{n-1}\left(-1\right)^{n-m}a_{m}\partial_{t}^{2m}x=0.
\]

\subsection{Positive energy case\label{subsec:pos-Ham}}

The main point of this section is that in the case when the energy
is non-negative, that is $\mathcal{H}\left(P,Q\right)\geq0$, then
the system spectral properties are ultimately determined by a self-adjoint,
and hence diagonalizable, operator $\Omega$ defined by equations
(\ref{eq:radis8}). The argument is as follows, \cite{FigWel14}.
Suppose that
\begin{equation}
\alpha=\alpha^{\mathrm{T}}\geq0,\quad\eta=\eta^{\mathrm{T}}\geq0.\label{eq:dlag7ab}
\end{equation}
Then representations (\ref{eq:dlag8}), (\ref{eq:dlag7a}) combined
with the inequalities (\ref{eq:dlag2}) and (\ref{eq:dlag7ab}) imply
\begin{equation}
\mathcal{H}\left(P,Q\right)\geq0\text{ and }M_{\mathrm{H}}=M_{\mathrm{H}}^{\mathrm{T}}\geq0.\label{eq:dlag8a}
\end{equation}
Notice that matrix $M_{\mathrm{H}}$ can be recast as 
\begin{equation}
M_{\mathrm{H}}=K^{\mathrm{T}}K,\label{eq:radis5}
\end{equation}
where the matrix $K$ is the block matrix
\begin{gather}
K=\left[\begin{array}{lr}
K_{\mathrm{q}} & 0\\
0 & K_{\mathrm{p}}
\end{array}\right]\left[\begin{array}{rr}
\mathbf{\mathbf{\mathbb{I}}} & 0\\
-\theta & \mathbf{\mathbf{\mathbb{I}}}
\end{array}\right]=\left[\begin{array}{lr}
K_{\mathrm{q}} & 0\\
-K_{\mathrm{p}}\theta & K_{\mathrm{p}}
\end{array}\right],\quad K_{\mathrm{q}}=\sqrt{\eta},\quad K_{\mathrm{p}}=\sqrt{\alpha}^{-1},\label{eq:radis6}
\end{gather}
which manifestly takes into account the gyroscopic term $\theta$.
Here $\sqrt{\alpha}$ and $\sqrt{\eta}$ denote the unique positive
semidefinite square roots of the matrices $\alpha$ and $\eta$, respectively.
\ In particular, it follows from the properties (\ref{eq:dlag2})
and the proof of \cite[S VI.4, Theorem VI.9]{ReSi1} that $K_{\mathrm{p}}$,
$K_{\mathrm{q}}$ are $n\times n$ matrices with real-valued entries
with the properties 
\begin{equation}
K_{\mathrm{p}}=K_{\mathrm{p}}^{\mathrm{T}}>0,\text{ \ \ }K_{\mathrm{q}}=K_{\mathrm{q}}^{\mathrm{T}}\geq0.\label{eq:radis6b}
\end{equation}
If we introduce now the force variable
\begin{equation}
v=Ku.\label{eq:radis7}
\end{equation}
then the evolution equation (\ref{eq:dlag9}) can be recast into the
following form
\begin{gather}
\partial_{t}v=-\mathrm{i}\Omega v,\quad\Omega=\Omega^{\ast}=\mathrm{i}K\mathbb{J}K^{\mathrm{T}}=\left[\begin{array}{cr}
0 & \mathrm{i}\Phi\\
-\mathrm{i}\Phi^{\mathrm{T}} & \Omega_{\mathrm{p}}
\end{array}\right],\quad\Omega_{\mathrm{p}}=-\mathrm{i}2K_{\mathrm{p}}\theta K_{\mathrm{p}}^{\mathrm{T}},\quad\Phi=K_{\mathrm{q}}K_{\mathrm{p}}^{\mathrm{T}}.\label{eq:radis8}
\end{gather}
where $\Omega$ is evidently a self-adjoint operator.

\subsection{Symplectic and Hamiltonian matrices basics\label{subsec:ham-mat}}

Hamiltonian matrices arise naturally as the matrices governing the
evolution of Hamiltonian systems, see Section \ref{subsec:Ham}.

Let $\mathbb{J}\in\mathbb{R}^{2n\times2n}$ be unit imaginary matrix
defined by equations (\ref{eq:dlag9}). It satisfies the identities
(\ref{eq:dlag9J}).
\begin{defn}[Symplectic matrix]
 A matrix $T\in\mathbb{R}^{2n\times2n}$ is called \emph{symplectic}
if it satisfies the following identity \cite[3.1]{Mey}:
\begin{equation}
T^{\mathrm{T}}\mathbb{J}T=\mathbb{J},\quad\mathbb{J}=\left[\begin{array}{rr}
0 & \mathbf{\mathbb{I}}_{n}\\
-\mathbf{\mathbb{I}}_{n} & 0
\end{array}\right].\label{eq:sympmat1a}
\end{equation}
It readily follows from equations (\ref{eq:sympmat1a}) and (\ref{eq:dlag9J})
that symplectic matrix $T$ is nonsingular and
\[
T^{-1}=-\mathbb{J}T^{\mathrm{T}}\mathbb{J}.
\]
It is also evident that $T$ is symplectic if and only if matrices
$T^{-1}$ and $T^{\mathrm{T}}$ are symplectic.
\end{defn}

Evidently symplectic matrices in $\mathbb{R}^{2n\times2n}$ form a
group.
\begin{defn}[Hamiltonian matrix]
 A matrix $M\in\mathbb{R}^{2n\times2n}$ is called \emph{Hamiltonian}
(or infinitesimally symplectic) if it satisfies the following identity
\cite[3.1]{Mey}:
\begin{equation}
\mathbb{J}^{-1}M^{\mathrm{T}}\mathbb{J}=-M,\quad\mathbb{J}=\left[\begin{array}{rr}
0 & \mathbf{\mathbb{I}}_{n}\\
-\mathbf{\mathbb{I}}_{n} & 0
\end{array}\right],\label{eq:Hammat1a}
\end{equation}
The Hamiltonian matrix property (\ref{eq:Hammat1a}) is evidently
equivalent to the symmetry of the matrix $\mathbb{J}M$, that is,
\begin{equation}
\left(\mathbb{J}M\right)^{\mathrm{T}}=\mathbb{J}M.\label{eq:Hammat1as}
\end{equation}
In other words, a Hamiltonian matrix $\mathscr{A}$ is a matrix of
the form
\begin{equation}
\mathscr{A}=\mathbb{J}A,\quad A^{\mathrm{T}}=A.\label{eq:Hammat1aa}
\end{equation}
\end{defn}

Since the definition of Hamiltonian matrix involves a transposed matrix
the following general statement it is of importance to know that a
matrix over the field of complex numbers is always similar to its
transposed \cite[3.2.3]{HorJohn}.
\begin{prop}[Similarity of a matrix and its transposed]
\label{prop:simtrans-1} Let $A\in\mathbb{C}^{n\times n}$. There
is exists a nonsingular complex symmetric matrix $S$ such that $A^{\mathrm{T}}=SAS^{-1}.$
\end{prop}

The following statement provides different equivalent descriptions
of a Hamiltonian matrix \cite[3.1]{Mey}:
\begin{prop}[Hamiltonian matrix]
 The following are equivalent: (i) $M$ is Hamiltonian, (ii) $M=\mathbb{J}A$
where $A$ is symmetric, and (iii) $\mathbb{J}A$ is symmetric. Moreover,
if $M$ and $K$ are Hamiltonian, then so are $M^{\mathrm{T}}$, $\alpha M$,
$\alpha\in\mathbb{R}$, $M\pm K$, and $\left[M,K\right]\equiv MK-KM$.
\end{prop}

The following representation holds for a Hamiltonian matrix $\mathcal{A}$
\cite[3.1]{BernM}:
\begin{prop}[Hamiltonian matrix]
\label{prop:hamrep} A matrix $\mathscr{A}\in\mathbb{C}^{2n\times2n}$
is a Hamiltonian matrix if and only if there exist matrices $A,B,C\in\mathbb{F}^{n\times n}$
such $B$ and $C$ are symmetric and
\begin{equation}
\mathscr{A}=\left[\begin{array}{rr}
A & B\\
C & -A^{\mathrm{T}}
\end{array}\right],\quad B=B^{\mathrm{T}},\quad C=C^{\mathrm{T}}.\label{eq:Hammat1b}
\end{equation}
The set of all Hamiltonian matrices forms a Lie algebra.
\end{prop}

In fact, a matrix over the field of complex numbers is always similar
to its transposed \cite[3.2.3]{HorJohn}.
\begin{prop}[Similarity of a matrix and its transposed]
\label{prop:simtrans} Let $A\in\mathbb{C}^{n\times n}$. There is
exists a nonsingular complex symmetric matrix $S$ such that $A^{\mathrm{T}}=SAS^{-1}.$
\end{prop}

The proof of Proposition \ref{prop:simtrans} can be obtained from
the matrix similarity to its Jordan canonical form.

Important spectral properties of Hamiltonian matrices and their canonical
forms are studied in \cite[2.2]{ArnGiv}, \cite{LauMey}, \cite[3.3, 4.6, 4.7]{Mey}.
As to the more detailed spectral properties of Hamiltonian matrices
the following statements holds.
\begin{prop}[Jordan structure of a real Hamiltonian matrix]
\label{prop:char-ham} The characteristic polynomial of a real Hamiltonian
matrix is an even polynomial. Thus if $\zeta$ is an eigenvalue of
a Hamiltonian matrix, then $-\zeta$, $\bar{\zeta}$ and $-\bar{\zeta}$
are also its eigenvalues with the same multiplicity. The entire Jordan
block structure is same for $\zeta$,$-\zeta$, $\bar{\zeta}$ and
$-\bar{\zeta}$.
\end{prop}

\section{A Sketch of the Basics of Electric Networks\label{sec:e-net}}

For the sake of self-consistency, we provide in this section basic
information on the basics of the electric network theory and relevant
notations.

Electrical networks is a well established subject represented in many
monographs. We present here basic elements of the electrical network
theory following mostly to \cite[2]{BalBic}, \cite{Cau}, \cite{SesRee}.
The electrical network theory constructions are based on the graph
theory concepts of branches (edges), nodes (vertices) and their incidences.
This approach is efficient in loop (fundamental circuit) analysis
and the determination of independent variables for the Kirchhoff current
and voltage laws - the subjects relevant to our studies here.

We are particularly interested in conservative electrical network
which is a particular case of an electrical network composed of electric
elements of three types: capacitors, inductors and gyrators. We remind
that a capacitor or an inductor are the so-called two-terminal electric
elements whereas a gyrator is four-terminal electric element as discussed
below. We assume that capacitors and inductors can have positive or
negative respective capacitances and inductances.

\subsection{Circuit elements and their voltage-current relationships\label{subsec:cir-elem}}

The elementary electric network (circuit) elements of interest here
are a \emph{capacitor}, an \emph{inductor}, a \emph{resistor} and
a \emph{gyrator}, \cite[1.5, 2.6]{BalBic}, \cite[App.5.4]{Cau},
\cite[10]{Iza}. These elements are characterized by the relevant
\emph{voltage-current relationships}. These relationships for the
capacitor, inductor and resistor are respectively as follows \cite[1.5]{BalBic},
\cite[3-Circuit theory]{Rich}, \cite[1.3]{SesBab}:
\begin{equation}
I=C\partial_{t}V,\quad V=L\partial_{t}I,\quad V=RI,\label{eq:cirvc1a}
\end{equation}
where $I$ and $V$ are respectively the \emph{current} and the \emph{voltage},
and real $C$, $L$ and $R$ are called respectively the \emph{capacitance},
the \emph{inductance} and the \emph{resistance}. The voltage-current
relationship for the gyrator depicted in Fig. \ref{fig:cir-gyr} are
\begin{gather}
(a):\:\begin{bmatrix}V_{1}\\
V_{2}
\end{bmatrix}=\begin{bmatrix}-GI_{2}\\
GI_{1}
\end{bmatrix},\quad(b):\:\begin{bmatrix}V_{1}\\
V_{2}
\end{bmatrix}=\begin{bmatrix}GI_{2}\\
-GI_{1}
\end{bmatrix},\label{eq:crvc1b}
\end{gather}
where $I_{1},\:I_{2}$ and $V_{1},\:V_{2}$ are respectively the \emph{currents}
and the \emph{voltages}, and quantity $G$ is called the \emph{gyration
resistance.}

The common graphic representations of the network elements are depicted
in Figures \ref{fig:cir-CLR} and \ref{fig:cir-gyr}. The arrow next
to the symbol $G$ in Fig. \ref{fig:cir-gyr} shows the direction
of gyration.

The gyrator has the so-called inverting property as shown in Fig.
\ref{fig:cir-gyr-inv}, \cite[1.5]{BalBic}, \cite[10]{Iza}, \cite[29.1]{Dorf}.
Namely, when a capacitor or an inductor connected to the output port
of the gyrator it behaves as an inductor or capacitor respectively
with the following effective values
\begin{equation}
L_{\mathrm{ef}}=G^{2}C,\quad C_{\mathrm{ef}}=\frac{L}{G^{2}}.\label{eq:cirvc1d}
\end{equation}

Notice that the voltage-current relationships in the second equation
in (\ref{eq:crvc1b}) can be obtained from the first equation in (\ref{eq:crvc1b})
by substituting $-G$ for $G$. The gyrator is a device that accounts
for physical situations in which the reciprocity condition does not
hold. The voltage-current relationships in equations (\ref{eq:crvc1b})
show that the gyrator is a non-reciprocal circuit element. In fact,
it is antireciprocal. Notice, that the gyrator, like the ideal transformer,
is characterized by a single parameter $G$, which is the gyration
resistance. The arrows next to the symbol $G$ in Fig. \ref{fig:cir-gyr}(a)
and (b) show the direction of gyration.
\begin{figure}[h]
\centering{}\includegraphics[scale=0.25]{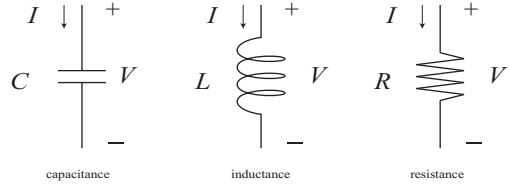}\caption{\label{fig:cir-CLR} Capacitance, inductance and resistance.}
\end{figure}
\begin{figure}[h]
\centering{}\includegraphics[scale=0.35]{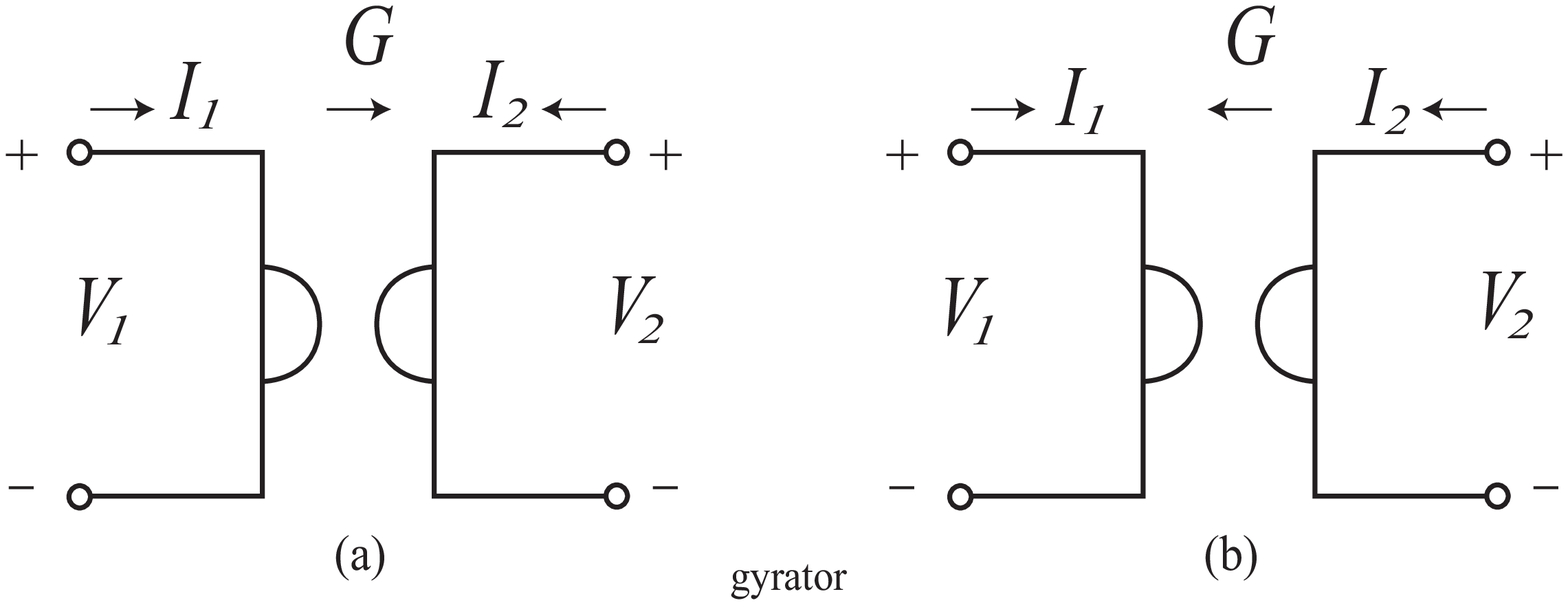}\caption{\label{fig:cir-gyr} Gyrator.}
\end{figure}
\begin{figure}[h]
\centering{}\includegraphics[scale=0.35]{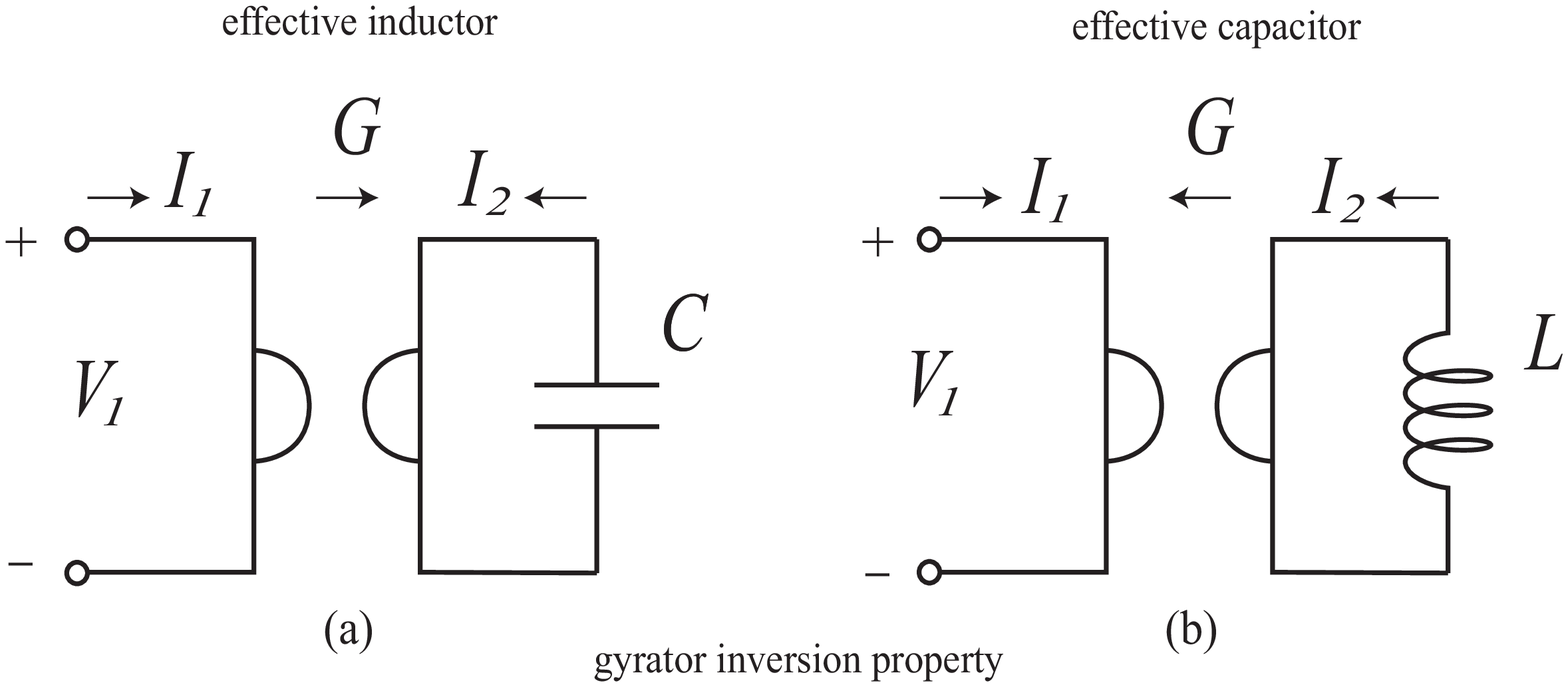}\caption{\label{fig:cir-gyr-inv} (a) Effective inductor; (b) effective capacitor.}
\end{figure}

Along with the voltage $V$ and the current $I$ variables we introduce
the \emph{charge }variable $Q$ and the \emph{momentum} (per unit
of charge) variable $P$ by the following formulas
\begin{gather}
Q\left(t\right)=\intop I\left(t\right)\,dt,\quad I\left(t\right)=\partial_{t}Q,\label{eq:cirvc2a}\\
P\left(t\right)=\intop V\left(t\right)\,dt,\quad V\left(t\right)=\partial_{t}P.\label{eq:cirvc2b}
\end{gather}
We introduce also the energy stored variable $W$, \cite[Circuit Theory]{Rich}.
Then the voltage-current relations (\ref{eq:cirvc1a}) and the stored
energy $W$ can be represented as follows:

\begin{gather}
\text{capacitor: }V=\frac{Q}{C},\quad I=\partial_{t}Q=C\partial_{t}V,\quad Q=CV=C\partial_{t}P;\label{eq:cirvc2c}\\
W=\frac{1}{2}VQ=\frac{Q^{2}}{2C}=\frac{CV^{2}}{2}=\frac{C\left(\partial_{t}P\right)^{2}}{2};\label{eq:cirvc2d}
\end{gather}

\begin{gather}
\text{inductor: }V=L\partial_{t}I,\quad P=LI=L\partial_{t}Q,\quad\partial_{t}Q=\frac{P}{L};\label{eq:cirvc2e}\\
W=\frac{PI}{2}=\frac{LI^{2}}{2}=\frac{L\left(\partial_{t}Q\right)^{2}}{2}=\frac{P^{2}}{2L};\label{eq:cirvc2f}
\end{gather}
\begin{equation}
\text{resistor: }V=RI,\quad P=RQ.\label{eq:cirvc2g}
\end{equation}

The Lagrangian associated with the network elements are as follows
\cite[9]{GantM}, \cite[3]{Rich}:
\begin{equation}
\text{capacitor: }\mathcal{L}=\frac{Q^{2}}{2C},\:\text{inductor: }\mathcal{L}=\frac{L\left(\partial_{t}Q\right)^{2}}{2},\label{eq:cirvc3a}
\end{equation}
\begin{gather}
\text{gyrator: }\mathcal{L}=GQ_{1}\partial_{t}Q_{2},\quad\mathcal{L}=\frac{G\left(Q_{1}\partial_{t}Q_{2}-Q{}_{2}\partial_{t}Q_{1}\right)}{2}.\label{eq:cirvc3b}
\end{gather}
Notice that the difference between two alternatives for the Lagrangian
in equations (\ref{eq:cirvc3b}) is $\frac{1}{2}G\partial_{t}\left(Q_{1}Q_{2}\right)$
which is evidently the complete time derivative. Consequently, the
EL equation are the same for both Lagrangians, see Section \ref{subsec:Lag}.

\subsection{Circuits of negative impedance, capacitance and inductance\label{subsec:neg-RCL}}

There are a number of physical devices that can provided for negative
capacitances and inductances needed for our circuits \cite[29]{Dorf}.
Following to \cite[10]{Iza} we show below circuits in Fig. \ref{fig:neg-imp}
that utilize operational amplifiers to achieve negative impedance,
capacitance and inductance respectively.
\begin{figure}[h]
\begin{centering}
\includegraphics[scale=0.45]{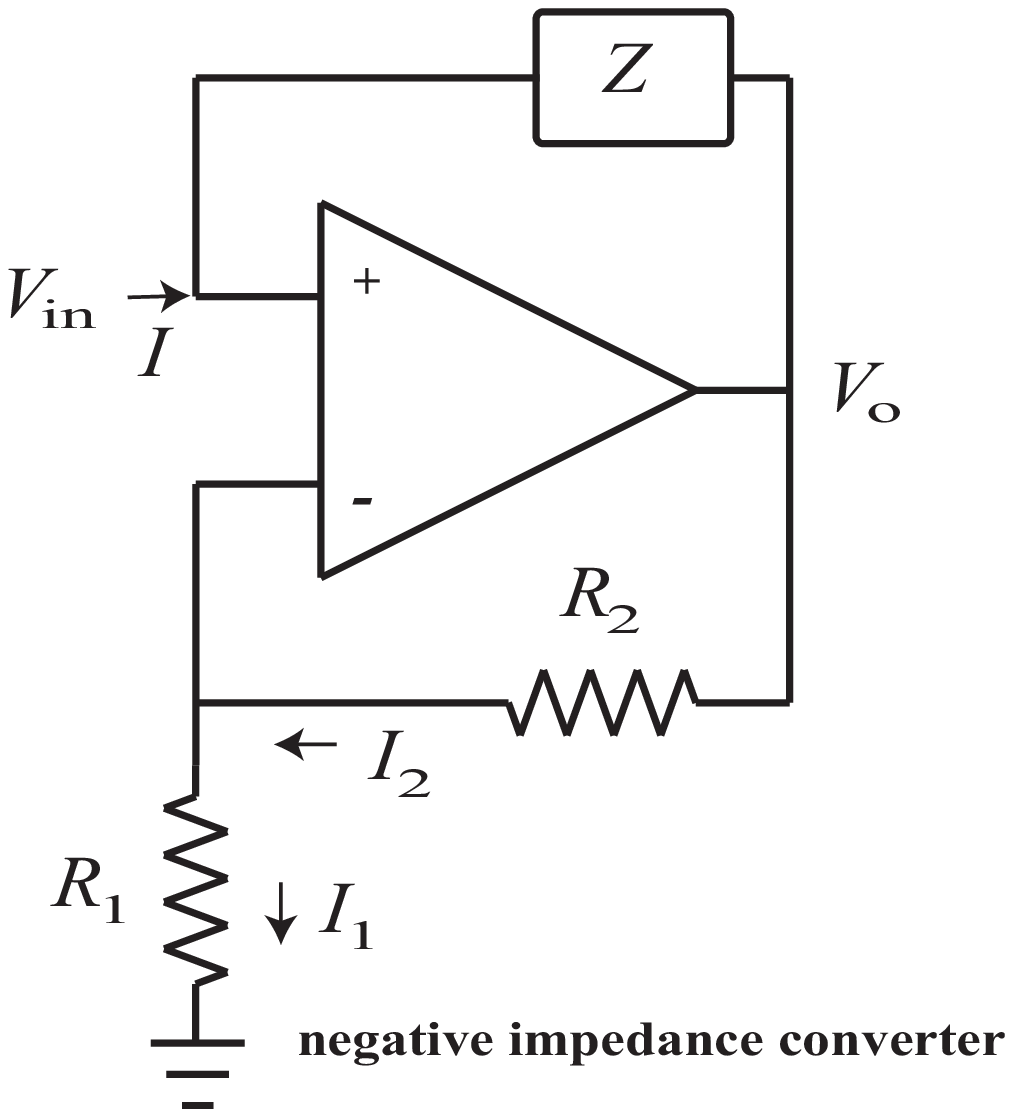}\hspace{1.5cm}\includegraphics[scale=0.45]{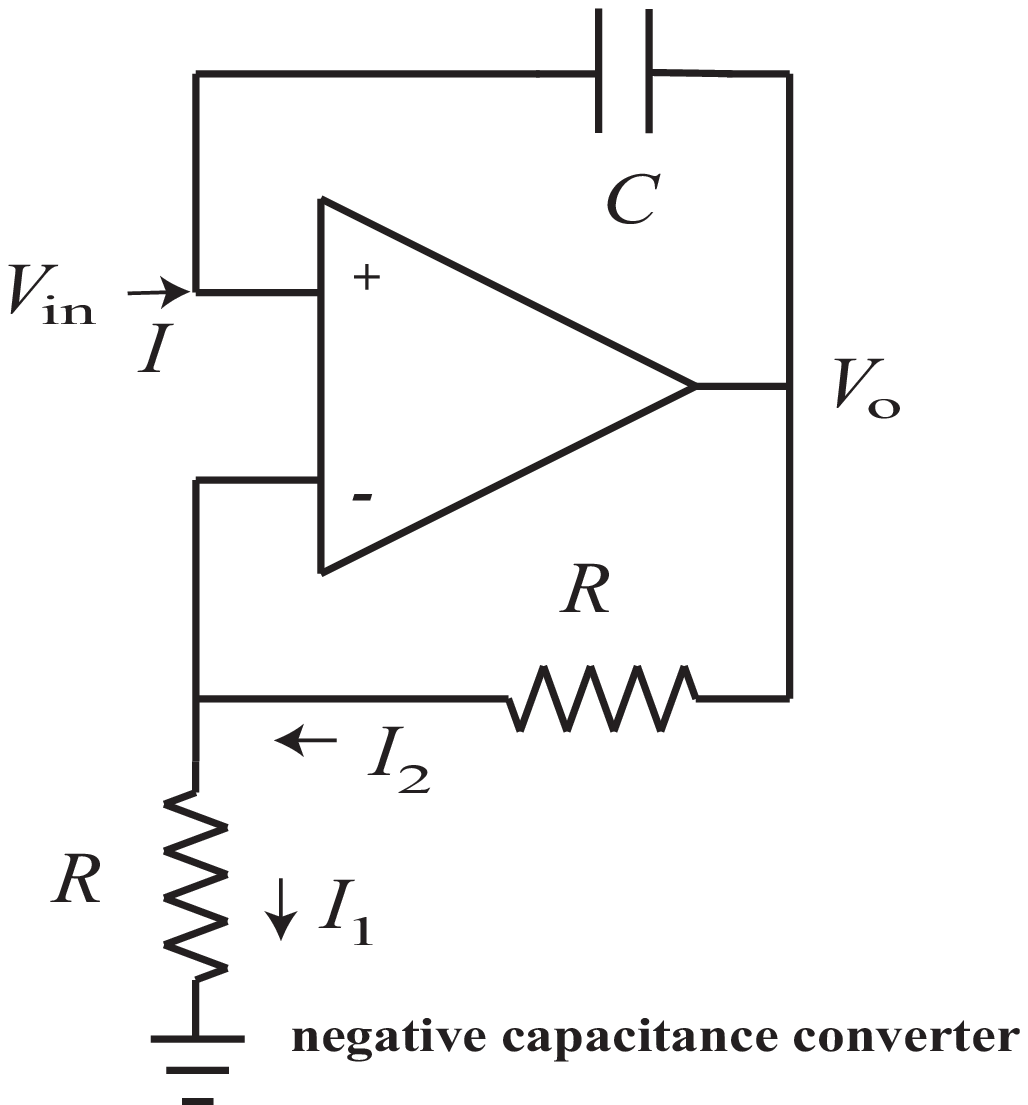}\hspace{1.5cm}\includegraphics[scale=0.45]{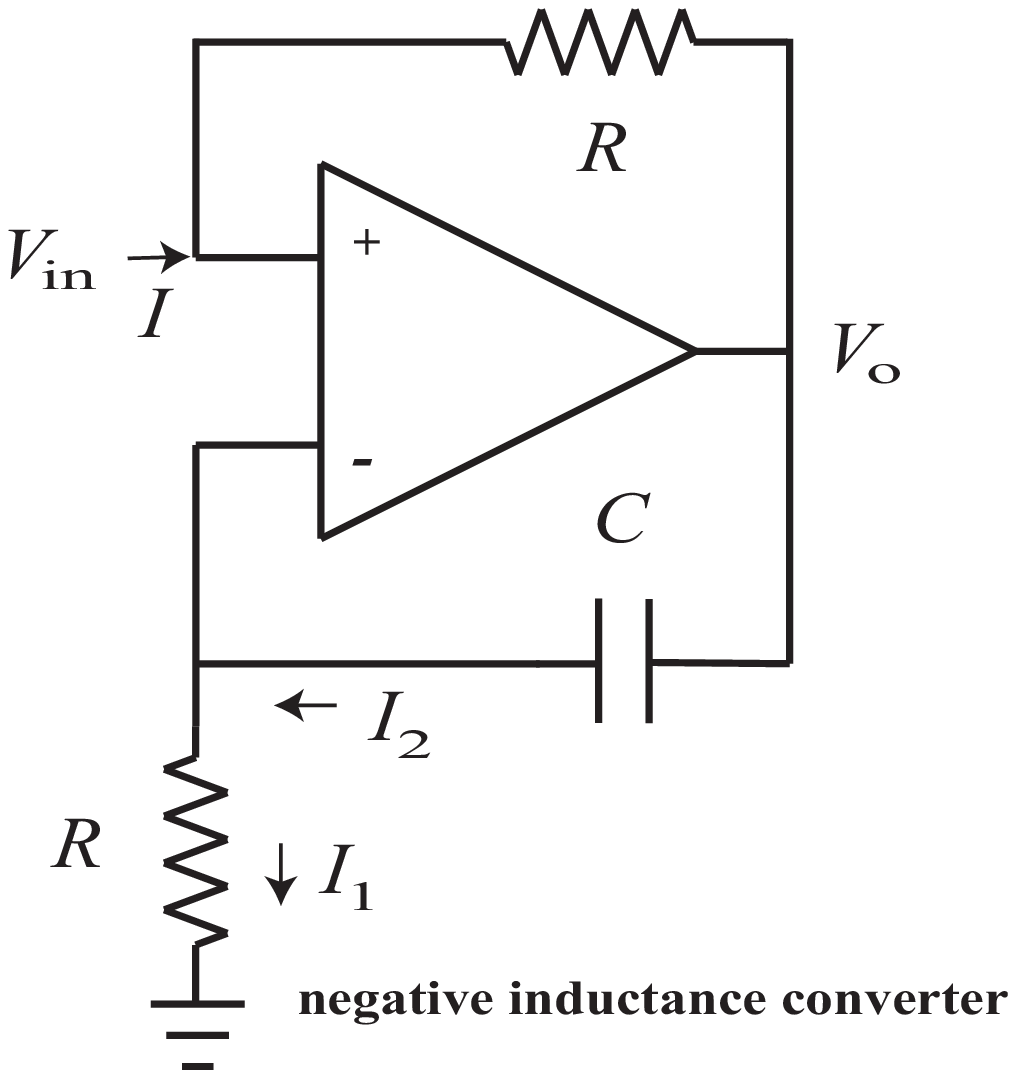}
\par\end{centering}
\begin{centering}
(a)\hspace*{5cm}(b)\hspace*{5cm}(c)
\par\end{centering}
\centering{}\caption{\label{fig:neg-imp} Operational-amplifier-based negative (a) impedance
converter; (b) capacitance converter; (c) inductance converter.}
\end{figure}
The currents and voltages for circuits depicted in Fig. \ref{fig:neg-imp}
are respectively as follows: (i) for negative impedance as in Fig.
\ref{fig:neg-imp}(a)

\begin{equation}
V_{\mathrm{in}}=-ZI,\quad V_{\mathrm{o}}=2V_{\mathrm{in}},\quad I_{1}=I_{2}=\frac{V_{\mathrm{in}}}{R};\label{eq:opamp1a}
\end{equation}
(ii) for negative capacitance as in Fig. \ref{fig:neg-imp}(b)

\begin{gather}
V_{\mathrm{in}}=Z_{\mathrm{in}}I,\quad Z_{\mathrm{in}}=-\frac{\mathrm{i}}{\omega C},\quad V_{\mathrm{o}}=2V_{\mathrm{in}},\quad I_{1}=I_{2}=\frac{V_{\mathrm{in}}}{R};\label{eq:opamp1b}
\end{gather}
(iii) for negative inductance as in Fig. \ref{fig:neg-imp}(c) 

\begin{gather}
V_{\mathrm{in}}=Z_{\mathrm{in}}I,\quad Z_{\mathrm{in}}=-\mathrm{i}\omega R^{2}C,\quad I_{1}=I_{2}=\frac{V_{\mathrm{in}}}{R},\quad V_{\mathrm{o}}=V_{\mathrm{in}}\left(1+\frac{1}{\mathrm{i}\omega RC}\right).\label{eq:opamp1c}
\end{gather}

\subsection{Topological aspects of the electric networks\label{subsec:net-top}}

We follow here mostly to \cite[2]{BalBic}. The purpose of this section
is to concisely describe and illustrate relevant concepts with understanding
that the precise description of all aspects of the concepts is available
in \cite[2]{BalBic}.

To describe topological (geometric) features of the electric network
we use the concept of \emph{linear graph} defined as a collection
of points, called \emph{nodes}, and line segments called branches,
the nodes being joined together by the branches as indicated in Fig.
\ref{fig:cir-gyr} (b). Branches whose ends fall on a node are said
to be incident at the node. For instance, Fig. \ref{fig:cir-gyr}
(b) branches \textbf{\textit{1}}, \textbf{\textit{2}}, \textbf{\textit{3}},
\textbf{\textit{4}} are incident at node 2. Each branch in Fig. \ref{fig:cir-gyr}
(b) carries an arrow indicating its orientation. A graph with oriented
branches is called an oriented graph. The elements of a network associated
with its graph have both a voltage and a current variable, each with
its own reference. In order to relate the orientation of the branches
of the graph to these references the convention is made that the voltage
and current of an element have the standard reference - voltage-reference
\textquotedbl plus\textquotedbl{} at the tail of the current-reference
arrow. The branch orientation of a graph is assumed to coincide with
the associated current reference as shown in Figures \ref{fig:cir-CLR}
and \ref{fig:cir-gyr}.

We denote the number of branches of the network by $N_{\mathrm{b}}\geq2$,
and the number of nodes by $N_{\mathrm{n}}\geq2$. 

A \emph{subgraph} is a subset of the branches and nodes of a graph.
The subgraph is said to be \emph{proper} if it consists of strictly
less than all the branches and nodes of the graph. A \emph{path} is
a particular subgraph consisting of an ordered sequence of branches
having the following properties:
\begin{enumerate}
\item At all but two of its nodes, called internal nodes, there are incident
exactly two branches of the subgraph. 
\item At each of the remaining two nodes, called the terminal nodes, there
is incident exactly one branch of the subgraph. 
\item No proper subgraph of this subgraph, having the same two terminal
nodes, has properties 1 and 2.
\end{enumerate}
A graph is called \emph{connected} if there exists at least one path
between any two nodes. We consider here only connected graphs such
as shown in Fig. \ref{fig:tree} (b).

A \emph{loop} (cycle) is a particular connected subgraph of a graph
such that at each of its nodes there are exactly two incident branches
of the subgraph. Consequently, if the two terminal nodes of a path
coincide we get a ``closed path'', that is a loop. In Fig. \ref{fig:tree}
(b) branches \textbf{\textit{7}}, \textbf{\textit{1}}, \textbf{\textit{3}},
\textbf{\textit{5}} together with nodes 1, 2, 3, and 4 form a loop.
We can specify a loop by an either the ordered list of the relevant
branched or the ordered list of the relevant nodes.

We remind that each branch of the network graph is associated with
two functions of time $t$: its current $I(t)$ and its voltage $V(t)$.
The set of these functions satisfy two Kirchhoff's laws, \cite[2.2]{BalBic},
\cite[2]{Cau}, \cite[Circuit Theory]{Rich}, \cite[1]{SesRee}. The
\emph{Kirchhoff current law} (KCL) states that in any electric network
the sum of all currents leaving any node equals zero at any instant
of time. The \emph{Kirchhoff voltage law} (KVL) states that in any
electric network, the sum of voltages of all branches forming any
loop equals zero at any instant of time. It turns out that the number
of independent KCL equations is $N_{\mathrm{n}}-1$ and the number
KVL equations is $N_{\mathrm{fl}}=N_{\mathrm{b}}-N_{\mathrm{n}}+1$
(the first Betti number \cite[2]{Cau}, \cite[2.3]{SesRee}).
\begin{figure}[th]
\begin{centering}
\includegraphics[clip,scale=0.5]{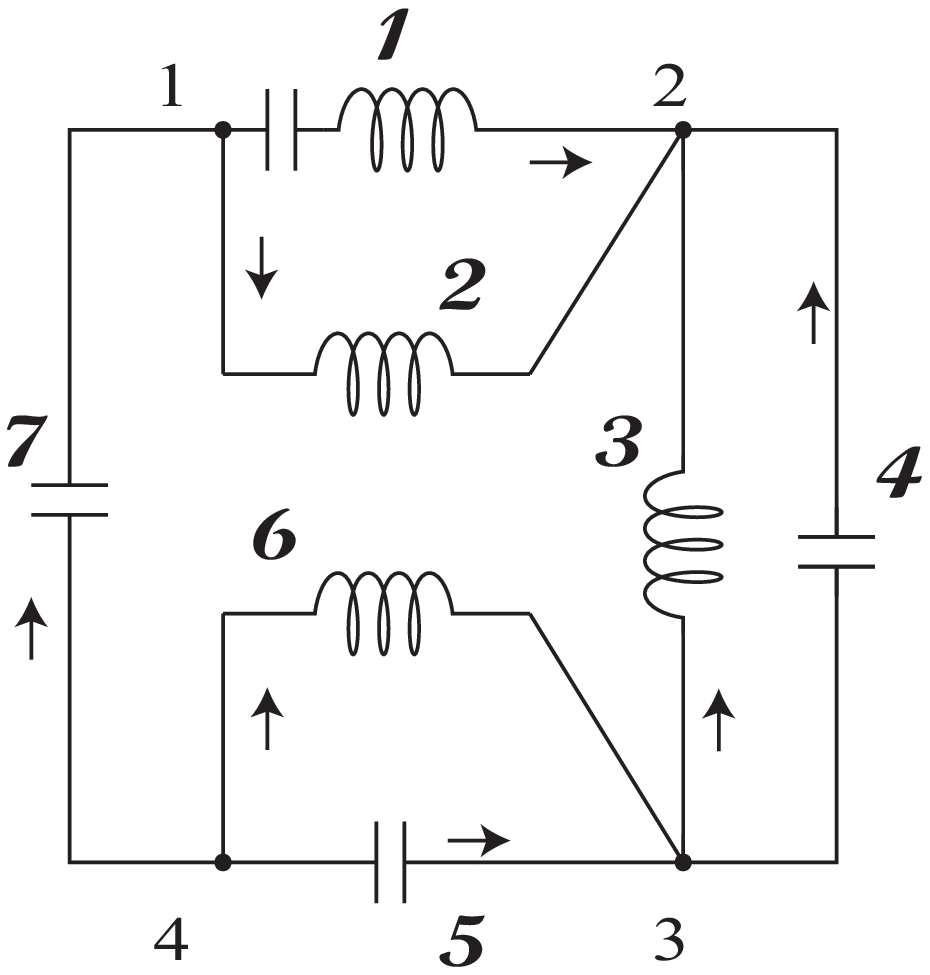}\hspace*{2cm}\includegraphics[clip,scale=0.5]{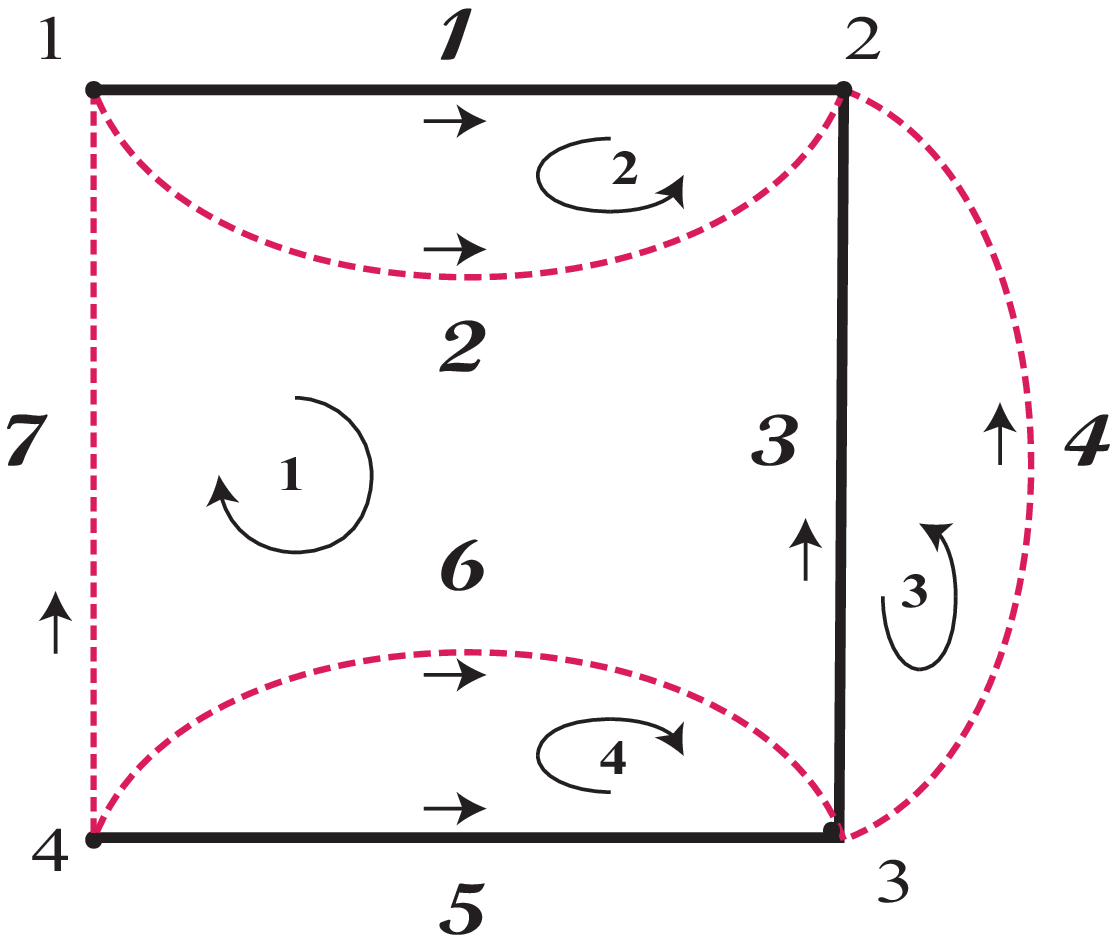}
\par\end{centering}
\centering{}\hspace*{-1cm}(a)\hspace*{6.5cm}(b)\caption{\label{fig:tree} The network (a) and its graph (b). There are $4$
nodes marked by small disks (black). In (b) there are $3$ twigs identified
by bolder (black) lines and labeled by numbers \textbf{\textit{1}},
\textbf{\textit{3}}, \textbf{\textit{5}}. There are $4$ links identified
by dashed (red) lines and labeled by numbers \textbf{\textit{2}},
\textbf{\textit{4}}, \textbf{\textit{6}}, \textbf{\textit{7}}. There
also $4$ oriented f-loops formed by the branches as shown.}
\end{figure}

There is an important concept of a \emph{tree} in the network graph
theory \cite[2.2]{BalBic}, \cite[2.1]{Cau} and \cite[2.3]{SesRee}.
A \emph{tree}, known also as \emph{complete tree}, is defined as a
connected subgraph of a connected graph containing all the nodes of
the graph but containing no loops as illustrated in Fig. \ref{fig:tree}
(b). The branches of the tree are called \emph{twigs} and those branches
that are not on a tree are called \emph{links} \cite[2.2]{BalBic}.
The links constitute the complement of the tree, or the \emph{cotree}.
The decomposition of the graph into a tree and cotree is not a unique.

The \emph{system of fundamental loops} or \emph{system of }f-loops
for short, \cite[2.2]{BalBic}, \cite[2.1]{Cau} and \cite[2.3]{SesRee},
is of particular importance to our studies. The system of time-dependent
charges (defined as the time integrals of the currents) associated
with the system of f-loops provides a c complete set of independent
variables. When the network tree is selected then every link defines
the containing it f-loop. The orientation of an f-loop is defined
by the orientation of the link it contains. Consequently, there are
as many of f-loops in as there are links, and
\begin{equation}
\text{number of \ensuremath{f}-loops}:\;N_{\mathrm{fl}}=N_{\mathrm{b}}-N_{\mathrm{n}}+1.\label{eq:numfL}
\end{equation}
The number $N_{\mathrm{fl}}$ of f-loops defined by equation (\ref{eq:numfL})
quantifies the connectivity of the network graph, and it is known
in the algebraic topology as the first Betti number \cite[2]{Cau},
\cite[2.3]{SesRee}), \cite{Witt}.

The discussed concepts of the graph of an electric network such as
the tree, twigs, links and f-loops are illustrated in Fig. \ref{fig:tree}.
In particular, there are $4$ nodes marked by small disks (black).
In Fig. \ref{fig:tree} (b) there are $3$ twigs identified by bolder
(black) lines and labeled by numbers \textbf{\textit{1}}, \textbf{\textit{3}},
\textbf{\textit{5}}. There are $4$ links identified by dashed (red)
lines and labeled by numbers \textbf{\textit{2}}, \textbf{\textit{4}},
\textbf{\textit{6}}, \textbf{\textit{7}}. There also $4$ oriented
f-loops formed by the branches as follows: (1) \textbf{\textit{7}},
\textbf{\textit{1}}, \textbf{\textit{3}}, \textbf{\textit{5}}; (2)\textbf{\textit{
2}}, \textbf{\textit{1}}; (3)\textbf{\textit{ 4}}, \textbf{\textit{3}};
(2)\textbf{\textit{ 6}}, \textbf{\textit{5}}. These representations
of the f-loops as ordered lists of branches identify the corresponding
links as number in the first position in every list.

One also distinguishes simpler \emph{planar} networks with graphs
that can be drawn so that lines representing branches do not intersect.
The graph of a general electric network does not have to be planar
though. Networks with non-planar graphs can still be represented graphically
with more complex display arrangements or algebraically by the \emph{incidence
matrices}, \cite[2.2]{BalBic}.

\section{Conclusions}

We developed here complete mathematical theory allowing to synthesize
circuits with evolution matrices exhibiting prescribed Jordan canonical
forms subject to natural constraints. In particular, we synthesized
simple lossless circuits associated with pairs of Jordan blocks of
size 2, 3 and 4, analyzed all their significant properties and derived
closed form algebraic expressions for all significant matrices. Importantly,
the elements of the constructed circuits involve negative capacitances
and/or inductances. Naturally, those negative values are needed for
chosen fixed frequencies only and that is beneficiary for efficiently
achieving them based on operational amplifier converters.

The data that supports the findings of this study are available within
the article.
\begin{quotation}
\textbf{\vspace{0.2cm}
}
\end{quotation}
\textbf{Acknowledgment:} This research was supported by AFOSR grant
\# FA9550-19-1-0103 and Northrop Grumman grant \# 2326345.

We are grateful to Prof. F. Capolino, University of California at
Irvine, for reading the manuscript and giving valuable suggestions.

\section{Appendix A: Jordan canonical form\label{sec:jord-form}}

We provide here very concise review of Jordan canonical forms following
mostly to \cite[III.4]{Hale}, \cite[3.1,3.2]{HorJohn}. As to a demonstration
of how Jordan block arises in the case of a single $n$-th order differential
equation we refer to \cite[25.4]{ArnODE}.

Let $A$ be an $n\times n$ matrix and $\lambda$ be its eigenvalue,
and let $r\left(\lambda\right)$ be the least integer $k$ such that
$\mathcal{N}\left[\left(A-\lambda\mathbb{I}\right)^{k}\right]=\mathcal{N}\left[\left(A-\lambda\mathbb{I}\right)^{k+1}\right]$,
where $\mathcal{N}\left[C\right]$ is a null space of a matrix $C$.
Then we refer to $M_{\lambda}=\mathcal{N}\left[\left(A-\lambda\mathbb{I}\right)^{r\left(\lambda\right)}\right]$
is the \emph{generalized eigenspace} of matrix $A$ corresponding
to eigenvalue $\lambda$. Then the following statements hold, \cite[III.4]{Hale}.
\begin{prop}[generalized eigenspaces]
\label{prop:gen-eig} Let $A$ be an $n\times n$ matrix and $\lambda_{1},\ldots,\lambda_{p}$
be its distinct eigenvalues. Then generalized eigenspaces $M_{\lambda_{1}},\ldots,M_{\lambda_{p}}$
are linearly independent, invariant under the matrix $A$ and
\begin{equation}
\mathbb{C}^{n}=M_{\lambda_{1}}\oplus\ldots\oplus M_{\lambda_{p}}.\label{eq:Mgeneig1a}
\end{equation}
Consequently, any vector $x_{0}$ in $\mathbb{C}^{n}$can be represented
uniquely as
\begin{equation}
x_{0}=\sum_{j=1}^{p}x_{0,j},\quad x_{0,j}\in M_{\lambda_{j}},\label{eq:Mgeneig1b}
\end{equation}
and
\begin{equation}
\exp\left\{ At\right\} x_{0}=\sum_{j=1}^{p}e^{\lambda_{j}t}p_{j}\left(t\right),\label{eq:Mgeneig1c}
\end{equation}
where column-vector polynomials $p_{j}\left(t\right)$ satisfy
\begin{gather}
p_{j}\left(t\right)=\sum_{k=0}^{r\left(\lambda_{j}\right)-1}\left(A-\lambda_{j}\mathbb{I}\right)^{k}\frac{t^{k}}{k!}x_{0,j},\quad x_{0,j}\in M_{\lambda_{j}},\quad1\leq j\leq p.\label{eq:Mgeneig1d}
\end{gather}
\end{prop}

For a complex number $\lambda$ a Jordan block $J_{r}\left(\lambda\right)$
of size $r\geq1$ is a $r\times r$ upper triangular matrix of the
form
\begin{gather}
J_{r}\left(\lambda\right)=\lambda\mathbb{I}_{r}+K_{r}=\left[\begin{array}{ccccc}
\lambda & 1 & \cdots & 0 & 0\\
0 & \lambda & 1 & \cdots & 0\\
0 & 0 & \ddots & \cdots & \vdots\\
\vdots & \vdots & \ddots & \lambda & 1\\
0 & 0 & \cdots & 0 & \lambda
\end{array}\right],\quad J_{1}\left(\lambda\right)=\left[\lambda\right],\quad J_{2}\left(\lambda\right)=\left[\begin{array}{cc}
\lambda & 1\\
0 & \lambda
\end{array}\right],\label{eq:Jork1a}
\end{gather}
\begin{equation}
K_{r}=J_{r}\left(0\right)=\left[\begin{array}{ccccc}
0 & 1 & \cdots & 0 & 0\\
0 & 0 & 1 & \cdots & 0\\
0 & 0 & \ddots & \cdots & \vdots\\
\vdots & \vdots & \ddots & 0 & 1\\
0 & 0 & \cdots & 0 & 0
\end{array}\right].\label{eq:Jork1k}
\end{equation}
The special Jordan block $K_{r}=J_{r}\left(0\right)$ defined by equation
(\ref{eq:Jork1k}) is an nilpotent matrix that satisfies the following
identities

\begin{gather}
K_{r}^{2}=\left[\begin{array}{ccccc}
0 & 0 & 1 & \cdots & 0\\
0 & 0 & 0 & \cdots & \vdots\\
0 & 0 & \ddots & \cdots & 1\\
\vdots & \vdots & \ddots & 0 & 0\\
0 & 0 & \cdots & 0 & 0
\end{array}\right],\cdots,\;K_{r}^{r-1}=\left[\begin{array}{ccccc}
0 & 0 & \cdots & 0 & 1\\
0 & 0 & 0 & \cdots & 0\\
0 & 0 & \ddots & \cdots & \vdots\\
\vdots & \vdots & \ddots & 0 & 0\\
0 & 0 & \cdots & 0 & 0
\end{array}\right],\quad K_{r}^{r}=0.\label{eq:Jord1a}
\end{gather}
A general Jordan $n\times n$ matrix $J$ is defined as a direct sum
of Jordan blocks, that is

\begin{equation}
J=\left[\begin{array}{ccccc}
J_{n_{1}}\left(\lambda_{1}\right) & 0 & \cdots & 0 & 0\\
0 & J_{n_{2}}\left(\lambda_{2}\right) & 0 & \cdots & 0\\
0 & 0 & \ddots & \cdots & \vdots\\
\vdots & \vdots & \ddots & J_{n_{q-1}}\left(\lambda_{n_{q}-1}\right) & 0\\
0 & 0 & \cdots & 0 & J_{n_{q}}\left(\lambda_{n_{q}}\right)
\end{array}\right],\quad n_{1}+n_{2}+\cdots n_{q}=n,\label{eq:Jork1b}
\end{equation}
where $\lambda_{j}$ need not be distinct. Any square matrix $A$
is similar to a Jordan matrix as in equation (\ref{eq:Jork1b}) which
is called \emph{Jordan canonical form} of $A$. Namely, the following
statement holds, \cite[3.1]{HorJohn}.
\begin{prop}[Jordan canonical form]
\label{prop:jor-can} Let $A$ be an $n\times n$ matrix. Then there
exists a non-singular $n\times n$ matrix $Q$ such that the following
block-diagonal representation holds
\begin{equation}
Q^{-1}AQ=J\label{eq:QAQC1a}
\end{equation}
where $J$ is the Jordan matrix defined by equation (\ref{eq:Jork1b})
and $\lambda_{j}$, $1\leq j\leq q$ are not necessarily different
eigenvalues of matrix $A$. Representation (\ref{eq:QAQC1a}) is known
as the \emph{Jordan canonical form} of matrix $A$, and matrices $J_{j}$
are called \emph{Jordan blocks}. The columns of the $n\times n$ matrix
$Q$ constitute the \emph{Jordan basis} providing for the Jordan canonical
form (\ref{eq:QAQC1a}) of matrix $A$.
\end{prop}

A function $f\left(J_{r}\left(s\right)\right)$ of a Jordan block
$J_{r}\left(s\right)$ is represented by the following equation \cite[7.9]{MeyCD},
\cite[10.5]{BernM}

\begin{gather}
f\left(J_{r}\left(s\right)\right)==\left[\begin{array}{ccccc}
f\left(s\right) & \partial f\left(s\right) & \frac{\partial^{2}f\left(s\right)}{2} & \cdots & \frac{\partial^{r-1}f\left(s\right)}{\left(r-1\right)!}\\
0 & f\left(s\right) & \partial f\left(s\right) & \cdots & \frac{\partial^{r-2}f\left(s\right)}{\left(r-2\right)!}\\
0 & 0 & \ddots & \cdots & \vdots\\
\vdots & \vdots & \ddots & f\left(s\right) & \partial f\left(s\right)\\
0 & 0 & \cdots & 0 & f\left(s\right)
\end{array}\right].\label{eq:JJordf1a}
\end{gather}
Notice that any function $f\left(J_{r}\left(s\right)\right)$ of the
Jordan block $J_{r}\left(s\right)$ is evidently an upper triangular
Toeplitz matrix.

There are two particular cases of formula (\ref{eq:JJordf1a}) which
can be also derived straightforwardly using equations (\ref{eq:Jord1a}):
\begin{gather}
\exp\left\{ K_{r}t\right\} =\sum_{k=0}^{r-1}\frac{t^{k}}{k!}K_{r}^{k}=\left[\begin{array}{ccccc}
1 & t & \frac{t^{2}}{2!} & \cdots & \frac{t^{r-1}}{\left(r-1\right)!}\\
0 & 1 & t & \cdots & \frac{t^{r-2}}{\left(r-2\right)!}\\
0 & 0 & \ddots & \cdots & \vdots\\
\vdots & \vdots & \ddots & 1 & t\\
0 & 0 & \cdots & 0 & 1
\end{array}\right],\label{eq:Jord1c}
\end{gather}
 
\begin{gather}
\left[J_{r}\left(s\right)\right]^{-1}=\sum_{k=0}^{r-1}s^{-k-1}\left(-K_{r}\right)^{k}=\left[\begin{array}{ccccc}
\frac{1}{s} & -\frac{1}{s^{2}} & \frac{1}{s^{3}} & \cdots & \frac{\left(-1\right)^{r-1}}{s^{r}}\\
0 & \frac{1}{s} & -\frac{1}{s^{2}} & \cdots & \frac{\left(-1\right)^{r-2}}{s^{r-1}}\\
0 & 0 & \ddots & \cdots & \vdots\\
\vdots & \vdots & \ddots & \frac{1}{s} & -\frac{1}{s^{2}}\\
0 & 0 & \cdots & 0 & \frac{1}{s}
\end{array}\right].\label{eq:JJordf1b}
\end{gather}

\section{Appendix B: Companion matrix and cyclicity condition\label{sec:co-mat}}

The companion matrix $C\left(a\right)$ for monic polynomial
\begin{equation}
a\left(s\right)=s^{\nu}+\sum_{1\leq k\leq\nu}a_{\nu-k}s^{\nu-k}\label{eq:compas1a}
\end{equation}
where coefficients $a_{k}$ are complex numbers is defined by \cite[5.2]{BernM}
\begin{equation}
C\left(a\right)=\left[\begin{array}{ccccc}
0 & 1 & \cdots & 0 & 0\\
0 & 0 & 1 & \cdots & 0\\
0 & 0 & 0 & \cdots & \vdots\\
\vdots & \vdots & \ddots & 0 & 1\\
-a_{0} & -a_{1} & \cdots & -a_{\nu-2} & -a_{\nu-1}
\end{array}\right].\label{eq:compas1b}
\end{equation}
Notice that
\begin{equation}
\det\left\{ C\left(a\right)\right\} =\left(-1\right)^{\nu}a_{0}.\label{eq:compas1c}
\end{equation}

An eigenvalue is called \emph{cyclic (nonderogatory)} if its geometric
multiplicity is 1. A square matrix is called \emph{cyclic (nonderogatory)}
if all its eigenvalues are cyclic \cite[5.5]{BernM}. The following
statement provides different equivalent descriptions of a cyclic matrix
\cite[5.5]{BernM}.
\begin{prop}[criteria for a matrix to be cyclic]
\label{prop:cyc1} Let $A\in\mathbb{C}^{n\times n}$ be $n\times n$
matrix with complex-valued entries. Let $\mathrm{spec}\,\left(A\right)=\left\{ \zeta_{1},\zeta_{2},\ldots,\zeta_{r}\right\} $
be the set of all distinct eigenvalues and $k_{j}=\mathrm{ind}{}_{A}\,\left(\zeta_{j}\right)$
is the largest size of Jordan block associated with $\zeta_{j}$.
Then the minimal polynomial $\mu_{A}\left(s\right)$ of the matrix
$A$, that is a monic polynomial of the smallest degree such that
$\mu_{A}\left(A\right)=0$, satisfies
\begin{equation}
\mu_{A}\left(s\right)=\prod_{j=1}^{r}\left(s-\zeta_{j}\right)^{k_{j}}.\label{eq:compas1d}
\end{equation}
Furthermore, and following statements are equivalent:
\end{prop}

\begin{enumerate}
\item $\mu_{A}\left(s\right)=\chi_{A}\left(s\right)=\det\left\{ s\mathbb{I}-A\right\} $.
\item $A$ is cyclic.
\item For every $\zeta_{j}$ the Jordan form of $A$ contains exactly one
block associated with $\zeta_{j}$.
\item $A$ is similar to the companion matrix $C\left(\chi_{A}\right)$.
\end{enumerate}
\begin{prop}[companion matrix factorization]
\label{prop:cyc2} Let $a\left(s\right)$ be a monic polynomial having
degree $\nu$ and $C\left(a\right)$ is its $\nu\times\nu$ companion
matrix. Then, there exist unimodular $\nu\times\nu$ matrices $S_{1}\left(s\right)$
and $S_{2}\left(s\right)$, that is $\det\left\{ S_{m}\right\} =\pm1$,
$m=1,2$, such that
\begin{equation}
s\mathbb{I}_{\nu}-C\left(a\right)=S_{1}\left(s\right)\left[\begin{array}{lr}
\mathbb{I}_{\nu-1} & 0_{\left(\nu-1\right)\times1}\\
0_{1\times\left(\nu-1\right)} & a\left(s\right)
\end{array}\right]S_{2}\left(s\right).\label{eq:compas1e}
\end{equation}
Consequently, $C\left(a\right)$ is cyclic and
\begin{equation}
\chi_{C\left(a\right)}\left(s\right)=\mu_{C\left(a\right)}\left(s\right)=a\left(s\right).\label{eq:compas1f}
\end{equation}
\end{prop}

The following statement summarizes important information on the Jordan
form of the companion matrix and the generalized Vandermonde matrix,
\cite[5.16]{BernM}, \cite[2.11]{LanTsi}, \cite[7.9]{MeyCD}.
\begin{prop}[Jordan form of the companion matrix]
\label{prop:cycJ} Let $C\left(a\right)$ be an $n\times n$ a companion
matrix of the monic polynomial $a\left(s\right)$ defined by equation
(\ref{eq:compas1a}). Suppose that the set of distinct roots of polynomial
$a\left(s\right)$ is $\left\{ \zeta_{1},\zeta_{2},\ldots,\zeta_{r}\right\} $
and $\left\{ n_{1},n_{2},\ldots,n_{r}\right\} $ is the corresponding
set of the root multiplicities such that
\begin{equation}
n_{1}+n_{2}+\cdots+n_{r}=n.\label{eq:compas2a}
\end{equation}
Then 
\begin{equation}
C\left(a\right)=RJR^{-1},\label{eq:compas2b}
\end{equation}
where
\begin{equation}
J=\mathrm{diag}\,\left(J_{n_{1}}\left(\zeta_{1}\right),J_{n_{2}}\left(\zeta_{2}\right),\ldots,J_{n_{r}}\left(\zeta_{r}\right)\right)\label{eq:compas2c}
\end{equation}
is the the Jordan form of companion matrix $C\left(a\right)$ and
$n\times n$ matrix $R$ is the so-called generalized Vandermonde
matrix defined by
\begin{equation}
R=\left[R_{1}|R_{2}|\cdots|R_{r}\right],\label{eq:compas2d}
\end{equation}
 where $R_{j}$ is $n\times n_{j}$ matrix of the form
\begin{equation}
R_{j}=\left[\begin{array}{rrcr}
1 & 0 & \cdots & 0\\
\zeta_{j} & 1 & \cdots & 0\\
\vdots & \vdots & \ddots & \vdots\\
\zeta_{j}^{n-2} & \binom{n-2}{1}\,\zeta_{j}^{n-3} & \cdots & \binom{n-2}{n_{j}-1}\,\zeta_{j}^{n-n_{j}-1}\\
\zeta_{j}^{n-1} & \binom{n-1}{1}\,\zeta_{j}^{n-2} & \cdots & \binom{n-1}{n_{j}-1}\,\zeta_{j}^{n-n_{j}}
\end{array}\right].\label{eq:compas2f}
\end{equation}
As a consequence of representation (\ref{eq:compas2c}) $C\left(a\right)$
is a cyclic matrix.
\end{prop}

As to the structure of matrix $R_{j}$ in equation (\ref{eq:compas2f}),
if we denote by $Y\left(\zeta_{j}\right)$ its first column then it
can be expressed as follows \cite[2.11]{LanTsi}:
\begin{equation}
R_{j}=\left[Y^{\left(0\right)}|Y^{\left(1\right)}|\cdots|Y^{\left(n_{j}-1\right)}\right],\quad Y^{\left(m\right)}=\frac{1}{m!}\partial_{s_{j}}^{m}Y\left(\zeta_{j}\right),\quad0\leq m\leq n_{j}-1.\label{eq:compas3a}
\end{equation}
In the case when all eigenvalues of a cyclic matrix are distinct then
the generalized Vandermonde matrix turns into the standard Vandermonde
matrix
\begin{equation}
V=\left[\begin{array}{rrcr}
1 & 1 & \cdots & 1\\
\zeta_{1} & \zeta_{2} & \cdots & \zeta_{n}\\
\vdots & \vdots & \ddots & \vdots\\
\zeta_{1}^{n-2} & \zeta_{2}^{n-2} & \cdots & \zeta_{n}^{n-2}\\
\zeta_{1}^{n-1} & \zeta_{2}^{n-1} & \cdots & \zeta_{n}^{n-1}
\end{array}\right].\label{eq:compas3c}
\end{equation}

\section{Appendix C: Matrix polynomials\label{sec:mat-poly}}

An important incentive for considering matrix polynomials is that
they are relevant to the spectral theory of the differential equations
of the order higher than 1, particularly the Euler-Lagrange equations
which are the second-order differential equations in time. We provide
here selected elements of the theory of matrix polynomials following
mostly to \cite[II.7, II.8]{GoLaRo}, \cite[9]{Baum}. General matrix
polynomial eigenvalue problem reads
\begin{equation}
A\left(s\right)x=0,\quad A\left(s\right)=\sum_{j=0}^{\nu}A_{j}s^{j},\quad x\neq0,\label{eq:Aux1a}
\end{equation}
where $s$ is complex number, $A_{k}$ are constant $m\times m$ matrices
and $x\in\mathbb{C}^{m}$ is $m$-dimensional column-vector. We refer
to problem (\ref{eq:Aux1a}) of funding complex-valued $s$ and non-zero
vector $x\in\mathbb{C}^{m}$ as polynomial eigenvalue problem. 

If a pair of a complex $s$ and non-zero vector $x$ solves problem
(\ref{eq:Aux1a}) we refer to $s$ as an \emph{eigenvalue} or as a\emph{
characteristic value} and to $x$ as the corresponding to $s$ \emph{eigenvector}.
Evidently the characteristic values of problem (\ref{eq:Aux1a}) can
be found from polynomial \emph{characteristic equation}
\begin{equation}
\det\left\{ A\left(s\right)\right\} =0.\label{eq:Aux1b}
\end{equation}
We refer to matrix polynomial $A\left(s\right)$ as \emph{regular}
if $\det\left\{ A\left(s\right)\right\} $ is not identically zero.
We denote by $m\left(s_{0}\right)$ the \emph{multiplicity} (called
also \emph{algebraic multiplicity}) of eigenvalue $s_{0}$ as a root
of polynomial $\det\left\{ A\left(s\right)\right\} $. In contrast,
the \emph{geometric multiplicity} of eigenvalue $s_{0}$ is defined
as $\dim\left\{ \ker\left\{ A\left(s_{0}\right)\right\} \right\} $,
where $\ker\left\{ A\right\} $ defined for any square matrix $A$
stands for the subspace of solutions $x$ to equation $Ax=0$. Evidently,
the geometric multiplicity of eigenvalue does not exceed its algebraic
one, see Corollary \ref{cor:dim-ker}. 

It turns out that the matrix polynomial eigenvalue problem (\ref{eq:Aux1a})
can be always recast as the standard ``linear'' eigenvalue problem,
namely
\begin{equation}
\left(s\mathsf{B}-\mathsf{A}\right)\mathsf{x}=0,\label{eq:Aux1c}
\end{equation}
where $m\nu\times m\nu$ matrices $\mathsf{A}$ and $\mathsf{B}$
are defined by
\begin{gather}
\mathsf{B}=\left[\begin{array}{ccccc}
\mathbb{I} & 0 & \cdots & 0 & 0\\
0 & \mathbb{I} & 0 & \cdots & 0\\
0 & 0 & \ddots & \cdots & \vdots\\
\vdots & \vdots & \ddots & \mathbb{I} & 0\\
0 & 0 & \cdots & 0 & A_{\nu}
\end{array}\right],\quad\mathsf{A}=\left[\begin{array}{ccccc}
0 & \mathbb{I} & \cdots & 0 & 0\\
0 & 0 & \mathbb{I} & \cdots & 0\\
0 & 0 & 0 & \cdots & \vdots\\
\vdots & \vdots & \ddots & 0 & \mathbb{I}\\
-A_{0} & -A_{1} & \cdots & -A_{\nu-2} & -A_{\nu-1}
\end{array}\right],\label{eq:CBA1b}
\end{gather}
with $\mathbb{I}$ being $m\times m$ identity matrix. Matrix $\mathsf{A}$,
particularly in monic case, is often referred to as \emph{companion
matrix}. In the case of \emph{monic polynomial} $A\left(\lambda\right)$,
when $A_{\nu}=\mathbb{I}$ is $m\times m$ identity matrix, matrix
$\mathsf{B}=\mathsf{I}$ is $m\nu\times m\nu$ identity matrix. The
reduction of original polynomial problem (\ref{eq:Aux1a}) to an equivalent
linear problem (\ref{eq:Aux1c}) is called \emph{linearization}.

The linearization is not unique, and one way to accomplish is by introducing
the so-called known ``\emph{companion polynomia}l'' which is $m\nu\times m\nu$
matrix
\begin{gather}
\mathsf{C}_{A}\left(s\right)=s\mathsf{B}-\mathsf{A}=\left[\begin{array}{ccccc}
s\mathbb{I} & -\mathbb{I} & \cdots & 0 & 0\\
0 & s\mathbb{I} & -\mathbb{I} & \cdots & 0\\
0 & 0 & \ddots & \cdots & \vdots\\
\vdots & \vdots & \vdots & s\mathbb{I} & -\mathbb{I}\\
A_{0} & A_{1} & \cdots & A_{\nu-2} & sA_{\nu}+A_{\nu-1}
\end{array}\right].\label{eq:CBA1a}
\end{gather}
Notice that in the case of the EL equations the linearization can
be accomplished by the relevant Hamilton equations.

To demonstrate the equivalency between the eigenvalue problems for
$m\nu\times m\nu$ companion polynomial $\mathsf{C}_{A}\left(s\right)$
and the original $m\times m$ matrix polynomial $A\left(s\right)$
we introduce two $m\nu\times m\nu$ matrix polynomials $\mathsf{E}\left(s\right)$
and $\mathsf{F}\left(s\right)$. Namely,
\begin{gather}
\mathsf{E}\left(s\right)=\left[\begin{array}{ccccc}
E_{1}\left(s\right) & E_{2}\left(s\right) & \cdots & E_{\nu-1}\left(s\right) & \mathbb{I}\\
-\mathbb{I} & 0 & 0 & \cdots & 0\\
0 & -\mathbb{I} & \ddots & \cdots & \vdots\\
\vdots & \vdots & \ddots & 0 & 0\\
0 & 0 & \cdots & -\mathbb{I} & 0
\end{array}\right],\label{eq:CBA1c}\\
\det\left\{ \mathsf{E}\left(s\right)\right\} =1,\nonumber 
\end{gather}
where $m\times m$ matrix polynomials $E_{j}\left(s\right)$ are defined
by the following recursive formulas
\begin{gather}
E_{\nu}\left(s\right)=A_{\nu},\quad E_{j-1}\left(s\right)=A_{j-1}+sE_{j}\left(s\right),\quad j=\nu,\ldots,2.\label{eq:CBA1d}
\end{gather}
Matrix polynomial $\mathsf{F}\left(s\right)$ is defined by
\begin{gather}
\mathsf{F}\left(s\right)=\left[\begin{array}{ccccc}
\mathbb{I} & 0 & \cdots & 0 & 0\\
-s\mathbb{I} & \mathbb{I} & 0 & \cdots & 0\\
0 & -s\mathbb{I} & \ddots & \cdots & \vdots\\
\vdots & \vdots & \ddots & \mathbb{I} & 0\\
0 & 0 & \cdots & -s\mathbb{I} & \mathbb{I}
\end{array}\right],\quad\det\left\{ \mathsf{F}\left(s\right)\right\} =1.\label{eq:CBA1e}
\end{gather}
Notice, that both matrix polynomials $\mathsf{E}\left(s\right)$ and
$\mathsf{F}\left(s\right)$ have constant determinants readily implying
that their inverses $\mathsf{E}^{-1}\left(s\right)$ and $\mathsf{F}^{-1}\left(s\right)$
are also matrix polynomials. Then it is straightforward to verify
that
\begin{gather}
\mathsf{E}\left(s\right)\mathsf{C}_{A}\left(s\right)\mathsf{F}^{-1}\left(s\right)=\mathsf{E}\left(s\right)\left(s\mathsf{B}-\mathsf{A}\right)\mathsf{F}^{-1}\left(s\right)=\left[\begin{array}{ccccc}
A\left(s\right) & 0 & \cdots & 0 & 0\\
0 & \mathbb{I} & 0 & \cdots & 0\\
0 & 0 & \ddots & \cdots & \vdots\\
\vdots & \vdots & \ddots & \mathbb{I} & 0\\
0 & 0 & \cdots & 0 & \mathbb{I}
\end{array}\right].\label{eq:CBA1f}
\end{gather}
The identity (\ref{eq:CBA1f}) where matrix polynomials $\mathsf{E}\left(s\right)$
and $\mathsf{F}\left(s\right)$ have constant determinants can be
viewed as the definition of equivalency between matrix polynomial
$A\left(s\right)$ and its companion polynomial $\mathsf{C}_{A}\left(s\right)$. 

Let us take a look at the eigenvalue problem for eigenvalue $s$ and
eigenvector $\mathsf{x}\in\mathbb{C}^{m\nu}$ associated with companion
polynomial $\mathsf{C}_{A}\left(s\right)$, that is
\begin{gather}
\left(s\mathsf{B}-\mathsf{A}\right)\mathsf{x}=0,\quad\mathsf{x}=\left[\begin{array}{c}
x_{0}\\
x_{1}\\
x_{2}\\
\vdots\\
x_{\nu-1}
\end{array}\right]\in\mathbb{C}^{m\nu},\quad x_{j}\in\mathbb{C}^{m},\quad0\leq j\leq\nu-1,\label{eq:CBAx1a}
\end{gather}
where
\begin{equation}
\left(s\mathsf{B}-\mathsf{A}\right)\mathsf{x}=\left[\begin{array}{c}
sx_{0}-x_{1}\\
sx_{1}-x_{2}\\
\vdots\\
sx_{\nu-2}-x_{\nu-1}\\
\sum_{j=0}^{\nu-2}A_{j}x_{j}+\left(sA_{\nu}+A_{\nu-1}\right)x_{\nu-1}
\end{array}\right].\label{eq:CBAx1b}
\end{equation}
With equations (\ref{eq:CBAx1a}) and (\ref{eq:CBAx1b}) in mind we
introduce the following vector polynomial
\begin{equation}
\mathsf{x}_{s}=\left[\begin{array}{c}
x_{0}\\
sx_{0}\\
\vdots\\
s^{\nu-2}x_{0}\\
s^{\nu-1}x_{0}
\end{array}\right],\quad x_{0}\in\mathbb{C}^{m}.\label{eq:CBAx1c}
\end{equation}
Not accidentally, the components of the vector $\mathsf{x}_{s}$ in
its representation (\ref{eq:CBAx1c}) are in evident relation with
the derivatives $\partial_{t}^{j}\left(x_{0}\mathrm{e}^{st}\right)=s^{j}x_{0}\mathrm{e}^{st}$.
That is just another sign of the intimate relations between the matrix
polynomial theory and the theory of systems of ordinary differential
equations, see Section \ref{sec:dif-jord}. 
\begin{thm}[eigenvectors]
\label{thm:matpol-eigvec} Let $A\left(s\right)$ as in equations
(\ref{eq:Aux1a}) be regular, that $\det\left\{ A\left(s\right)\right\} $
is not identically zero, and let $m\nu\times m\nu$ matrices $\mathsf{A}$
and $\mathsf{B}$ be defined by equations (\ref{eq:Aux1b}). Then
the following identities hold
\begin{equation}
\left(s\mathsf{B}-\mathsf{A}\right)\mathsf{x}_{s}=\left[\begin{array}{c}
0\\
0\\
\vdots\\
0\\
A\left(s\right)x_{0}
\end{array}\right],\;\mathsf{x}_{s}=\left[\begin{array}{c}
x_{0}\\
sx_{0}\\
\vdots\\
s^{\nu-2}x_{0}\\
s^{\nu-1}x_{0}
\end{array}\right],\label{eq:CBAx1d}
\end{equation}
\begin{gather}
\det\left\{ A\left(s\right)\right\} =\det\left\{ s\mathsf{B}-\mathsf{A}\right\} ,\quad\det\left\{ \mathsf{B}\right\} =\det\left\{ A_{\nu}\right\} ,\label{eq:CBAx1g}
\end{gather}
where $\det\left\{ A\left(s\right)\right\} =\det\left\{ s\mathsf{B}-\mathsf{A}\right\} $
is a polynomial of the degree $m\nu$ if $\det\left\{ \mathsf{B}\right\} =\det\left\{ A_{\nu}\right\} \neq0$.
There is one-to-one correspondence between solutions of equations
$A\left(s\right)x=0$ and $\left(s\mathsf{B}-\mathsf{A}\right)\mathsf{x}=0$.
Namely, a pair $s,\:\mathsf{x}$ solves eigenvalue problem $\left(s\mathsf{B}-\mathsf{A}\right)\mathsf{x}=0$
if and only if the following equalities hold
\begin{gather}
\mathsf{x}=\mathsf{x}_{s}=\left[\begin{array}{c}
x_{0}\\
sx_{0}\\
\vdots\\
s^{\nu-2}x_{0}\\
s^{\nu-1}x_{0}
\end{array}\right],\quad A\left(s\right)x_{0}=0,\quad x_{0}\neq0;\quad\det\left\{ A\left(s\right)\right\} =0.\label{eq:CBAx1e}
\end{gather}
\end{thm}

\begin{proof}
Polynomial vector identity (\ref{eq:CBAx1d}) readily follows from
equations (\ref{eq:CBAx1b}) and (\ref{eq:CBAx1c}). Identities (\ref{eq:CBAx1g})
for the determinants follow straightforwardly from equations (\ref{eq:CBAx1c}),
(\ref{eq:CBAx1e}) and (\ref{eq:CBA1f}). If $\det\left\{ \mathsf{B}\right\} =\det\left\{ A_{\nu}\right\} \neq0$
then the degree of the polynomial $\det\left\{ s\mathsf{B}-\mathsf{A}\right\} $
has to be $m\nu$ since $\mathsf{A}$ and $\mathsf{B}$ are $m\nu\times m\nu$
matrices.

Suppose that equations (\ref{eq:CBAx1e}) hold. Then combining them
with proven identity (\ref{eq:CBAx1d}) we get $\left(s\mathsf{B}-\mathsf{A}\right)\mathsf{x}_{s}=0$
proving that expressions (\ref{eq:CBAx1e}) define an eigenvalue $s$
and an eigenvector $\mathsf{x}=\mathsf{x}_{s}$.

Suppose now that $\left(s\mathsf{B}-\mathsf{A}\right)\mathsf{x}=0$
where $\mathsf{x}\neq0$. Combing that with equations (\ref{eq:CBAx1b})
we obtain
\begin{gather}
x_{1}=sx_{0},\quad x_{2}=sx_{1}=s^{2}x_{0},\cdots,\quad x_{\nu-1}=s^{\nu-1}x_{0},\label{eq:CBAx2a}
\end{gather}
implying that
\begin{equation}
\mathsf{x}=\mathsf{x}_{s}=\left[\begin{array}{c}
x_{0}\\
sx_{0}\\
\vdots\\
s^{\nu-2}x_{0}\\
s^{\nu-1}x_{0}
\end{array}\right],\quad x_{0}\neq0,\label{eq:CBAx2b}
\end{equation}
 and 
\begin{equation}
\sum_{j=0}^{\nu-2}A_{j}x_{j}+\left(sA_{\nu}+A_{\nu-1}\right)x_{\nu-1}=A\left(s\right)x_{0}.\label{eq:CBAx2c}
\end{equation}
Using equations (\ref{eq:CBAx2b}) and identity (\ref{eq:CBAx1d})
we obtain
\begin{equation}
0=\left(s\mathsf{B}-\mathsf{A}\right)\mathsf{x}=\left(s\mathsf{B}-\mathsf{A}\right)\mathsf{x}_{s}=\left[\begin{array}{c}
0\\
0\\
\vdots\\
0\\
A\left(s\right)x_{0}
\end{array}\right].\label{eq:CBAx2d}
\end{equation}
 Equations (\ref{eq:CBAx2d}) readily imply $A\left(s\right)x_{0}=0$
and $\det\left\{ A\left(s\right)\right\} =0$ since $x_{0}\neq0$.
That completes the proof.
\end{proof}
\begin{rem}[characteristic polynomial degree]
\label{rem:char-pol-deg} Notice that according to Theorem \ref{thm:matpol-eigvec}
the characteristic polynomial $\det\left\{ A\left(s\right)\right\} $
for $m\times m$ matrix polynomial $A\left(s\right)$ has the degree
$m\nu$, whereas in linear case $s\mathbb{I}-A_{0}$ for $m\times m$
identity matrix $\mathbb{I}$ and $m\times m$ matrix $A_{0}$ the
characteristic polynomial $\det\left\{ s\mathbb{I}-A_{0}\right\} $
is of the degree $m$. This can be explained by observing that in
the non-linear case of $m\times m$ matrix polynomial $A\left(s\right)$
we are dealing effectively with many more $m\times m$ matrices $A$
than just a single matrix $A_{0}$.
\end{rem}

Another problem of our particular interest related to the theory of
matrix polynomials is eigenvalues and eigenvectors degeneracy and
consequently the existence of non-trivial Jordan blocks, that is Jordan
blocks of dimensions higher or equal to 2. The general theory addresses
this problem by introducing so-called ``Jordan chains'' which are
intimately related to the theory of system of differential equations
expressed as $A\left(\partial_{t}\right)x\left(t\right)=0$ and their
solutions of the form $x\left(t\right)=p\left(t\right)e^{st}$ where
$p\left(t\right)$ is a vector polynomial, see Section \ref{sec:dif-jord}
and \cite[I, II]{GoLaRo}, \cite[9]{Baum}. Avoiding the details of
Jordan chains developments we simply notice that an important to us
point of Theorem \ref{thm:matpol-eigvec} is that there is one-to-one
correspondence between solutions of equations $A\left(s\right)x=0$
and $\left(s\mathsf{B}-\mathsf{A}\right)\mathsf{x}=0$, and it has
the following immediate implication.
\begin{cor}[equality of the dimensions of eigenspaces]
\label{cor:dim-ker} Under the conditions of Theorem \ref{thm:matpol-eigvec}
for any eigenvalue $s_{0}$, that is $\det\left\{ A\left(s_{0}\right)\right\} =0$,
we have
\begin{equation}
\dim\left\{ \ker\left\{ s_{0}\mathsf{B}-\mathsf{A}\right\} \right\} =\dim\left\{ \ker\left\{ A\left(s_{0}\right)\right\} \right\} .\label{eq:CBAx2e}
\end{equation}
In other words, the geometric multiplicities of the eigenvalue $s_{0}$
associated with matrices $A\left(s_{0}\right)$ and $s_{0}\mathsf{B}-\mathsf{A}$
are equal. In view of identity (\ref{eq:CBAx2e}) the following inequality
holds for the (algebraic) multiplicity $m\left(s_{0}\right)$
\begin{equation}
m\left(s_{0}\right)\geq\dim\left\{ \ker\left\{ A\left(s_{0}\right)\right\} \right\} .\label{eq:CBAx2f}
\end{equation}
\end{cor}

The next statement shows that if the geometric multiplicity of an
eigenvalue is strictly less than its algebraic one than there exist
non-trivial Jordan blocks, that is Jordan blocks of dimensions higher
or equal to 2.
\begin{thm}[non-trivial Jordan block]
\label{thm:Jord-block} Assuming notations introduced in Theorem
\ref{thm:matpol-eigvec} let us suppose that the multiplicity $m\left(s_{0}\right)$
of eigenvalue $s_{0}$ satisfies
\begin{equation}
m\left(s_{0}\right)>\dim\left\{ \ker\left\{ A\left(s_{0}\right)\right\} \right\} .\label{eq:CBAAx3a}
\end{equation}
Then the Jordan canonical form of companion polynomial $\mathsf{C}_{A}\left(s\right)=s\mathsf{B}-\mathsf{A}$
has a least one nontrivial Jordan block of the dimension exceeding
2.

In particular, if 
\begin{equation}
\dim\left\{ \ker\left\{ s_{0}\mathsf{B}-\mathsf{A}\right\} \right\} =\dim\left\{ \ker\left\{ A\left(s_{0}\right)\right\} \right\} =1,\label{eq:CBAAx3b}
\end{equation}
 and $m\left(s_{0}\right)\geq2$ then the Jordan canonical form of
companion polynomial $\mathsf{C}_{A}\left(s\right)=s\mathsf{B}-\mathsf{A}$
has exactly one Jordan block associated with eigenvalue $s_{0}$ and
its dimension is $m\left(s_{0}\right)$.
\end{thm}

The proof of Theorem \ref{thm:Jord-block} follows straightforwardly
from the definition of the Jordan canonical form and its basic properties.
Notice that if equations (\ref{eq:CBAAx3b}) hold that implies that
the eigenvalue $0$ is cyclic (nonderogatory) for matrix $A\left(s_{0}\right)$
and eigenvalue $s_{0}$ is cyclic (nonderogatory) for matrix $\mathsf{B}^{-1}\mathsf{A}$
provided $\mathsf{B}^{-1}$ exists, see Section \ref{sec:co-mat}.

\section{Appendix D: Vector differential equations and the Jordan canonical
form\label{sec:dif-jord}}

In this section we relate the vector ordinary equations to the matrix
polynomials reviewed in Section \ref{sec:mat-poly} following to \cite[5.1, 5.7]{GoLaRo2},
\cite[III.4]{Hale}, \cite[7.9]{MeyCD}.

Equation $A\left(s\right)x=0$ with polynomial matrix $A\left(s\right)$
defined by equations (\ref{eq:Aux1a}) corresponds to the following
$m$-vector $\nu$-th order ordinary differential
\begin{equation}
A\left(\partial_{t}\right)x\left(t\right)=0,\text{ where }A\left(\partial_{t}\right)=\sum_{j=0}^{\nu}A_{j}\partial_{t}^{j},\label{eq:Adtx1}
\end{equation}
where $A_{j}$ are $m\times m$ matrices. Then differential equation
(\ref{eq:Adtx1}) can be recast in standard fashion as $m\nu$-vector
first order differential equation
\begin{equation}
\mathsf{B}\partial_{t}Y\left(t\right)=\mathsf{A}Y\left(t\right),\label{eq:yxt1a}
\end{equation}
where $\mathsf{A}$ and $\mathsf{B}$ are $m\nu\times m\nu$ companion
matrices defined by equations (\ref{eq:CBA1b}) and
\begin{equation}
Y\left(t\right)=\left[\begin{array}{c}
x\left(t\right)\\
\partial_{t}x\left(t\right)\\
\vdots\\
\partial_{t}^{\nu-2}x\left(t\right)\\
\partial_{t}^{\nu-1}x\left(t\right)
\end{array}\right]\label{eq:yxt1b}
\end{equation}
is $m\nu$-column-vector function.

In the case when $A_{\nu}$ is an invertible $m\times m$ matrix equation
(\ref{eq:yxt1a}) can be recast further as
\begin{equation}
\partial_{t}Y\left(t\right)=\dot{\mathsf{A}}Y\left(t\right),\label{eq:yxt1c}
\end{equation}
where
\begin{gather}
\dot{\mathsf{A}}=\left[\begin{array}{ccccc}
0 & \mathbb{I} & \cdots & 0 & 0\\
0 & 0 & \mathbb{I} & \cdots & 0\\
0 & 0 & 0 & \cdots & \vdots\\
\vdots & \vdots & \ddots & 0 & \mathbb{I}\\
-\dot{A}_{0} & -\dot{A}_{1} & \cdots & -\dot{A}_{\nu-2} & -\dot{A}_{\nu-1}
\end{array}\right],\quad\dot{A}_{j}=A_{\nu}^{-1}A_{j},\quad0\leq\nu-1.\label{eq:yxt1d}
\end{gather}
Notice one can interpret equation (\ref{eq:yxt1c}) as particular
case of equation (\ref{eq:yxt1a}) where matrices $A_{\nu}$ and $\mathsf{B}$
are identity matrices of the respective dimensions $m\times m$ and
$m\nu\times m\nu$, and that polynomial matrix $A\left(s\right)$
defined by equations (\ref{eq:Aux1a}) becomes monic matrix polynomial
$\dot{A}\left(s\right)$, that is
\begin{gather}
\dot{A}\left(s\right)=\mathbb{I}s^{\nu}+\sum_{j=0}^{\nu-1}\dot{A}_{j}s^{j},\quad\dot{A}_{j}=A_{\nu}^{-1}A_{j},\quad0\leq\nu-1.\label{eq:yxt1e}
\end{gather}
Notice that in view of equation (\ref{eq:yxt1b}) one recovers $x\left(t\right)$
from $Y\left(t\right)$ by the following formula
\begin{equation}
x\left(t\right)=P_{1}Y\left(t\right),\quad P_{1}=\left[\begin{array}{ccccc}
\mathbb{I} & 0 & \cdots & 0 & 0\end{array}\right],\label{eq:yxt2b}
\end{equation}
where $P_{1}$ evidently is $m\times m\nu$ matrix.

Observe also that, \cite[Prop. 5.1.2]{GoLaRo2}, \cite[14]{LanTsi}
\begin{gather}
\left[\dot{A}\left(s\right)\right]^{-1}=P_{1}\left[\mathbb{I}s-\dot{\mathsf{A}}\right]^{-1}R_{1},\quad P_{1}=\left[\begin{array}{ccccc}
\mathbb{I} & 0 & \cdots & 0 & 0\end{array}\right],\quad R_{1}=\left[\begin{array}{c}
0\\
0\\
\vdots\\
0\\
\mathbb{I}
\end{array}\right],\label{eq:yxt2a}
\end{gather}
where $P_{1}$ and $R_{1}$ evidently respectively $m\times m\nu$
and $m\nu\times m$ matrices.

The general form for the solution to vector differential equation
(\ref{eq:yxt1c}) is
\begin{equation}
Y\left(t\right)=\exp\left\{ \dot{\mathsf{A}}t\right\} Y_{0},\quad Y_{0}\in\mathbb{C}^{m\nu}.\label{eq:yxt2c}
\end{equation}
Then using the formulas (\ref{eq:yxt2b}), (\ref{eq:yxt2c}) and Proposition
\ref{prop:gen-eig} we arrive the following statement.
\begin{prop}[solution to the vector differential equation ]
\label{prop:dif-sol-g} Let $\dot{\mathsf{A}}$ be $m\nu\times m\nu$
companion matrix defined by equations (\ref{eq:yxt1d}), $\zeta_{1},\ldots,\zeta_{p}$
be its distinct eigenvalues, and $M_{\zeta_{1}},\ldots,M_{\zeta_{p}}$
be the corresponding generalized eigenspaces of the corresponding
dimensions $r\left(\zeta_{j}\right)$, $1\leq j\leq p$. Then the
$m\nu$ column-vector solution $Y\left(t\right)$ to differential
equation (\ref{eq:yxt1c}) is of the form
\begin{gather}
Y\left(t\right)=\exp\left\{ \dot{\mathsf{A}}t\right\} Y_{0}=\sum_{j=1}^{p}e^{\zeta_{j}t}p_{j}\left(t\right),\quad Y_{0}=\sum_{j=1}^{p}Y_{0,j},\quad Y_{0,j}\in M_{\zeta_{j}},\label{eq:yxt2d}
\end{gather}
where $m\nu$-column-vector polynomials $p_{j}\left(t\right)$ satisfy
\begin{gather}
p_{j}\left(t\right)=\sum_{k=0}^{r\left(\zeta_{j}\right)-1}\frac{t^{k}}{k!}\left(\dot{\mathsf{A}}-\zeta_{j}\mathbb{I}\right)^{k}Y_{0,j},\quad1\leq j\leq p.\label{eq:yxt2e}
\end{gather}
Consequently, the general $m$-column-vector solution $x\left(t\right)$
to differential equation (\ref{eq:Adtx1}) is of the form
\begin{gather}
x\left(t\right)=\sum_{j=1}^{p}e^{\zeta_{j}t}P_{1}p_{j}\left(t\right),\quad P_{1}=\left[\begin{array}{ccccc}
\mathbb{I} & 0 & \cdots & 0 & 0\end{array}\right].\label{eq:yxt2f}
\end{gather}
\end{prop}

Notice that $\chi_{\dot{\mathsf{A}}}\left(s\right)=\det\left\{ s\mathbb{I}-\dot{\mathsf{A}}\right\} $
is the characteristic function of the matrix $\dot{\mathsf{A}}$ then
using notations of Proposition \ref{prop:dif-sol-g} we obtain
\begin{equation}
\chi_{\dot{\mathsf{A}}}\left(s\right)=\prod_{j=1}^{p}\left(s-\zeta_{j}\right)^{r\left(\zeta_{j}\right)}.\label{eq:yxt3a}
\end{equation}
Notice also that for any values of complex-valued coefficients $b_{k}$
we have
\begin{gather}
\left(\partial_{t}-\zeta_{j}\right)^{r\left(\zeta_{j}\right)}\left[e^{\zeta_{j}t}p_{j}\left(t\right)\right]=0,\quad p_{j}\left(t\right)=\sum_{k=0}^{r\left(\zeta_{j}\right)-1}b_{k}t^{k},\label{eq:yxt3b}
\end{gather}
implying together with representation (\ref{eq:yxt3a})
\begin{gather}
\chi_{\dot{\mathsf{A}}}\left(\partial_{t}\right)\left[e^{\zeta_{j}t}p_{j}\left(t\right)\right]=0,\quad p_{j}\left(t\right)=\sum_{k=0}^{r\left(\zeta_{j}\right)-1}b_{k}t^{k}.\label{eq:yxt3c}
\end{gather}
Combing now Proposition \ref{prop:dif-sol-g} with equation (\ref{eq:yxt3c})
we obtain the following statement.
\begin{cor}[property of a solution to the vector differential equation]
\label{cor:dif-sol-g} Let $x\left(t\right)$ be the general $m$-column-vector
solution $x\left(t\right)$ to differential equation (\ref{eq:Adtx1}).
Then $x\left(t\right)$ satisfies
\begin{equation}
\chi_{\dot{\mathsf{A}}}\left(\partial_{t}\right)x\left(t\right)=0.\label{eq:yxt3d}
\end{equation}
\end{cor}

\section{Appendix E: Some properties of block matrices\label{sec:block-mat}}

The statements on block matrices below are useful for our studies
\cite[2.8]{BernM}.
\begin{prop}[factorization of a block matrix]
\label{prop:4blockA1} Let $A\in\mathbb{C}^{n\times n}$, $B\in\mathbb{C}^{n\times m}$,
$C\in\mathbb{C}^{p\times n}$, $D\in\mathbb{C}^{p\times m}$, and
assume $A$ is nonsingular. Then
\begin{gather}
\left[\begin{array}{cr}
A & B\\
C & D
\end{array}\right]=\left[\begin{array}{rc}
I & 0\\
CA^{-1} & I
\end{array}\right]\left[\begin{array}{cr}
A & 0\\
0 & D-CA^{-1}B
\end{array}\right]\left[\begin{array}{rr}
I & A^{-1}B\\
0 & I
\end{array}\right],\label{eq:BLock1d}
\end{gather}
and
\begin{gather}
\mathrm{Rank}\,\left\{ \left[\begin{array}{cr}
A & B\\
C & D
\end{array}\right]\right\} =n+\mathrm{Rank}\,\left\{ D-CA^{-1}B\right\} .\label{eq:BLock1e}
\end{gather}
If furthermore, $m=p$, that is $A\in\mathbb{C}^{n\times n}$, $B\in\mathbb{C}^{n\times m}$,
$C\in\mathbb{C}^{m\times n}$, $D\in\mathbb{C}^{m\times m}$, then
\begin{gather}
\det\left\{ \left[\begin{array}{cr}
A & B\\
C & D
\end{array}\right]\right\} =\det\left\{ A\right\} \det\left\{ D-CA^{-1}B\right\} .\label{eq:Block1f}
\end{gather}
\end{prop}

\section{Appendix F: Canonical forms for quadratic Hamiltonian\label{sec:can-forms}}

Canonical forms of Hamiltonian matrices and Hamiltonians is a well-studied
subject \cite[App. 6]{ArnMec}, \cite[2.4]{ArnGiv}, \cite[3.3, 4.6]{Miln},
\cite[3.3, 4.6, 4.7]{Mey}. In particular, an approach to the canonical
forms due to D. Galin is as follows \cite[App. 6]{ArnMec}.

Hamiltonian associated with a pair of Jordan blocks of order $n$
with real eigenvalues $\pm a$ is as follows
\begin{equation}
H=-a\sum_{j=1}^{n}p_{j}q_{j}+\sum_{j=1}^{n-1}p_{j}q_{j+1}.\label{eq:cafoH1a}
\end{equation}
Hamiltonian associated with a pair of Jordan blocks of odd order $2n+1$
with purely imaginary eigenvalues $\pm b\mathrm{i}$, that is $b$
is real, is one of the following two nonequivalent types:
\begin{gather}
H=\pm\frac{1}{2}\left[H_{1}-H_{2}\right]-\sum_{j=1}^{2n}p_{j}q_{j+1},\quad H_{1}=\sum_{j=1}^{n}\left(b^{2}p_{2j}q_{2n-2j+2}+p_{2j}q_{2n-2j+2}\right),\label{eq:cafoH1b}\\
H_{2}=\sum_{j=1}^{n+1}\left(b^{2}p_{2j-1}q_{2n-2j+3}+p_{2j-1}q_{2n-2j+3}\right).\nonumber 
\end{gather}
Hamiltonian associated with a pair of Jordan blocks of even order
$2n$ with purely imaginary eigenvalues $\pm b\mathrm{i}$, that is
$b$ is real, is one of the following two nonequivalent types:
\begin{gather}
H=\pm\frac{1}{2}\left[H_{1}-H_{2}\right]-b^{2}\sum_{j=1}^{n}p_{2j}q_{2j-1}+\sum_{j=1}^{n}p_{2j}q_{2j-1}.\label{eq:cafoH1c}\\
H_{1}=\sum_{j=1}^{n}\left(\frac{1}{b^{2}}q_{2j-1}q_{2n-2j+1}+q_{2j}q_{2n-2j+2}\right),\nonumber \\
H_{2}=\sum_{j=1}^{n-1}\left(b^{2}p_{2j+1}q_{2n-2j+1}+p_{2j+2}q_{2n-2j+3}\right).\nonumber 
\end{gather}
In particular for $n=1$ the above formula turns into
\begin{equation}
H=\pm\frac{1}{2}\left(\frac{1}{b^{2}}q_{1}^{2}+q_{2}^{2}\right)-b^{2}p_{1}q_{2}+p_{2}q_{1}.\label{eq:cafoH1d}
\end{equation}

Hamiltonian associated with a quadruple of Jordan blocks of order
$n$ with eigenvalues $\pm a\pm b\mathrm{i}$, is as follows
\begin{gather}
H=-a\sum_{j=1}^{2n}p_{j}q_{j}+b\sum_{j=1}^{n}\left(p_{2j-1}q_{2j}-p_{2j}q_{2j-1}\right)+\sum_{j=1}^{2n-2}p_{j}q_{j+2}.\label{eq:cafoH1e}
\end{gather}

Canonical forms of Hamiltonians according to \cite[2.4]{ArnGiv} are
as follows. 

Hamiltonian associated with a pair of Jordan blocks of order $n$
with real or pure imaginary eigenvalues $\pm\chi$ is as follows

\begin{gather}
H=\pm\left[\sum_{j=1}^{n-1}p_{j}q_{j+1}+\frac{p_{n}^{2}}{2}-\frac{1}{2}\sum_{j=1}^{n}\left(\begin{array}{c}
n\\
j-1
\end{array}\right)\chi^{2\left(n-j+1\right)}q_{j}^{2}\right],\text{ where}\label{eq:cafoH1f}\\
\left(\begin{array}{c}
n\\
j
\end{array}\right)=C_{j}^{n}=\frac{n!}{j!\left(n-j\right)!}\text{ is the binomial coefficient.}\label{eq:cafoH1g}
\end{gather}

Hamiltonian associated with a quadruple of Jordan blocks of order
$m=\frac{n}{2}$ where $n$ is an even positive integer with eigenvalues
$\pm a\pm b\mathrm{i}$, is as follows
\begin{equation}
H=\sum_{j=1}^{n-1}p_{j}q_{j+1}+\frac{p_{n}^{2}}{2}-\frac{1}{2}\sum_{j=1}^{n}a_{j-1}q_{j}^{2},\quad n=2m,\label{eq:cafoH1h}
\end{equation}
where
\begin{gather}
\sum_{j=0}^{2m}a_{j}\zeta^{j}==\left[\zeta^{2}+2\left(a^{2}-b^{2}\right)\zeta+\left(a^{2}+b^{2}\right)^{2}\right]^{m}.\label{eq:cafoH1j}
\end{gather}

\section{Appendix G: Notations\label{sec:notation}}
\begin{itemize}
\item $\mathbb{C}$ is a set of complex number.
\item $\bar{s}$ is complex-conjugate to complex number $s$
\item $\mathbb{C}^{n}$ is a set of $n$ dimensional column vectors with
complex complex-valued entries.
\item $\mathbb{C}^{n\times m}$ is a set of $n\times m$ matrices with complex-valued
entries.
\item $\mathbb{R}^{n\times m}$ is a set of $n\times m$ matrices with real-valued
entries.
\item $\mathrm{spec}\,\left(A\right)$ is the set of all distinct eigenvalues
of a $n\times n$ matrix $A$.
\item $\mathrm{ind}{}_{A}\,\left(\lambda\right)$ is defined for an eigenvalue
$\lambda$ of a $n\times n$ matrix $A$ to be the largest size of
Jordan block associated with $\lambda$.
\item $\chi_{A}\left(s\right)=\det\left\{ s\mathbb{I}_{\nu}-A\right\} $
is the characteristic polynomial of a $\nu\times\nu$ matrix $A$.
\item $\mu_{A}\left(s\right)$ is the minimal polynomial of a $n\times n$
matrix $A$, that is the smallest degree polynomial such that $\mu_{A}\left(A\right)=0.$ 
\item $\mathrm{col}\,\left(A,k\right)$ is $k$-th column of matrix $A$.
\item $\mathrm{row}\,\left(A,k\right)$ is $k$-th row of matrix $A$.
\item $\mathbb{I}_{\nu}$ is $\nu\times\nu$ identity matrix.
\item $\mathbb{J},\;\mathbb{J}_{\nu}$ is $2\nu\times2\nu$ unit imaginary
matrix.
\item $M^{\mathrm{T}}$ is a matrix transposed to matrix $M$.
\item $\mathrm{diag}\,\left(b_{1},b_{2},\ldots,b_{n}\right)$ is $n\times n$
diagonal matrix with indicated entries.
\item $\mathrm{diag}\,\left(B_{1},B_{2},\ldots,B_{n}\right)$ is diagonal
block-matrix where $B_{J}$ are square matrices.
\item $\left[R_{1}|R_{2}|\cdots|R_{r}\right]$ is an $n\times m$ matrix
formed by putting next to each other in a row $n\times m_{j}$ matrices
$R_{j}$, $1\leq j\leq r$ where $m=m_{1}+\cdots+m_{r}$.
\item $L_{j}$, $C_{j}$ and $G_{j}$ are respectively inductances, capacitances
and gyrator resistances.
\item EL stands for the Euler-Lagrange (equations).
\end{itemize}


\begin{thebibliography}{CarKroPea}
\bibitem[ArnODE]{ArnODE} Arnold V., \textsl{Ordinary Differential
Equations}, 3rd ed., Springer, 1992.

\bibitem[ArnMec]{ArnMec} Arnold V., \textsl{Mathematical Methods
of Classical Mechanics}, Springer, 1989.

\bibitem[ArnGiv]{ArnGiv} Arnold. V. and Givental A., \textsl{Symplectic
geometry}, in: \textsl{Dynamical Systems, IV}, in: Encyclopaedia Math.
Sci., vol.4, Springer, pp.1\textendash 138, 2001.

\bibitem[BalBic]{BalBic} Balabanian N. and Bickart T., \textsl{Electrical
Network Theory}, John Wiley \& Sons, 1969.

\bibitem[Baum]{Baum} Baumgartel H., \textsl{Analytic Perturbation
Theory for Matrices and Operators}, Birkhauser, 1985.

\bibitem[BenBoe]{BenBoe} Bender C. and Boettcher S., Real Spectra
in Non- Hermitian Hamiltonians Having PT Symmetry, Phys. Rev. Lett.,
\textbf{80}, 5243 (1998).

\bibitem[BernM]{BernM} Bernstein D., \textsl{Matrix Mathematics:
Theory, Facts, and Formulas}, 2 edn., Princeton University Press,
2009.

\bibitem[Berr]{Berr} Berry, M., \textsl{Physics of non-Hermitian
degeneracies}, Che (EPD), Czech. J. Phys. \textbf{54}, , 1039\textendash 1047,
(2004).

\bibitem[Cau]{Cau} Cauer W., \textsl{Synthesis of Linear Communication
Networks}, Volumes I, II, McGraw-Hill, 1958.

\bibitem[CarKroPea]{CarKroPea} Carrier G., Krook M. and Pearson C.,
\textsl{Functions of a complex variable. Theory and technique}, SIAM,
2005.

\bibitem[CheN]{CheN} Chen W. et. al., \textsl{Exceptional points
enhance sensing in an optical microcavity}, Nature, \textbf{548},
192-196, (2017).

\bibitem[Dorf]{Dorf} Dorf R., \textsl{The Electrical Engineering
Handbook}, CRC Press, 2000.

\bibitem[FigTWTbk]{FigTWTbk} Figotin A., \textsl{An Analytic Theory
of Multi-stream Electron Beams in Traveling Wave Tubes}, World Scientific,
to appear in 2020.

\bibitem[FigWel14]{FigWel14} Figotin A. and Welters A., \textsl{Lagrangian
Framework for Systems Composed of High-Loss and Lossless Components},
Jour. Math. Phys., \textbf{55}, 062902 (2014).

\bibitem[GantM]{GantM} Gantmacher F., \textsl{Lectures in Analytical
Mechanics}, Mir, 1975.

\bibitem[GoLaRo]{GoLaRo} Gohberg I., Lancaster P., L. and Rodman
L., \textsl{Matrix Polynomials}, SIAM, 2009.

\bibitem[GoLaRo2]{GoLaRo2} Gohberg I., Lancaster P. and Rodman L.,
\textsl{Invariant Subspaces of Matrices with Applications}, SIAM,
2006.

\bibitem[GoLaRo3]{GoLaRo3} Gohberg I., Lancaster P. and Rodman L,
\textsl{Indefinite Linear Algebra and Applications}, Birkhauser, 2005.

\bibitem[Hale]{Hale} Hale J., \textsl{Ordinary Differential Equations},
2nd ed., Krieger Publishing Co., 1980.

\bibitem[HHWGECK]{HHWGECK} Hodaei H., Hassan A., Wittek S., Garcia-Gracia
H., El-Ganainy R., Christodoulides D and M. Khajavikhan, \textsl{Enhanced
sensitivity at higher-order exceptional points}, Nature, \textbf{548},
187 (2017). 

\bibitem[HMHCK]{HMHCK} Hodaei H., Miri M., Heinrich M., Christodoulides
D. and Khajavikhan M., \textsl{Parity-time\textendash symmetric microring
lasers}, Science, \textbf{346}, 975 (2014).

\bibitem[HorJohn]{HorJohn} Horn R. and Johnson C., \textsl{Matrix
Analysis}, 2nd ed., Cambridge University Press, 2013.

\bibitem[Iza]{Iza} Izadian A., \textsl{Fundamentals of Modern Electric
Circuit Analysis and Filter Synthesis. A Transfer Function Approach},
Springer S 2019.

\bibitem[Kato]{Kato} Kato T., \textsl{Perturbation theory for linear
operators}, Springer 1995.

\bibitem[KNAC]{KNAC} Kazemi H., Nada M., Mealy T., Abdelshafy A.
and Capolino F., \textsl{Exceptional Points of Degeneracy Induced
by Linear Time-Periodic Variation}, Phys. Rev. Applied, \textbf{11},
014007 (2019).

\bibitem[LanLifM]{LanLifM} Landau L. and Lifshitz E., \textsl{Mechanics},
3rd ed., Elsevier, 1976.

\bibitem[LanTsi]{LanTsi} Lancaster P. and Tismenetsky M., \textsl{The
Theory of Matrices}, 2nd ed., Academic Press, 1985.

\bibitem[LauMey]{LauMey} Laub A. and Meyer K., \textsl{Canonical
Forms for Symplectic and Hamiltonian Matrices}, Celestial Mechanics,
\textbf{9}, 213-238, (1974)

\bibitem[Masc]{Masc} Maschke B. et. al., \textsl{An intrinsic Hamiltonian
formulation of the dynamics of LC circuits}, IEEE Trans., \textbf{42},
No.2, 73-82, (1995).

\bibitem[MLSPL]{MLSPL} McCall S., Levi A., Slusher R., Pearton S.
and Logan R., \textsl{Whispering-gallery mode microdisk lasers}, Appl.
Phys. Lett., \textbf{60}, 289 (1992).

\bibitem[Mey]{Mey} Meyer K. et. al., \textsl{Introduction to Hamiltonian
Dynamical Systems and the N-Body Problem}, Springer, 2009.

\bibitem[MeyCD]{MeyCD} Meyer C., \textsl{Matrix analysis and applied
linear algebra,} SIAM, 2010.

\bibitem[Miln]{Miln} Milne-Thomson L., \textsl{Theoretical Aerodynamics},
Dover, 1958.

\bibitem[OGC]{OGC} Othman M., Galdi V. and Capolino F., \textsl{Exceptional
points of degeneracy and PT symmetry in photonic coupled chains of
scatterers}, Phys. Rev. B, \textbf{95}, 104305 (2017).

\bibitem[OTC]{OTC} Othman M., Tamma V., and Capolino F., Theory and
new amplification regime in periodic multimodal slow wave structures
with degeneracy interacting with an electron beam, IEEE Trans. Plasma
Sci., \textbf{44}, 594 (2016).

\bibitem[OthCap]{OthCap} Othman M. and F. Capolino F., in 2017 IEEE
Int. Symp. Antennas Propag. Usn. Natl. Radio Sci. Meet. (2017), pp.
57\textendash 58.

\bibitem[OVFC]{OVFC} Othman M, Veysi M., A. Figotin A. and Capolino
F., \textsl{Low starting electron beam current in degenerate band
edge oscillators}, IEEE Trans. Plasma Sci., \textbf{44}, 918 (2016).

\bibitem[OVFC1]{OVFC1} Othman M., Veysi M., Figotin A. and Capolino
F., \textsl{Giant amplification in degenerate band edge slow-wave
structures interacting with an electron beam}, Phys. Plasmas, \textbf{23},
033112 (2016).

\bibitem[PeLiXu]{PeLiXu} Peng C., Li Z., and Xu A., \textsl{Rotation
sensing based on a slow-light resonating structure with high group
dispersion}, Appl. Opt., \textbf{46}, 4125 (2007).

\bibitem[Rich]{Rich} Richards P., \textsl{Manual of mathematical
physics}, Pergamon Press, 1959.

\bibitem[RKEC]{RKEC} Ramezani H., Kottos T., El-Ganainy R. and Christodoulides
D., \textsl{Unidirectional nonlinear PT-symmetric optical structures},
Phys. Rev. A, \textbf{82}, 043803 (2010).

\bibitem[ReSi1]{ReSi1} M. Reed and B. Simon, \textsl{Methods of modern
mathematical physics, Vol. 1: Functional analysis}, Second edition,
Academic Press, Inc., New York, 1980.

\bibitem[SesRee]{SesRee} Seshu S. and Reed M., \textsl{Linear Graphs
and Electrical Networks}, Addison-Wesley, 1961.

\bibitem[SesBab]{SesBab} Seshu S. and Balabanian N.,\textsl{ Linear
Networks Analysis}, John Wiley \& Sons, 1964.

\bibitem[Scheck]{Scheck} Scheck F., \textsl{Mechanics - From Newton's
Laws to Deterministic Chaos}, 6th ed., Springer, 2018.

\bibitem[SLZEK]{SLZEK} Schindler J., Li A., Zheng M., Ellis F. and
Kottos T., \textsl{Experimental study of active LRC circuits with
PT symmetries}, Phys. Rev. A, \textbf{84}, 040101 (2011).

\bibitem[SteHeSch]{SteHeSch} Stehmann T., Heiss W. and Scholtz F.,
\textsl{Observation of exceptional points in electronic circuits},
J. Phys. Math. Gen. 37, 7813 (2004).

\bibitem[vdWae1]{vdWae1} van der Waerden B., \textsl{Algebra}, vol.
1, Springer, 2003.

\bibitem[VPFC]{VOFC} Veysi M., Othman M., Figotin A. and Capolino
F., \textsl{Degenerate band edge laser}, Phys. Rev. B, \textbf{97},
195107 (2018).

\bibitem[Wie]{Wie} Wiersig J., \textsl{Enhancing the Sensitivity
of Frequency and Energy Splitting Detection by Using Exceptional Points
- Application to Microcavity Sensors for Single-Particle Detection},
Phys. Rev. Lett., \textbf{112}, 203901 (2014).

\bibitem[Wie1]{Wie1} Wiersig J., \textsl{Sensors operating at exceptional
points: General theory}, Phys. Rev. A, \textbf{93}, 033809 (2016).

\bibitem[Witt]{Witt} Witten E., \textsl{Supersymmetry and Morse theory},
Jour. Diff. Geom., \textbf{17}, 661-692 (1982).
\end{thebibliography}
\end{document}